\documentclass[11pt, letterpaper]{article}

\usepackage[final]{microtype}
\usepackage{color}
\usepackage[ruled]{algorithm2e}
\usepackage{varwidth}
\usepackage{capt-of}
\usepackage{boldline}
\usepackage{epsfig}
\usepackage{pgf}
\usepackage{latexsym}
\usepackage{url}
\usepackage{float} % jackie commented out
\usepackage[hypertexnames=false,hyperfootnotes=false]{hyperref}
\usepackage{texnansi}
\usepackage{color}
\usepackage{tikz}
\usepackage{afterpage}
\usepackage{enumerate}
\usepackage[normalem]{ulem}
\usepackage{lmodern}
\usepackage{rotating}
\usepackage[ruled]{algorithm2e}
\usepackage[noend]{algpseudocode}
\usepackage{graphicx}
\usepackage{subcaption}
\usepackage{bbm}
\usepackage{layouts}
\usepackage{booktabs}

\usepackage{sectsty}
\usepackage{ifthen}
\usepackage[leqno]{amsmath}
\usepackage{amsthm}
\usepackage{amssymb}
\usepackage{amsfonts}
\usepackage{mathtools}
\usepackage[margin=10pt,font=small,labelfont=bf]{caption}

\usepackage{dsfont}
\usepackage{bm}
\usepackage[capitalise]{cleveref}
\usepackage{xparse} % to use myproof environment

\usepackage{natbib}
 \bibpunct[, ]{(}{)}{,}{a}{}{,}%

\makeatletter
\g@addto@macro{\UrlBreaks}{\UrlOrds}
\makeatother

\theoremstyle{plain}
\newtheorem{theorem}{Theorem}[section]%don't mess with this; it is set up correctly
\newtheorem{lemma}[theorem]{Lemma}
\newtheorem{claim}[theorem]{Claim}
\newtheorem{corollary}[theorem]{Corollary}
\newtheorem{proposition}[theorem]{Proposition}

\theoremstyle{definition}
\newtheorem{definition}[theorem]{Definition}
\newtheorem{property}[theorem]{Property}

\theoremstyle{remark}

\newtheorem{assumption}[theorem]{Assumption}

\newcommand{\bI}{\mathds{1}}

\newcommand{\bR}{\mathbb{R}}
\newcommand{\bE}{\mathbb{E}}

\newcommand{\ii}{\textit}

\newcommand{\cI}{\mathcal{J}}
\newcommand{\cJ}{\mathcal{J}}

\newcommand{\Exp}{\text{Exp}}
\newcommand{\Bern}{\text{Bern}}

\newcommand{\relerror}{\mathrm{RelError}}

\newcommand{\tone}{t_{\rho_1}}
\newcommand{\ttwo}{t_{\rho_2}}
\newcommand{\tC}{\tilde{C}}
\newcommand{\tS}{\tilde{S}}

\newcommand{\ti}{\tilde{i}}

\newcommand{\partN}{\frac{\partial}{\partial N}}

\newcommand{\fourth}{\frac{1}{4}}
\newcommand{\hb}{\hat{\beta}}
\newcommand{\hg}{\hat{\gamma}}
\newcommand{\pop}{N}

\newcommand{\basspeak}{t_{k^*}}
\newcommand{\basscr}{t_{k^{\rm CR}}}
\newcommand{\sircr}{t^d_{\rm CR}}
\newcommand{\sirpeak}{t^d_{*}}
\newcommand{\sirrtpeak}{t_2^d}

\newcommand{\edit}[1]{#1}

\DeclareMathOperator*{\argmin}{\mathrm{argmin}}

\usetikzlibrary{arrows,patterns,plotmarks}
\tikzstyle{every picture} += [>=stealth]
\allsectionsfont{\sffamily}

\makeatletter
\def\@seccntformat#1{\csname the#1\endcsname.\quad}
\makeatother

\floatstyle{ruled}
\provideboolean{fastcompile}
\newcommand{\hidefastcompile}[1]{\ifthenelse{\boolean{fastcompile}}{}{#1}}
\usepackage[margin=1in]{geometry}
\usepackage{setspace}
\newtheorem*{le:ell}{Lemma~\ref{le:ell}}
\newtheorem*{le:feas}{Lemma~\ref{le:feas}}
\newtheorem*{le:sample_complexity}{Lemma~\ref{le:sample_complexity}}

\NewDocumentEnvironment{myproof}{o}
  {\IfNoValueTF{#1}{\paragraph{{\normalfont \textit{Proof.}}}} {\paragraph{{\normalfont \textit{#1.}}}} }
  {\hfill$\square$}

\usepackage{setspace}

\title{\textsf{\textbf{The Limits to Learning a Diffusion Model}}}
\author{%
 Jackie Baek\thanks{Stern School of Business, NYU.} \\
  \texttt{\small baek@stern.nyu.edu} \\
  \and
  Vivek F. Farias\thanks{Sloan School of Management, MIT.} \\
  \texttt{\small vivekf@mit.edu} \\
  \and
  Andreea Georgescu\thanks{Operations Research Center, MIT.} \\
  \texttt{\small andreeag@mit.edu} \\
  \and
  Retsef Levi\footnotemark[2] \\
  \texttt{\small retsef@mit.edu} \\
  \and
  Tianyi Peng\thanks{Department of Aeronautics and Astronautics, MIT.} \\
  \texttt{\small tianyi@mit.edu} \\
  \and
  Deeksha Sinha\footnotemark[3] \\
  \texttt{\small deeksha.sinha7@gmail.edu} \\
  \and
  Joshua Wilde\footnotemark[3] \\
  \texttt{\small jtwilde@mit.edu} \\
  \and
  Andrew Zheng\footnotemark[3] \\
  \texttt{\small atz@mit.edu} \\
}
\date{\vspace{-12mm}}

\begin{document}
\maketitle

% \doublespacing
\begin{abstract}
\setstretch{1}
\noindent 
\vspace{-2mm}
%!TEX root=../SIR-model-MS.tex

% Original abstract (2022): 
This paper provides the first sample complexity lower bounds for the estimation of simple diffusion models, including the Bass model (used in modeling consumer adoption) and the SIR model (used in modeling epidemics). 
We show that one cannot hope to learn such models until quite late in the diffusion. Specifically, we show that the time required to collect a number of observations that exceeds our sample complexity lower bounds is large. 
For Bass models with low innovation rates, our results imply that one cannot hope to predict the eventual number of adopting customers until one is at least two-thirds of the way to the time at which the rate of new adopters is at its peak. In a similar vein, our results imply that in the case of an SIR model, one cannot hope to predict the eventual number of infections until one is approximately two-thirds of the way to the time at which the infection rate has peaked. 
This lower bound in estimation further translates into a lower bound in regret for decision-making in epidemic interventions.
Our results formalize the challenge of accurate forecasting and highlight the importance of incorporating additional data sources.
To this end, we analyze the benefit of a seroprevalence study in an epidemic, 
where we characterize the size of the study needed to improve SIR model estimation.
Extensive empirical analyses on product adoption and epidemic data support our theoretical findings.
% These limits are borne out in both product adoption data (Amazon), as well as epidemic data (COVID-19).

\vskip 5pt
% \noindent {\it Keywords:} diffusion; epidemics; SIR model; Bass model; COVID-19 modelling; forecasting
\end{abstract}

%\maketitle
%%%%%%%%%%%%%%%%%%%%%%%%%%%%%%%%%%%%%%%%%%%%%%%%%%%%%%%%%%%%%%%%%%%%%%
\setstretch{1.5}
% \doublespacing

% Actual content here:

%!TEX root=../SIR-model-MS.tex
\section{Introduction} \label{sec:intro}

Diffusion models are simple reduced form models (typically described by a system of differential equations) that seek to explain the diffusion of an epidemic in a network. The Susceptible-Infected-Recovered (SIR) model is a classic example, proposed nearly a century ago \citep{kermack1927contribution}. The SIR model remains a cornerstone for the forecasting of epidemics. The so-called Bass model \citep{bass1969new}, proposed over fifty years ago, is similarly another example that remains a basic building block in forecasting consumer adoption of new products and services. The durability of these models arises from the fact that they have shown an excellent fit to data, in numerous studies spanning both the epidemiology and marketing literatures. Somewhat paradoxically, using these same models as reliable forecasting tools presents a challenge. 

%In numerous retrospective studies the model has been found to fit the trajectory of epidemics well, while simultaneously providing a meaningful level of interpretability. As such, the plurality of models used for forecasting efforts related to the COVID-19 epidemic are either SIR models or close cousins thereof. Surprisingly, the experience with these forecasts has illustrated that predicting the cumulative number of cases (or peak number of cases) in an epidemic, early in its course, is a challenging task.
%%Surprisingly, the forecasts generated using this model have had questionable predictive power. For instance, the widely cited model produced by \cite{ihme} has produced forecasts for total number of deaths that have varied over time by nearly an order of magnitude despite being at a coarse granularity (a state). \textcolor{blue} {Josh: IHME is not strictly an SIR model. }
%It is perhaps no surprise then that more sophisticated models that generalize the SIR model and are typically harder to estimate have found relatively little use in forecasting efforts related to the COVID-19 epidemic.

While we are ultimately motivated by the problem of forecasting a diffusion model, this paper asks a more basic question that is surprisingly unanswered: {\it What are the limits to learning a diffusion model?} We answer this question by characterizing sample complexity lower bounds for a class of stochastic diffusion models that encompass both the Bass model and the SIR model. 
We show that the time to collect a number of observations that exceeds these lower bounds is too large to allow for accurate forecasts early in the process. 
In the context of the Bass model, our results imply that when adoption is driven by imitation, one cannot hope to predict the eventual number of adopting customers until one is at least two-thirds of the way to the time at which the rate of new adopters is at its peak. In a similar vein, our results imply that in the case of an SIR model, one cannot hope to predict the eventual number of infections until one is approximately two-thirds of the way to the time at which the infection rate has peaked. Our analysis is conceptually simple and relies on the Cramer-Rao bound. The core technical difficulty in our analysis rests in characterizing the Fisher information in the observations available due to the fact that they have a non-trivial correlation structure.

\edit{
Specifically, the SIR and Bass models are each characterized by two parameters that determine the rate of diffusion, as well as a population parameter, denoted by $N$.
Our analysis finds that the bottleneck in learning these models is in estimating the parameter $N$.
In the Bass model, $N$ represents the eventual total number of adopters, while in the SIR model, $N$ represents the `effective population', which is unknown in scenarios where an unknown fraction of infections are reported \citep{li2020substantial,lau2021evaluating,pullano2021underdetection} or an unknown fraction of the population is susceptible. An accurate estimation of $N$ is essential, as several important statistics such as the total number of eventual infections scale with $N$ \citep{weiss2013sir}.

Our main result shows that an accurate estimation of $N$ requires at least $\Omega(N^{2/3})$ observations of the stochastic diffusion model.
The point at which $\Omega(N^{2/3})$ observations are collected corresponds to two-thirds of the time to peak for the SIR model, as well as the Bass model with low innovation rates.
We show that the other parameters of the diffusion models, including those related to the `rate of imitation' (in the Bass model) or the `reproduction number' (in the SIR model) are relatively easy to learn. 

We then establish a lower bound on the regret of an intervention decision problem.
Specifically, we formalize generic decision problem where the decision is whether to impose a `drastic intervention' that is associated with a cost, but will immediately stop all further infections.
The difficulty in the estimation of $N$ translates to the difficulty of this
decision problem --- we show that any policy that makes this decision based on the observations of the diffusion model will incur a regret of $\Omega(N^{2/3})$.

Our results highlight the challenges of using infection trajectories for accurate forecasting and underscore the need to incorporate additional data sources. 
In the context of an epidemic, an example of such a data source can come from 
a seroprevalence study \citep{havers2020seroprevalence,bendavid2021covid}.
% surveillance testing, which have been conducted to provide an estimate of seroprevalence \citep{havers2020seroprevalence,bendavid2021covid}.
We investigate the benefit of such a data source by characterizing the necessary size of the seroprevalence test to meaningfully improve the estimation accuracy of $N$.
Our results show that one can improve upon the $\Omega(N^{2/3})$ lower bound via a sublinear size of the test:
% survellience testing campaign:
after $\Theta(N^b)$ samples of the diffusion process, a campaign of size $\omega(N^{1-b})$ will lead to an accurate estimation of $N$.

We conduct extensive simulations to corroborate the theoretical results.
We demonstrate that maximum likelihood estimation (MLE) of diffusion models on product adoption datasets (for products on Amazon.com), and epidemic data (COVID-19) illustrate precisely the behavior predicted by our theory. 
We show that our results are robust to versions of the SIR model that capture heterogeneously mixing subpopulations.
Lastly, we describe a heuristic method that was deployed for a real-world COVID-19 forecasting tool for US counties, that used a complex variant of the SIR model that accounted for non-stationarities and rich county-level covariates.
We show that, even in this complex variant of the SIR model, the estimation of $N$ remains a first-order issue.
We develop a heuristic to construct a biased estimator of $N$ that leverages the plurality of counties, which substantially reduces the forecasting error compared to a naive MLE estimator.
% We describe a heuristic to produce a biased estimator of $N$, 
% which leverages the heterogeneity in the timing of infections across the many counties. We illustrate that this biasing heuristic substantially reduces forecasting error compared to the naive MLE estimator.

}

\subsection{Related Literature}
% \noindent{\textbf{Related Literature:}}
Diffusion models find broad application in at least two key domains: epidemiology and marketing science. While there is surprisingly little literature that cuts across the two application domains, the dominant themes are quite similar. 

\subsubsection{SIR Model.}
The SIR model \citep{kermack1927contribution} is perhaps the best known and most widely analyzed and used diffusion model in the epidemiology literature. For instance, the plurality of COVID-19 modeling efforts are founded on SIR-type models (eg. \cite{calafiore2020modified, gaeta2020simple, giordano2020modelling, binti2020coronatracker, kucharski2020early, wu2020nowcasting, anastassopoulou2020data, biswas2020covid, chikina2020modeling, massonnaud2020covid, goel2020mobility}). It is common to consider generalizations to the SIR model that add additional states or `compartments' (\cite{giordano2020modelling} is a nice recent example); not surprisingly, learning gets harder as the number of states increases \citep{roosa2019assessing}.

% \textbf{Estimation in SIR Models:} 
The identifiability of the stochastic SIR model \citep{bartlett1949some, darling2008differential} is not well understood in the literature. In fact, even identification of the deterministic model is a non-trivial matter \citep{evans2005structural}.  
Specifically, calibrating a vanilla SIR model to data requires learning the so-called infectious period and basic reproduction rate. Both these parameters are relatively easy to calibrate with limited data; this is supported both by the present paper, but also commonly observed empirically; see for instance \cite{roosa2019assessing}. 
\edit{
For COVID-19, several empirical works have demonstrated the limitations of using SIR-based models for forecasting \citep{moein2021inefficiency,castro2020turning,bertozzi2020challenges}. 
These works cite several possible reasons for these limitations, from behavioral changes, variations in air pollution, to mixing heterogeneities \citep{moein2021inefficiency}.
Our work raises a fundamental estimation issue that arises even when all of the model assumptions are satisfied, the difficulty in estimating the parameter $N$.

\textbf{Unknown $N$.} 
% In addition to these parameters, however, one needs to measure both the initial number of infected individuals and the size of the susceptible population. 
While assuming that $N$ is `unknown' is not the \textit{default} assumption in the SIR model, because under-reporting is a prevalent issue, existing works incorporate this issue in slightly different ways. For example, several papers explicitly split the `infection' compartment in the SIR model into two, which  represent observed and unobserved infections, and a new parameter is introduced which denotes the probability of an infection being reported \citep{giordano2020modelling,gaeta2020simple,ivorra2020mathematical,mit_delphi-covid}. 
\cite{calafiore2020modified} does not explicitly create new compartments, but simply writes  $I(t) = \alpha \tilde{I}(t)$ for $\alpha > 1$, where $I$ and $\tilde{I}$ represent the true and observed infections respectively.
% In this case, $\alpha$ is a new parameter that must be estimated.
These approaches are mathematically equivalent to assuming that $N$ is unknown.
An alternative approach is to fit a model to deaths \citep{ihme2021modeling}, 
which suffers less from under-reporting bias.
From deaths, one can recover infections using the so-called infection-fatality ratio (IFR), the fraction of cases that lead to fatalities.
This approach relies on an accurate estimation of the IFR.
Overall, while `unknown $N$' is not the default assumption in the SIR model,
it represents a prevalent issue that many of the existing epidemic forecasting works have incorporated in slightly different ways.
}

% Estimating the number of infected individuals poses a challenge in the presence of limited testing and asymptomatic carriers. Indeed, epidemiological models for COVID-19 typically assume that measured infections are some fraction of true infections; eg. \cite{calafiore2020modified, giordano2020modelling}. 
% This challenge is closely related to that of measuring the true fraction of cases that lead to fatalities (or the so-called Infection Fatality Rate) \citep{basu2020estimating}. 
% Our main theorem shows that having to learn the true initial prevalence of the infection presents a fundamental difficulty to learning SIR models with limited data; this is complemented by heretofore unexplained empirical work \citep{chowell2017fitting, capaldi2012parameter}. 

\subsubsection{Bass Model.}
The Bass model \citep{bass1969new} remains the best known and most widely analyzed diffusion model in the marketing science literature. The model has found applications in a staggering variety of industries over the past fifty years. Surveys such as \cite{bass2004comments, mahajan2000new, hauser2006research} provide a sense of this breadth, showing that the model and its generalizations have found application in tasks ranging from forecasting the adoption of technologies, brands and products to describing information cascades on services such as Twitter \citep{bakshy2011everyone}. Just as in the case of the SIR model, a number of generalizations of the Bass model have been proposed over the years, including \cite{peterson1978multi, bass1994bass, van2007new}. Similar random processes related to Bass model have also been studied in mathematical immunology \citep{hawkins2007model,duffy2012activation}. 

% \textbf{Estimation in Bass Models:} 
The Bass model has traditionally been estimated using a variety of weighted least squares estimators; \cite{srinivasan1986nonlinear, jain1990effect} are popularly used examples. The key parameters that must be estimated here are the so-called coefficient of imitation (the analogue of the reproduction number in the SIR model) and the coefficient of innovation (which does not have an analogue in the SIR model). In addition one must estimate the size of the eventual population that will adopt (arguably one of the key quantities one would care to forecast). It has been empirically observed that existing estimation approaches are `unstable' in the sense that estimates of the size of the population that adopts can vary dramatically even half-way through the diffusion model \citep{van1997bias, hardie1998empirical} among other undesirable features. This has been viewed as a limitation of the estimators employed, and has led to corrections to the estimators that purport to address some of these issues \citep{boswijk2005econometrics}. In contrast, our results imply that this behavior is fundamental; as one example we show that no unbiased estimator of the Bass model can hope to learn the population size until at least two-thirds of the way through the diffusion model.

%!TEX root=../SIR-model-MS.tex

\section{Model}

We first define a general deterministic diffusion model using a system of ODEs.
Our paper focuses on two parameter regimes of this model, which represent the Bass model (\cref{sec:model_bass}) and the SIR model (\cref{sec:model_sir}).
We then describe a stochastic variant of the diffusion model in \cref{sec:stochastic_model};
our main result in \cref{sec:limits_to_learning} describes the limits to learning the parameters of this stochastic model.

\subsection{Deterministic Diffusion Model}
We define a general diffusion model with three `compartments' over an `effective' population of size $N$ \citep{meyn2012markov}. 
Let $s(t), i(t)$ and $r(t)$ be the size of susceptible, infected, and recovered populations respectively, as observed at time $t$, where $s(t) + i(t) + r(t) = N$ for all $t \geq 0$.
The model is defined by the following system of ODEs, specified by the tuple of parameters $(N, \beta, \gamma, p)$:
\begin{align}
\label{eq:sirfluid}
\frac{ds}{dt} &=  - \beta \frac{s}{N} i - ps,   &
\frac{di}{dt} &= \beta \frac{s}{N} i - \gamma i + ps,  &
\frac{dr}{dt} &= \gamma i.
\end{align}
We assume that all parameters are non-negative, and that $\beta > \gamma$. The parameters here that we may need to estimate include $\beta, \gamma, p$ and $N$. 
%We treat all parameters as unknown, including $N$.
%Even though the total population of a region of interest is often known, this paper studies regimes in which the {\em effective} population, $N$, is unknown. 
%We discuss this point further in the next two sections.

% \subsection{$\gamma=0$: Bass Model}  \label{sec:model_bass}
\subsection{Bass Model ($\gamma=0$)}  \label{sec:model_bass}
% \subsection{Parameter Regimes}  \label{sec:model_bass}
% \subsubsection{Bass Model}  
The Bass model is the special case of the diffusion model above where $\gamma = 0$ and as already discussed has been variously used to describe the diffusion of a new product, technology, or even information in a population. 
% Specifically, it is defined by the following parameter regime:
% % We study an asymptotic regime of a sequence of systems of increasing size. 
% % Specifically, we consider 
% \begin{assumption} \label{assump:params_bass}
% Let $\gamma = 0$, $p \geq 0$, $\beta, N > 0$, and $r(0) = R_0 = 0$.
% % We study a sequence of systems in which $N \rightarrow \infty$.
% % We assume $\beta$ is fixed and $p = O\left(\frac{1}{N^{\alpha}}\right)$ for some fixed $\alpha > 0$.
% \end{assumption}
$i$ and $s$ represent the number of people who have and have not adopted the product respectively by time $t$.
Since $\gamma = 0$, there is effectively no $r$ compartment.
The term $\beta \frac{s}{N}i$ represents the instantaneous growth rate in adoption contributed by individuals `imitating' existing adopters, while $ps$ represents the instantaneous growth rate in adoption contributed by 
`innovators' who adopt the product without the influence of existing adopters.
The parameter $\beta$ is often called the {\em coefficient of imitation}\footnote{The marketing science literature will frequently use the letter $q$ in place of $\beta$.} , while $p$ is called the {\em coefficient of innovation}.

In the Bass model, the eventual number of adopters i.e., $\lim_{t \rightarrow \infty} i(t) = N$, is often an important quantity of interest. 
As such, $N$ is a key, unknown parameter to estimate in this setting.
We define an additional parameter $a \triangleq pN$.
Since $s \approx N$ initially, $a$ represents the growth rate of innovators near the beginning of the process.
%It is often easier to work with the parameter $a$ rather than $p$.
%It turns out that $a$ is an easy quantity to estimate (\cref{thm:bass-estimation-p-a}), while it is difficult to separate out $p$ and $N$ from $a$ until much later in the process (\cref{co:cr}).

\subsection{SIR Model ($p=0$)} \label{sec:model_sir}
The SIR model is the simplest compartmental model in epidemiology that models how a disease spreads amongst a population, and it can be described by the diffusion model in the case that $p= 0$. 
% Specifically, it is defined by the following parameter regime:
% \begin{assumption} \label{assump:params_sir}
% Let $p = 0$, $\beta > \gamma > 0$, and $N > 0$.
% % Let $N \in (0, P]$ for a known $P$.
% % Assume $N \rightarrow \infty$ while all other parameters are fixed.
% \end{assumption}
% $s$, $i$, and $r$ correspond to the size of the susceptible, infected, and recovered populations respectively.
The parameter $\gamma$ specifies the rate of recovery; $1/\gamma$ is frequently referred to as the {\em infectious period}. The parameter $\beta > 0$ quantifies the rate of transmission; $\beta/\gamma \triangleq R_0$ is also referred to as the {\em basic reproduction number}.
% We assume that $N \in (0, P]$, where $P$ is the known total population of the region of interest.

In using the SIR model to model an epidemic where only a fraction of all infections are observed (due to, for example, asymptomatic cases and limited testing) the $N$ parameter is effectively the actual population of the region being modeled multiplied by the fraction of observed infections. If the fraction of observed infections is unknown (which it typically is), then $N$ is effectively unknown. Specifically, the following proposition\footnote{An analogous result for a discrete-time model was shown in \cite{calafiore2020modified}.} shows that the quantities corresponding to observing a constant fraction of an SIR model also constitutes an SIR model with the same parameters $\beta$ and $\gamma$.

\begin{proposition}
\label{prop:constant_fraction}
Let $\{(s'(t),i'(t),r'(t)): t\geq 0\}$ be a solution to \eqref{eq:sirfluid} for parameters $N = N', \beta = \beta', \gamma = \gamma'$ and initial conditions $i(0) = i'(0), s(0) = s'(0)$. Then, for any $\eta > 0$, $\{(\eta s'(t), \eta i'(t), \eta r'(t)): t\geq 0\}$ is a solution to \eqref{eq:sirfluid} for parameters $N = \eta N', \beta = \beta', \gamma = \gamma'$ and $i(0) = \eta i'(0), s(0) = \eta s'(0)$.
\end{proposition}

In words, suppose a disease spreads according to an SIR model amongst the entire population of (known) size $N'$. Suppose we only observe a constant fraction from this process, where this fraction $\eta$ is unknown.
The proposition above states that the observed process is also an SIR model with the same parameters $\beta$ and $\gamma$, and an effectively unknown population $N=\eta N'$.
%
%In the epidemic modeling literature, a typical exposition of this model sets $N$ to be equal to the total population of a region, which is often a known quantity. This corresponds to the setting where the infection is entirely observed. More commonly, however, $N$ is effectively unknown in the setting where an unknown fraction of all infections are observed (due to, for example, asymptomatic cases and limited testing).
%% This setting is modelled by an unknown $N$, and $N/P$ represents the fraction of all infections that are observed.
%Specifically, the following proposition shows that the quantities corresponding to observing a constant fraction of an SIR model also constitutes an SIR model with the same parameters $\beta$ and $\gamma$:
%
%\begin{proposition}
%\label{prop:constant_fraction}
%Let $\{(s'(t),i'(t),r'(t)): t\geq 0\}$ be a solution to \eqref{eq:sirfluid} for parameters $N = N', \beta = \beta', \gamma = \gamma'$ and initial conditions $i(0) = i'(0), s(0) = s'(0)$. Then, for any $\eta > 0$, $\{(\eta s'(t), \eta i'(t), \eta r'(t)): t\geq 0\}$ is a solution to \eqref{eq:sirfluid} for parameters $N = \eta N', \beta = \beta', \gamma = \gamma'$ and $i(0) = \eta i'(0), s(0) = \eta s'(0)$.
%\end{proposition}
%
%That is, suppose a disease spreads according to an SIR model amongst the entire population of (known) size $N'$. Suppose we only observe a constant fraction from this process, where this fraction $\eta$ is unknown.
%This proposition states that the observed process is also an SIR model with the same parameters $\beta$ and $\gamma$, and the parameter $N=\eta N'$.

It is known that both cumulative and peak infections scale with $N$ \citep{weiss2013sir}.
As these are often the key quantities of interest, 
estimating $N$ accurately is a critical task.
% it is reasonable to require that $N$ be estimated accurately.

\subsection{Stochastic Diffusion Model} \label{sec:stochastic_model}
% \subsection{SIR Model} \label{sec:model_sir}
In the deterministic diffusion model, all parameters are identifiable if $i(t)$ is observable over an {\em infinitesimally small} period of time in either of the two regimes.
Specifically:
\begin{proposition}
\label{prop:identity_fluid}
Suppose either $p = 0$ or $\gamma = 0$. Let $i(t)$ be observed over some open set in $\mathbb{R}_+$. Then the parameters $(N, \beta, \gamma, p)$ are identifiable.
\end{proposition}
% A self-contained proof follows immediately from the identity theorem; a more involved argument based on identification results for non-linear systems can be found in \cite{evans2005structural}.

% \textbf{Stochastic Diffusion Process: }
Noise --- an essential ingredient of any real-world model --- dramatically alters this story. We next describe a natural continuous-time Markov chain variant of the deterministic diffusion model, proposed at least as early as \cite{bartlett1949some}. Specifically, the stochastic diffusion model, $\{(S(t),I(t),R(t)): t \geq 0\}$, is a multivariate counting process, with right-continuous-with-left-limits (RCLL) paths,  determined by the parameters $(N, \beta, \gamma, p)$. The jumps in this process occur at the rate in \eqref{eq:deftk}, and correspond either to a new observed infection or adopter (where $I(t)$ increments by one, and $S(t)$ decrements by one) or to a new observed recovery (where $I(t)$ decrements by one, and $R(t)$ increments by one). Let $C(t) = I(t) + R(t)$ denote the cumulative number of infections or adoptions observed up to time $t$. Denote by $t_k$ the time of the $k$th jump, and let $T_k$ be the time between the $(k-1)$st and $k$th jumps. Finally, let $I_k \triangleq I(t_k)$, and similarly define $R_k, S_k$ and $C_k$. The stochastic diffusion model is then completely specified by:
\begin{align}
% C_k - C_{k-1} &\sim \Bern\left\{\frac{\beta S_{k-1}}{\beta S_{k-1} + N\gamma}\right\} \label{eq:defck}\\
% T_k &\sim \Exp\left\{ \left(\frac{\beta S_{k-1}}{N} + \gamma\right)I_{k-1} + pS_{k-1} \right\}. \label{eq:deftk}
C_k - C_{k-1} &\sim \Bern\left\{\frac{S_{k-1} (\beta I_{k-1} + pN)}{S_{k-1} (\beta I_{k-1} + pN) + N \gamma I_{k-1}} \right\}, \label{eq:defck}\\
T_k &\sim \Exp\left\{ \frac{\beta S_{k-1}}{N} I_{k-1} + pS_{k-1} + \gamma I_{k-1} \right\}. \label{eq:deftk}
\end{align}
It is well known that solutions to the deterministic diffusion model~\eqref{eq:sirfluid} provide a good approximation to sample paths of the diffusion model (described by \eqref{eq:defck}, \eqref{eq:deftk}) in the so-called fluid regime; see \cite{wormald1995differential,darling2008differential}.

The next section analyzes the rate at which one may hope to learn the unknown parameters $(N, \beta, \gamma, p)$ as a function of $k$; our key result will illustrate that in large systems, $N$ is substantially harder to learn than $\beta$ or $\gamma$. In turn this will allow us to show that we cannot hope to learn the stochastic diffusion model described above until quite late in the diffusion.

%In fact, an approximation suggests that the time taken to learn this parameter accurately will be approximately two-thirds of the time required to hit the peak of the process.
% In Section~\ref{sec:experiments} we will consider a differentiable model inspired by the SIR process and show that a regularization strategy motivated by our learning analysis yields material performance gains.

% To have an accurate estimate of $\theta$, it is reasonable to require that the relative error goes to 

% For a parameter $\theta$, we say that an estimator $\hat{\theta}$ is ``accurate'' if 
% we say that we can ``learn'' $\theta$ if we have an estimator $\hat{\theta}$ such that
% \begin{align}
% \frac{\bE[(\hat{\theta} - \theta)^2]}{\theta^2} = o(1).
% \end{align}
% In other words, this requirement states that the root mean squared error of the estimator $\hat{\theta}$ is of order $o(\theta)$.

%!TEX root=../SIR-model-MS.tex
\section{Limits to Learning} \label{sec:limits_to_learning}

This section characterizes the rate at which one may hope to learn the parameters of the stochastic diffusion model, simply from observing the process.

\textbf{Observations: }Define the stopping time $\tau = \inf \{k: I_k = 0 \text{ or } I_k = N\}$; clearly $\tau$ is bounded. For clarity, when $k > \tau$, we define $C_k = C_{k-1}$, $I_k = I_{k-1}$, and $T_k = \infty$.
Note that $I_k$ and $R_k$ are deterministic given $C_k$, $I_0$, and $R_0$.
We define the $m$-th information set $O_m = (I_0, R_0, T_1, C_1, \dots, T_m, C_m)$ for all $m \geq 1$.

% For a parameter $\theta$, we evaluate the accuracy of an estimator $\hat{\theta}$ using the following metric, which we call \ii{relative error}:
\textbf{Evaluation Metric: }
For any parameter $\theta$, suppose $\hat{\theta}_m$ is an estimator based on the observations $O_m$.
We define the \ii{relative error} of $\hat{\theta}_m$ as:
\begin{align*}
% \relerror(\hat{\theta}_m, \theta) = \frac{\bE[(\hat{\theta}_m - \theta)^2]}{\theta^2}.
\relerror(\hat{\theta}_m, \theta) \triangleq \frac{(\hat{\theta}_m - \theta)^2}{\theta^2}.
\end{align*}
A relative error of 1 implies that the absolute error of the estimator is the same size as the true parameter.
Therefore, in order to estimate a parameter $\theta$, it is reasonable to require that the relative error be at most 1, and ideally shrinking to 0.
Our goal is to find the regime of $m$ relative to $N$ such that $\relerror(\hat{\theta}_m, \theta) = o(1)$.

Our main theorem lower bounds the relative error of any unbiased estimator of the parameter $N$.
% Our main theorem characterizies the Fisher information of $O_m$ relative to $N$ as $N$ grows large.
We first state the exact assumptions necessary for the two regimes:

\begin{assumption}[Bass Model] \label{assump:bass}
Assume $\gamma = 0$. Consider a sequence of systems of increasing size $N$, and $\beta$ and $a=pN$ are known constants. Assume $I_0 = 1, R_0 = 0$.
\end{assumption}

\begin{assumption}[SIR Model] \label{assump:sir}
Assume $p = 0$. Consider a sequence of systems of increasing size $N$, and $\beta$ and $\gamma$ are known constants. Assume $I_0$ is a sufficiently large constant and $R_0=o(N).$
%$I_0 \geq D$, where $D$ is a constant that can depend on $\beta$ and $\gamma$. Let $m = o(N)$, and $I_0, R_0 \leq m$.
%that may depends on $\beta$ and $\gamma$.
\end{assumption}

We now state our main result.
% \begin{corollary}[Constrained Cramer-Rao Bound]\label{co:cr}
\begin{theorem}\label{co:cr}
Under \cref{assump:bass} or \cref{assump:sir}, if $m=o(N)$ and $\hat N_m$ is any unbiased estimator of $N$ based on the observations $O_{m}$,
% Under \cref{assump:asymptotic}, suppose either $p=0$ or $\gamma=0$, and $\beta, \gamma$, and $pN$ are known.
% Assume $m = o(N)$ and $I_0, R_0 \leq m$.
% Let $\hat N$ be any unbiased estimator of $N$ based on the observations $O_{m}$.
% There exists a constant $D$ that only depends on $\beta$ and $\gamma$ such that if $I_0 \geq D$,
\begin{align}
% \frac{\variance(\hat N)}{N^2} = \Omega\left(\frac{N^2}{m^3}\right).
\bE[\relerror(\hat{N}_m, N)] = \Omega\left(\frac{N^2}{m^3}\right).
% \frac{\bE[(\hat N - N)^2]}{N^2} = \Omega\left(\frac{N^2}{m^3}\right).
\end{align}
\end{theorem}

\cref{co:cr} is the core result of this work. 
Observe that to have $\bE[\relerror(\hat{N}_m, N)] = o(1)$, we must have $m = \omega(N^{2/3})$.
That is, in order for the error of any unbiased estimator to be smaller than the value of $N$ itself,
% That is, in order for any estimator to have a `good' relative error, 
the number of adopters in a Bass model or the number of infected people in an SIR model needs to surpass $\sim N^{2/3}$ observations\footnote{Note that this implication is independent of initial conditions (i.e., $I_0,R_0$): a partial observation can only render the estimation harder. See Appendix~\ref{sec:generation-finite-sample} for a result that generalizes \cref{co:cr}.}.
% It implies that in order to estimate an unbiased $N$, from a sample-complexity perspective, the number of adopters in a Bass model or the number of infected people in an SIR model needs to surpass $\sim N^{2/3}$ (to have $\bE[\relerror(\hat{N}_m, N)] = o(1)$, we must have $m = \omega(N^{2/3})$ observations)\footnote{Note that this implication is independent of initial conditions (i.e., $I_0,R_0$): a partial observation can only render the estimation harder. See \cref{sec:generation-finite-sample} for a result that generalizes \cref{co:cr}.}. 
The magnitude of $N^{2/3}$ can be consequential in practice. For example, for $N=10M$, this corresponds to 45k infections. This no-go theorem provides a new  insight for understanding the difficulties of estimating diffusion processes in early stages: $N$ plays a key role of driving such difficulties in practical applications (e.g., see real-data experiments in \cref{sec:exp:benchmark_datasets}).
%
%We show that estimating $N$ is a key factor that drives the hardness of estimating a stochastic diffusion process in  \cref{sec:exp:benchmark_datasets} (including Covid and Amazon datasets), despite there may exist many other factors to impact the difficulties of estimating a stochastic diffusion process (e.g., model specification).}. 

The intuition of this no-go result can be best illustrated by \cref{fig:log_scale}, a plot of deterministic diffusion models with (largely) varying $N$ with other parameters fixed. This illustration shows that different diffusion processes share similar increasing curves for a significant amount of time before diverging. Although the differentiation of these processes in theory is easy due to their deterministic nature (see \cref{prop:identity_fluid}), incorporating noise renders this differentiation impossible. Our \cref{co:cr} then quantifies the exact hardness of such differentiation when noise is presented; further, it discovers a precise (yet unexpected) transition point in terms of sample-complexity: $N^{2/3}.$

\begin{figure}[h]
\vspace{-0.3cm}
\begin{center}
  \includegraphics[width=0.9\linewidth]{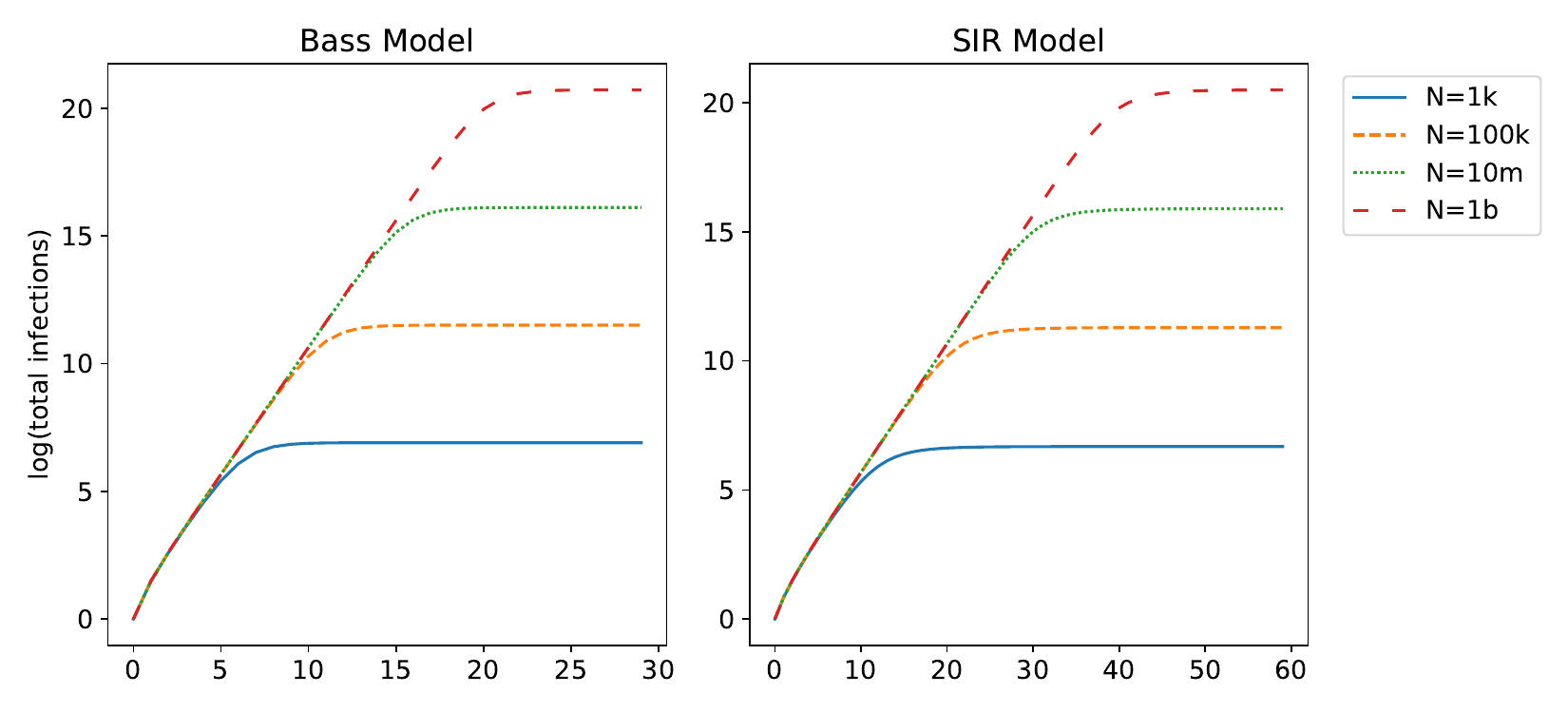}
  \caption{Log of cumulative infections for the deterministic Bass and SIR model with varying $N$. The left figure corresponds Bass models with $\beta=1, p=1/N$. The right figure corresponds to SIR models with $\beta=1, \gamma=0.5$. }
  \label{fig:log_scale}
\end{center}
% \vspace{-0.5cm}
\end{figure}

The general statement of \cref{co:cr} is a finite-sample result that holds for any initial conditions (see Appendix~\ref{sec:generation-finite-sample}), which is a direct consequence of applying the Cramer-Rao bound to the following theorem that characterizes the Fisher information of $O_m$ relative to $N$.
\begin{theorem} \label{thm:fi}
Under \cref{assump:bass} or \cref{assump:sir}, if $m=o(N)$, then the Fisher information of $O_{m}$ relative to $N$ is
% Under \cref{assump:asymptotic}, suppose either $p=0$ or $\gamma=0$, and $\beta, \gamma$, and $pN$ are known.
% Assume $m = o(N)$ and $I_0, R_0 \leq m$.
% There exists a constant $D$ that only depends on $\beta$ and $\gamma$ such that if $I_0 \geq D$, the Fisher information of $O_{m}$ relative to $N$ is
\begin{align}
\cI_{O_{m}}(N) = \Theta\left(\frac{m^3}{N^4}\right).
\end{align}
\end{theorem}
The proof of \cref{thm:fi} can be found in Section~\ref{sec:cr_proof}, which involves a non-trivial analysis of the Fisher information of a complex SIR/Bass stochastic process. It is notable that the result above provides a precise rate for the Fisher information as opposed to simply an upper bound. This further allows us to conclude that the relative error rate in Theorem~\ref{co:cr} is precisely the rate achieved by an efficient unbiased estimator for $N$.

The next section, \cref{sec:time_to_peak}, analyzes how long it takes to reach $N^{2/3}$ observations. We show that in many parameter regimes, the time it takes to reach $N^{2/3}$ observations is a constant portion (e.g., two thirds) of the time it takes to reach the peak infection rate of the process. In \cref{sec:estimating_other_params}, we analyze the relative error for the other parameters of the model, and we show that these other parameters are much easier to learn than $N$. In \cref{sec:extension-bias}, we extend \cref{co:cr} to provide lower bounds for biased estimators. 

% Components of the proof of \cref{thm:fi} are provided in Section~\ref{sec:cr_proof}.

\subsection{Time to Learn} \label{sec:time_to_peak}

\cref{co:cr} implies that at least $N^{2/3}$ observations are needed before we can learn $N$. Here we characterize how long the diffusion model takes to reach this point relative to the time it takes to reach the point when the rate of new infections is at its peak. In both settings, the peak corresponds to a time in which a constant fraction of the population has been infected.

\subsubsection{Bass Model.}

One way to characterize the time at which the rate of new adopters in the Bass model peaks is to identify the first epoch at which the expected time until the next adoption increases. That is, defining
\[
k^* = \inf \{k: \bE[T_{k}] \geq \bE[T_{k-1}]\ \},
\]
$t_{k^*}$ corresponds to the (random) time at which this peak in the rate of new adoptions occurs. We denote by $t_{k^{\rm CR}}$ (where $k^{\rm CR} \triangleq \lceil N^{2/3} \rceil$) the earlier time at which we have sufficiently many observations to estimate $N$ accurately per \cref{co:cr}.
%we define $t_{k^{\rm CR}} \triangleq t_{\lceil N^{2/3} \rceil}$ as the time at which we collect $N^{2/3}$ observations.
%
%
%Let $k^*$ be the first epoch at which the expected time until
%
%Letting $k^{\rm CR} \triangleq \lceil N^{2/3} \rceil$, $t_{k^{\rm CR}}$ is the time at which number of adopters is past $N^{2/3}$.
%
%Let $t_1 = \inf\{I(t) \geq N^{2/3}\}$ be the time when the number of adopters is past $N^{2/3}$.
%Define $\mathring{m} = \min\left\{k\geq 1~|~\bE[T_{k}] \geq \bE[T_{k-1}]\right\}$, and let $t_2 = \inf\{I(t) \geq \mathring{m}\}$
%be the time when the rate of decrease of $S(t)$ achieves its maximum (i.e. ``the time at which the sales rate reaches its peak'' from \cite{bass1969new}).
%It is easy to show that $\mathring{m} = \left\lceil\left(\left(1-\frac{p}{\beta}\right)N+1\right)/2\right\rceil.$
The following result characterizes the ratio ${\bE[t_{k^{\rm CR}}]}/{\bE[t_{k^*}]}$ as $N \rightarrow \infty$:

\begin{proposition} \label{prop:time_to_peak_bass}
Suppose $\gamma = 0, I_0= 1, \frac{p}{\beta} < c$ for some constant $c < 1$. Suppose $\frac{p}{\beta} = \Theta(\frac{1}{N^{\alpha}})$ for $\alpha \geq 0$.
\begin{align*}
\lim_{N \rightarrow \infty} \frac{\bE[t_{k^{\rm CR}}]}{\bE[t_{k^*}]}=
\begin{cases}
0 & \alpha \leq \frac{1}{3}\\
\frac{\alpha - \frac{1}{3}}{\alpha} & \frac{1}{3} < \alpha < 1\\
\frac{2}{3} &  \alpha \geq 1.
\end{cases}
\end{align*}
\end{proposition}
We see that the fraction of time until peak by which we can hope to learn the Bass model, ${\bE[t_{k^{\rm CR}}]}/{\bE[t_{k^*}]}$, depends on $p/\beta$. 
This latter quantity provides a measure of the relative contribution of innovators and imitators to the instantaneous rate of overall adoption.

% What the result above shows is that if adoption is driven largely by imitation (the case for many products that rely on word-of-mouth or network effects, and certainly for information) so that $p/\beta$ is small ($\alpha \geq 1$), we need to wait at least two-thirds of the way until peak to collect enough samples to learn $N$.

\edit{
When this quantity is small ($\alpha \geq 1$), we need to wait at least two-thirds of the way until peak to collect enough samples to learn $N$.
An interpretation of the regime of $\alpha = 1$ is the following.
If we treat $\beta$ as a constant, then $p = \Theta(1/N)$.
Since the growth of innovators is approximately $pN$ at the start of the process, this regime implies that the number of innovators in the early stages is a \textit{constant}, which does not depend on $N$.
Since $N$ is unknown in the early stages, $\alpha=1$ represents the regime where the rate of innovators at the start of the process does not depend on the (unknown) eventual popularity of the product.
% This is a sound assumption in the sense that in the early stages, since no one knows the eventual total number of adopters, the number of innovators would not depend on this unknown quantity.
There are also many empirical works that estimate the Bass model parameters for consumer products (e.g., air conditioners, TVs, etc.), which establish that adoption is mainly driven by imitation rather than innovation (e.g., \cite{sultan1990meta,mahajan1995diffusion,lee2014pre}).

% On the other hand, when $\alpha$ is smaller, the fraction of the time to the peak by which we can learn the Bass model decreases, reducing to 0 when $\alpha \leq 1/3$. Products that use new, preliminary technologies may lie in this regime; for example, \cite{islam2014household} estimates that for photo-voltaic solar panels, the coefficient of innovation and imitation are similar in size, owing to the ``very early stage of this innovation''.
}

\subsubsection{SIR Model.}

% Unfortunately, characterizing either of these (random) times exactly in the stochastic SIR process appears to be a difficult task and so we consider analyzing these times in the deterministic model.
For the SIR model, characterizing the random time in which the process hits either the peak infection rate or $N^{2/3}$ observations appears to be a difficult task.
Therefore, we analyze the analogs of $t_{k^{\rm CR}}$ and $t_{k^*}$ in the deterministic model \eqref{eq:sirfluid}.
Specifically, let $t^d_{\rm CR} = \inf\left\{t: c(t) \geq N^{2/3} \right\}$ and $t^d_{*} = \inf\left\{t: {d^2s}/{dt^2} > 0 \right\}$ for the process defined by \eqref{eq:sirfluid}.

\begin{proposition}
  \label{thm:deterministic-time}
Suppose $p=0$ and $\beta, \gamma$ are fixed. If $c(0) = O(\log(N))$,
\begin{align*}
\liminf_{N \rightarrow \infty}  \frac{t^d_{\rm CR}}{t^d_{*}} \geq \frac{2}{3}.
\end{align*}
\end{proposition}

This suggests that the sampling requirements made precise by \cref{co:cr} can only be met at such time where we are close to reaching the peak infection rate.
Unlike the Bass model, this ratio is not specific to a parameter regime for the model.

We note that the results of \cref{prop:time_to_peak_bass} and \cref{thm:deterministic-time} require $N \to \infty$.
While the specific value of `two-thirds' depends on the limit $N \to \infty$, we interpret the significance of this result to be that the time to learning $N$ is a \textit{constant fraction} of the time to peak (i.e.\ it is `late'), rather than focusing on the precise two-thirds value.
In \cref{sec:experiments}, we validate this observation on a real-world dataset by demonstrating that the time taken to acquire $N^{2/3}$ observations aligns with the late stages of the diffusion process.
 % (see Figure~\ref{fig:time-to-N23}).

\subsection{Estimating Other Parameters} \label{sec:estimating_other_params}
We now turn our attention to learning the other parameters of the model. The high-level message here is that parameters other than the population $N$ are in general easier to learn, and this is best understood through Table~\ref{tab:bass_estimation}. Specifically, the second row in that table shows the number of observations needed for a relative error less than one. Our earlier analysis provides lower bounds on this quantity for the estimation of $N$. Here we construct explicit estimators for the remaining parameters yielding upper bounds on the number of observations required to learn those parameters with a relative error less than one.

We immediately see that for the SIR model, we can accomplish this task with a number of observations that {\em does not scale} with the population size parameter. In the case of the Bass model the story is more nuanced: it is always easier to learn the coefficient of imitation, $\beta$. On the other hand when the rate of innovation is very low, learning $a := pN$ is hard, but also not relevant to tasks related to forecasting $N$. We next present formal results that support the quantities in Table~\ref{tab:bass_estimation}.

% informs template
% \begin{table}
% \TABLE
% {Summary of parameter estimation results for the Bass and SIR models.
% The first row shows the relative error of estimating each parameter with $m$ observations.
% The second row shows $\argmin_m \{\relerror(\hat{\theta}_m, \theta)\leq1\}$, the number of observations needed so that the relative error is less than 1.
% For the Bass model, $a=pN$, and $a = \Theta(N^{1-\alpha})$ for $\alpha \geq 0$.\label{tab:bass_estimation}}
% {
% \begin{minipage}{\columnwidth}
% \begin{center}
% % \begin{tabular}{@{}cccc@{}}
% \begin{tabular}{c|c|cc|cc}
% \toprule
%     &   & \multicolumn{2}{c|}{Bass} & \multicolumn{2}{c}{SIR} \\ % \cline{3-6}
%   & $N$* & $\beta$          & $a$**         & $\beta$       & $\gamma$       \\
% \midrule
% $\relerror(\hat{\theta}_m, \theta)$ & $\Omega\left(\frac{N^2}{m^3}\right)$ & $\tilde{O}\left(\frac{1}{m} + \frac{N^{2(1-\alpha)}}{m^3}\right)$& $\tilde{O}\left(\frac{1}{m} +\frac{1}{N^{(1-\alpha)}}\right)$   & ${O}\left(\frac{\log m}{m}\right)$ &  ${O}\left(\frac{\log m}{m}\right)$ \\
% \# observations needed & $\Omega\left(N^{2/3}\right)$ & $\tilde{O}\left( \max\{1, N^{\frac{2}{3}(1-\alpha)}\} \right)$& $\tilde{O}\left(1\right)$  & $O\left(1\right)$ & $O\left(1\right)$ \\
% \bottomrule
% \end{tabular}
% \bigskip\centering
% \\
% \footnotesize
% *The column for $N$ represents the \ii{expected} relative error, whereas the other parameters are high-probability results. \\
% **We note that the results for the parameter $a$ hold only for $\alpha < 1$.
% \end{center}
% \end{minipage}
% }
% {}
% \end{table}

% Normal template
\begin{table}%
\caption{
Summary of parameter estimation results for the Bass and SIR models.
The first row shows the relative error of estimating each parameter with $m$ observations.
The second row shows $\argmin_m \{\relerror(\hat{\theta}_m, \theta)\leq1\}$, the number of observations needed so that the relative error is less than 1.
For the Bass model, $a=pN$, and $a = \Theta(N^{1-\alpha})$ for $\alpha \geq 0$.
}
\label{tab:bass_estimation}
\begin{minipage}{\columnwidth}
%\begin{center}
% \begin{tabular}{@{}cccc@{}}
\begin{tabular}{c|c|cc|cc}
\toprule
    &   & \multicolumn{2}{c|}{Bass} & \multicolumn{2}{c}{SIR} \\ % \cline{3-6}
  & $N$* & $\beta$          & $a$**         & $\beta$       & $\gamma$       \\
\midrule
$\relerror(\hat{\theta}_m, \theta)$ & $\Omega\left(\frac{N^2}{m^3}\right)$ & $\tilde{O}\left(\frac{1}{m} + \frac{N^{2(1-\alpha)}}{m^3}\right)$& $\tilde{O}\left(\frac{1}{m} +\frac{1}{N^{(1-\alpha)}}\right)$   & ${O}\left(\frac{\log m}{m}\right)$ &  ${O}\left(\frac{\log m}{m}\right)$ \\
\# observations needed & $\Omega\left(N^{2/3}\right)$ & $\tilde{O}\left( \max\{1, N^{\frac{2}{3}(1-\alpha)}\} \right)$& $\tilde{O}\left(1\right)$  & $O\left(1\right)$ & $O\left(1\right)$ \\
\bottomrule
\end{tabular}
\bigskip\centering
\\
\footnotesize
*The column for $N$ represents the \ii{expected} relative error, whereas the other parameters are high-probability results. \\
**We note that the results for the parameter $a$ hold only for $\alpha < 1$.
%\end{center}
\end{minipage}
\end{table}%

\subsubsection{Bass Model.}
For the Bass model, we construct estimators for the parameters $\beta$ and $a:=pN$ (the explicit construction is given in \cref{sec:app:estimator_proof_bass}):
%
%As for the parameter $p$, estimating this is very closely linked with estimating $N$.
%Since $p$ only appears as $pS$ in \eqref{eq:deftk}, it is essentially impossible to decouple $p$ from $N$ since $S \approx N$ near the beginning of the process.
%We instead show that $a = pN$ is easy to estimate.
%This will imply that estimating $p$ is just as hard as estimating $N$, since they are linked through $a$.
%
%We analyze the relative errors of estimators for $\beta$ and $a$:
\begin{theorem} \label{thm:bass-estimation-p-a}
Suppose $\gamma = 0$ and $I_0 = 1$. Let $a = pN.$ Suppose $m \leq N^{2/3}\log^{1/3}(N)$. 
% There exist 
We construct estimators $\hat{a}_m, \hat{\beta}_m$ based on the observations $O_m$ such that with probability $1-O(\frac{1}{N})$,
\begin{align*}
\relerror(\hat{\beta}_m, \beta) = O\left(\frac{\log N}{m} + \frac{a^2}{\beta^2}\frac{\log N}{m^3} \right), \\
\relerror(\hat{a}_m, a) = O\left(\frac{\log N}{m} + \frac{\beta}{a} \log N\right).
\end{align*}
\end{theorem}
The above result demonstrates that learning the coefficient of imitation, $\beta$, is always easier than estimating $N$.  This is also the case for $a$ when $p/\beta = \omega(1/N)$; when $p/\beta = O(1/N)$, the number of innovators who adopt is negligible compared to the number of imitators and it is not possible to estimate $a$.

% Note that the estimation of the parameters $\beta$ and $a$ are also linked.
%To parse the above result, consider the regime where $\beta$ is constant and $p = \Theta(\frac{1}{N^\alpha})$, implying $a = \Theta(N^{1-\alpha})$ for $\alpha \geq 0$.
%In this regime, we summarize the results of \cref{thm:bass-estimation-p-a} in \cref{tab:bass_estimation}.
%The second row of the table displays the number of observations needed to get the relative error below 1 for each of the three parameters.
%We see that estimating $\beta$ is always faster than estimating $N$, and the extent to how much faster is determined by $\alpha$.
%When $\alpha \geq 1$, the number of observations needed to estimate $\beta$ is constant.
%On the other hand, $a$ is easy to estimate when $\alpha < 1$.
%In the case that $\alpha \geq 1$, the number of innovators is negligible compared to the number of imitators that it is not possible to estimate $a$.

% Table

\subsubsection{SIR Model.}

For the SIR model we construct estimators for the parameters $\beta$ and $\gamma$ (the explicit construction is given in \cref{sec:explicit-construction}):
\begin{theorem}\label{thm:estimation-beta-gamma}
%$\sqrt{\frac{5\log m}{m}} \leq \fourth$
Suppose $p=0$ and $\beta > \gamma$.
Let $C_{0}, m, N$ satisfy $m(m + C_{0}) \leq N$, and $\frac{\beta}{\beta+\gamma}\frac{N-m-C_{0}}{N} > \frac{1}{2}(\frac{\beta}{\beta+\gamma} + \frac{1}{2})$.
Then, we can construct estimators $\hb_{m}$ and $\hg_{m}$, both functions of $O_{m}$, such that with probability $1 - \frac{8}{m} - B_1 e^{-B_2 I_0}$,
\begin{align*}
\relerror(\hb_{m}, \beta) &\leq M_1 \left(\frac{\log m}{m}\right), \\
\relerror(\hg_{m}, \gamma) &\leq M_2 \left(  \frac{\beta^2}{\gamma^2}  \frac{\log m}{m} \right),
 % \bE[(\hb_{m} - \beta)^2] &\leq M_1 \beta^2 \frac{\log m}{m} , \\
 % \bE[(\hg_{m} - \gamma)^2] &\leq  M_2 \beta^2\frac{\log m}{m},
\end{align*}
% where $B_1, B_2>0$ depend only on $\beta$ and $\gamma$ and
where $M_1, M_2>0$ are absolute constants and $B_1, B_2 > 0$ depends only on $\beta$ and $\gamma$.
\end{theorem} When $\beta$ and $\gamma$ do not scale with the size of the system $N$ (which is the case for epidemics), this result shows that the relative error for both estimators is $O\left({\log m}/{m}\right)$, i.e. independent of $N$. Consequently, to achieve any desired level of accuracy, we simply need the number of observations $m$ to exceed a constant that is independent of the size of the system.
This is in stark contrast to \cref{co:cr}, in which $m$ needs to scale at least as $\omega(N^{2/3})$ in order to learn $N$.

\subsection{Extension to Biased Estimators}\label{sec:extension-bias}
Although we focus on unbiased estimators of $N$ in this work, a lower bound for \textit{biased} estimators of $N$ can also be easily obtained via the generalized Cramer-Rao bound \citep{cramer1946mathematical}, which bounds the variance of biased estimators with given bias and Fisher information.
Using the Fisher information from \cref{thm:fi}, the generalized Cramer-Rao bound implies the following result.
\begin{proposition}
Under \cref{assump:bass} or \cref{assump:sir}, if $\hat N_m$ is a biased estimator of $N$, with bias $b(N) = \bE[\hat N_m] - N$, based on the observations $O_{m}$,
% Under \cref{assump:asymptotic}, suppose either $p=0$ or $\gamma=0$, and $\beta, \gamma$, and $pN$ are known.
% Assume $m = o(N)$ and $I_0, R_0 \leq m$.
% Let $\hat N$ be any unbiased estimator of $N$ based on the observations $O_{m}$.
% There exists a constant $D$ that only depends on $\beta$ and $\gamma$ such that if $I_0 \geq D$,
\begin{align}
% \frac{\variance(\hat N)}{N^2} = \Omega\left(\frac{N^2}{m^3}\right).
\bE[\relerror(\hat{N}_m, N)] = \Omega\left(\frac{N^2 (1+b'(N))^2 }{m^3} + \frac{b(N)^2}{N^2}\right).
% \frac{\bE[(\hat N - N)^2]}{N^2} = \Omega\left(\frac{N^2}{m^3}\right).
\end{align}
\end{proposition}
When $|1+b'(N)| < 1$, this bound may be less than the unbiased Cramer-Rao bound in \cref{co:cr}. The result can be used to guide the design of estimators for balancing the bias and variance \citep{eldar2008rethinking}.

%Therefore, estimating $N$ is clearly the main challenge in the estimation of an SIR model.

% $N$ is clearly the limiting factor in estimating the SIR model, as the other two parameters are
% the parameters $\beta$ and $\gamma$ can be learned to any desired level of accuracy after a constant number of time steps.

% In stark contrast with Corollary~\ref{co:cr}, Theorem~\ref{thm:estimation-beta-gamma} shows that the variance in estimating $\beta$ and $\gamma$ is of order $O(\log m /m )$ and is independent of $N$. Consequently, to achieve any desired level of accuracy, we simply need the number of events $m$ to exceed a constant that is independent of the size of the system $N$.

\subsection{Decision Problem} \label{sec:decision}
In this section, we formalize a generic decision problem in the context of an epidemic. 
The decision is whether to impose a `drastic intervention' that is associated with a cost, but will immediately stop all further infections.
We establish a lower bound on the regret that any policy will incur for this decision problem.

We assume that there is a cost, $f_1 > 0$, for every infected individual.
A drastic intervention will immediately stop any further new infections, but the intervention will incur a fixed cost of $f_0 > 0$.
Then, a drastic intervention implemented at step $m$ incurs a cost of $f_0 + f_1 C_{m}$, recalling that $C_m$ represents the cumulative number of infected individuals up to step $m$.
On the other hand, if no intervention is implemented, the cost is solely from the infections, which is $f_1 C_{\infty}$, where $C_{\infty}$ represents the total number of accumulated infections at the end of the epidemic.
We study policies that decide when, if ever, to deploy this drastic intervention.

A problem instance is defined as $\mathcal{M} = (f_0, f_1, \beta, \gamma, N)$.
We define the optimal cost as ${\rm{cost}}^{*}(\mathcal{M}) = \min\{f_0 + f_1 C_0, f_1 C_{\infty}\}$, which is the cost of the optimal policy that has knowledge of the entire problem instance, $\mathcal{M}$.
We consider policies that have knowledge of all parameters except $N$. The policy has access to all of the observations of the diffusion process, and the policy faces a stopping problem regarding whether and when to employ the drastic intervention.
The regret of a policy $\pi$ is:
\begin{align}
{\rm{regret}}^{\pi}(\mathcal{M}) := \mathrm{E}[{\rm{cost}}^{\pi}(\mathcal{M})] - {\rm{cost}}^{*}(\mathcal{M}).
\end{align}

We prove the following lower bound on the regret.
\begin{proposition} \label{thm:decision}
    There exists a set of problem instances $\mathcal{S}_{N}$ that are parameterized by $N$, for any policy $\pi$, 
    $$
    \sup_{\mathcal{M} \in \mathcal{S}_{N}} {\rm{regret}}^{\pi}(\mathcal{M}) = \Omega(N^{2/3}).
    $$
    That is, the lower bound on the regret for any policy $\pi$ is $\Omega(N^{2/3}).$
\end{proposition}

The proof of \cref{thm:decision} relies on constructing two instances where the optimal decision is different, but it is difficult to distinguish between these two instances due to the uncertainty in the estimation of $N$. The full proof can be found in Appendix~\ref{sec:app:decision}.
This result shows that the hardness in estimating $N$ translates directly to the difficulty of a generic decision problem on implementing an intervention.

\subsection{Addressing Under-reporting through Seroprevalence Testing} 

The results so far have demonstrated that the estimation of the parameter $N$ is the bottleneck for forecasting using infection data. 
In the epidemic setting, and for COVID-19 in particular, one of the main source of uncertainty in $N$ came from under-reporting \citep{li2020substantial,lau2021evaluating,pullano2021underdetection}.
To overcome this challenge, one can potentially use other data sources in order to improve the estimation of $N$.
For instance, for COVID-19, surveillance tests were often conducted to estimate the prevalence of infections without under-reporting bias \citep{havers2020seroprevalence,bendavid2021covid}. 
In this section, we study the value of utilizing such a dataset.
Specifically, we assume that a random sample of $K$ people are tested for the infection after $m$ observations of the SIR process.
We compute the value of this information via the Fisher information and the Cramer-Rao bound, analogous to our main result of \cref{co:cr}.

% Recall that the information collected at the $m$-th step for the SIR model is:
% $$
% O_{m}=(I_0, T_{1}, C_{1}, \dotsc, T_{m}, C_{m})
% $$
% Here, $C_{k}$ denotes the accumulated observed infected population at step $k$, and $T_{k}$ represents the time interval between step $k-1$ and $k$. 
% In addition to the observations $O_m$ of the SIR process, 
After $m$ observations of the SIR process, we assume that $K$ people are chosen at random to be tested for infection.
Then, the infection rate for a randomly chosen person is:
\begin{align*}
\kappa_{m} = \frac{E[C_{m}]}{N},
\end{align*}
which is the ratio between the expected cumulative number of observed infections and the (effective) population. 
Letting $X_{k} \sim \mathrm{Ber}(\kappa_{m})$ be independent Bernoulli random variables that represent the infection outcome of the $k$-th chosen patient, the observation set from the test is:
\begin{align*}
\tilde{O}_{m} = (X_{1}, X_{2}, \dotsc, X_{K}).
\end{align*}
% Here, $X_{k} \sim \mathrm{Ber}(\kappa_{m})$ are independent Bernoulli random variables for $k \in [K]$, representing the infection outcome of the $k$-th person.
Considering this additional information, we establish the following result.
\begin{proposition} \label{thm:sero}
Under Assumption 2, if $m=o(N)$, then the Fisher information of $O_{m}\cup \tilde{O}_{m}$ relative to $N$ is:
\begin{align*}
J_{{O_{m}\cup \tilde{O}_{m}}}(N) = \Theta\left(\frac{m^{3}}{N^4}\right) + \Theta\left(\frac{Km}{N^3}\right)
\end{align*}
where $\Theta\left(\frac{Km}{N^3}\right)$ quantifies the exact additional Fisher information provided by $\tilde{O}_{m}$.
\end{proposition}

% Applying the Cram\'{e}r-Rao bound to the Fisher information of $\tilde{O}_m$, for any unbiased estimator $\hat N_m$ of $N$ based on $\tilde{O}_{m}$, we have
% \begin{align*}
% \bE[\text{RelError}(\hat{N}_m, N)] = \Omega\left(\frac{N}{Km}\right). 
% \end{align*}

The proof can be found in Appendix~\ref{sec:app:sero}.
Using this, we apply the Cramer-Rao bound for estimating $N$ based on $\tilde{O}_{m}$.
\begin{corollary} \label{cor:sero}
For any unbiased estimator $\hat N_m$ of $N$ based on the observations $\tilde{O}_{m}$,
\begin{align*}
\bE[\text{RelError}(\hat{N}_m, N)] = \Omega\left(\frac{N}{Km}\right). 
\end{align*} 
\end{corollary}

\cref{cor:sero} provides a lower bound for estimating $N$ using a seroprevalence study, and a naive MLE estimator can be used to achieve the lower bound.
In order for $\bE[\text{RelError}(\hat{N}_{m}, N)]=o(1)$, it is necessary to ensure that $Km = \omega(N)$.
This clearly delineates the trade-off in the size of the campaign, $K$, versus the timing of the campaign, $m$.
%However, as long as $m = \omega(1)$, the size of the test can be sublinear in $N$.
For example, if $m=\Theta(N^{1/3})$ we require $K = \omega(N^{2/3})$.
Therefore, with a sufficiently large seroprevalence test, we have the potential to surpass the lower bound barrier of two-thirds in the early stages of the epidemic.

% \input{sections/Tianyi_results.tex}

%!TEX root=../SIR-model-MS.tex
\section{Proof of Theorem~\ref{thm:fi}} \label{sec:cr_proof}

Recall that $O_m = (I_0, R_0, T_1, C_1, \dots, T_m, C_m)$.
We will take advantage of conditional independence to decompose the Fisher information $\cJ_{O_m}(N)$ into smaller pieces.
We first define the conditional Fisher information and state some known properties \citep{zegers2015fisher}.
\begin{definition}
Suppose $X, Y$ are random variables defined on the same probability space whose distributions depend on a parameter $\theta$.
Let $g_{X|Y}(x, y, \theta) = \frac{\partial}{\partial \theta} \log f_{X|Y; \theta}(x|y)^2$ be the square of the \textit{score} of the conditional distribution of $X$ given $Y=y$ with parameter $\theta$ evaluated at $x$.
Then, the \textit{conditional Fisher information} is defined as
$\cJ_{X|Y}(\theta)  = \bE_{X, Y}\left[g_{X| Y}(X, Y, \theta)\right]$.
\end{definition}
\begin{property} \label{property:chain_rule}
  $\cJ_{X_1, \dots, X_n}(\theta) = \cJ_{X_1}(\theta) + \sum_{i = 2}^n \cJ_{X_i | X_1, \dots, X_{i-1}}(\theta)$.
\end{property}
\begin{property}  \label{property:indep}
  If $X$ is independent of $Z$ conditioned on $Y$, $\cJ_{X|Y,Z}(\theta) = \cJ_{X|Y}(\theta)$.
\end{property}
\begin{property} \label{property:deterministic}
If $X$ is deterministic given $Y=y$, $g_{X| Y}(X, y, \theta) = 0$.
\end{property}
\begin{property} \label{property:reparam}
If $\theta(\eta)$ is a continuously differentiable function of $\eta$, $\cJ_X(\eta) = \cJ_X(\theta(\eta)) (\frac{d \theta}{d \eta})^2$.
\end{property}

Since $I_0$ and $R_0$ are known and not random, the Fisher information of $O_m$ is equal to the Fisher information of $(T_1, C_1, T_2, C_2, \dots, T_m, C_m)$.
% Since $O_m = (C_0, T_1, C_1, T_2, C_2, \dots, T_m, C_m)$,
Then, Property~\ref{property:chain_rule} implies
\begin{align} \label{eq:fi_expanded}
% \cI_{O_m}(N) &= \cI_{C_0}(N)  + \cI_{T_1|C_0}(N) + \cI_{C_1|T_1,C_0}(N) + \cI_{T_2|C_1,T_1,C_0}(N) +  \cI_{C_2|T_2,C_1,T_1,C_0}(N) + \dots.
\cI_{O_m}(N) &= \cI_{T_1}(N) + \cI_{C_1|T_1}(N) + \cI_{T_2|T_1,C_1}(N) +  \cI_{C_2|T_1,C_1,T_2}(N) + \dots
+\cI_{C_M|T_1, C_1, \dots, T_m}(N)
.
\end{align}

\textbf{Bass Model: }
The above expression simplifies greatly for the Bass model since every event corresponds to a new infection.
That is, we know $C_k = I_k = I_0 + k$ and $S_k = N - k - I_0$ deterministically.
Therefore, Property~\ref{property:deterministic} implies that $\cI_{C_k|\cdot}(N) = 0$ for all $k$.
Moreoever, since $T_k \sim \exp(\beta \frac{S_{k-1}}{N}I_{k-1} + \frac{a}{N} S_{k-1})$ is independent of $T_1, C_1, \dots, C_{k-1}$, $\cI_{T_k|T_1, C_1, \dots, C_{k-1}}(N) = \cI_{T_k}(N)$.
This yields
\begin{align} \label{eq:fi_bass1}
\cI_{O_m}(N) &= \sum_{k=1}^m \cI_{T_k}(N).
\end{align}
By letting $\lambda_k(N) = \left( \frac{\beta}{N}(k+I_0) + \frac{a}{N} \right)(N-k - I_0)$,
since $T_k \sim \exp(\lambda_k(N))$, Property~\ref{property:reparam} says that $\cI_{T_k}(N) = \cI_{T_k}(\lambda_k)\left( \frac{d \lambda_k}{d N}  \right)^2$.
Using that the Fisher Information of an exponential distribution with parameter $\lambda$ is $\frac{1}{\lambda^2}$, a few lines of algebra yields $\cI_{T_k}(N) = \frac{(k+I_0)^2}{N^2(N-k - I_0)^2}$.
Plugging back into \eqref{eq:fi_bass1}, we get
\begin{align} \label{eq:fi_bass2}
\cI_{O_m}(N) 
% = \sum_{k=1}^m \cI_{T_k}(N) 
= \frac{1}{N^2} \sum_{k=1}^m \frac{(k+I_0)^2}{(N-k - I_0)^2}.
\end{align}
Using $I_0=1$ and $m = o(N)$ from \cref{assump:bass}, we get the desired result $\cI_{O_m}(N)  = \Theta\left(\frac{m^3}{N^4}\right)$.
% \begin{align*}
% \cI_{T_k}(N) = \frac{k}{N^2(N-k)^2}.
% \end{align*}
% \begin{align*}
% \cI_{O_m}(N) 
% = \sum_{k=1}^m \cI_{T_k}(N) 
% = \frac{1}{N^2} \sum_{k=1}^m \frac{k}{(N-k)^2}
% = \Theta\left(\frac{m^3}{N^4}\right).
% \end{align*}

\textbf{SIR Model: }
The analysis for the SIR model is more complicated since $C_k$ is not deterministic and the distribution of $T_k$ depends on $C_{k-1}$.
Moreover, there is a non-zero probability that the process has terminated before the $k$'th jump for any $k$.
% Define $p \triangleq \frac{1}{2}(\frac{\beta}{\beta+\gamma} + \frac{1}{2}) > \frac{1}{2}$. Assume $N$ is large enough so that $m + C_0 \leq \frac{N}{2}$ and $\frac{\beta (N - m - C_0)}{\beta (N - m - C_0) + P \gamma} > p$ (this is possible since $\frac{\beta}{\beta + \gamma} > p$ and $m = o(N)$).
Define the indicator variable $E_k = \bI\{\tau > k\}$ on the event that the SIR process has not terminated after $k$ jumps.
The following lemma states that both $E_k$ and $I_k$ can be determined from $C_k$, $I_0$, and $R_0$, which will allow us to decouple variables in $O_m$ in the analysis of the Fisher information.
The result follows from the definitions of $\tau$, $E_k$, and $C_k$; the details can be found in the Appendix.

\begin{lemma} \label{lem:E_k}
Define $r_k \triangleq \frac{I_0 + k + 2R_0}{2}$ for all $k \geq 0$.
For all $k$, $E_k = \bI\{C_{k} > r_{k}\}$.
Moreover, when $E_k = 1$, $I_k = 2C_k - k - I_0 - 2R_0 > 0$.
% Moreover, $C_k$ determines $E_k$.
\end{lemma}
The next lemma writes an exact expression for $\cI_{O_m}(N)$, analogous of \eqref{eq:fi_bass2} for the Bass model:
\begin{lemma}  \label{lemma:fi}
The Fisher information of the observations $O_m$ with respect to the parameter $N$ is
\begin{align}
\cI_{O_m}(N)
&= \sum_{k=1}^{m} \Pr(E_{k-1}=1)  \bE\left[\frac{C_{k-1}^2}{N^2 (N - C_{k-1})(N - C_{k-1}  + \frac{\gamma}{\beta}N)} \;\bigg|\; E_{k-1} = 1 \right]. \label{eq:exactFI}
\end{align}
\end{lemma}
\begin{myproof}
We start from \eqref{eq:fi_expanded}.
Note that for any $k$, $C_k$ and $T_k$ only depend on $C_{k-1}$.
Indeed, since $C_{k-1}$ determines $E_{k-1}$, if $E_{k-1} = 0$ (the stopping time has passed), we have $C_k = C_{k-1}$ and $T_k = \infty$.
When $E_{k-1} = 1$, the distributions of $C_k$ and $T_k$ are given in \eqref{eq:defck}-\eqref{eq:deftk}. Since $\beta, \gamma, I_0, R_0$ are known, $S_{k-1} = P - C_{k-1}$, and $I_{k-1}$ can be determined from $C_{k-1}$ (Lemma~\ref{lem:E_k}), the distributions of $C_k$ and $T_k$ are determined by $C_{k-1}$.
Therefore, we use Property~\ref{property:indep} to simplify \eqref{eq:fi_expanded} to
\begin{align*}
\cI_{O_m}(N) &=\sum_{k=1}^m (\cJ_{C_k|C_{k-1}}(N) + \cJ_{T_k|C_{k-1}}(N)),
\end{align*}
where we used $\cI_{T_1}(N) =  \cI_{T_1|C_0}(N)$, $\cI_{C_1}(N) =  \cI_{C_1|C_0}(N)$.
Moreover, when $E_{k-1}=0$, $C_k$ and $T_k$ are deterministic conditioned on $C_{k-1}$, which implies the
score in this case is 0 (Property~\ref{property:deterministic}).
Therefore, we can condition on $E_{k-1}=1$ to write
\begin{align*}
\cI_{O_m}(N)
&= \sum_{k=1}^m \bE[g_{C_k|C_{k-1}}(C_k,C_{k-1}, N) + g_{T_k|C_{k-1}}(T_k,C_{k-1}, N)| E_{k-1}=1] \Pr(E_{k-1}=1).
\end{align*}
The last step is to evaluate $g_{C_k|C_{k-1}}(C_k,C_{k-1}, N)$ and $g_{T_k|C_{k-1}}(T_k,C_{k-1}, N)$.
When $E_{k-1}=1$, the distributions of $C_k$ and $T_k$ conditioned on $C_{k-1}$ have a simple form provided in \eqref{eq:defck}-\eqref{eq:deftk}.
Property~\ref{property:reparam} allows for straight-forward calculations, resulting in \eqref{eq:exactFI}. See \cref{sec:app:cr_calculations} for details of this last step.
\end{myproof}

What remains is to upper and lower bound \eqref{eq:exactFI}. 
The upper bound $\cI_{O_m}(N) = O\left( \frac{m^3}{N^4}  \right)$ follows from upper bounding $\Pr(E_{k-1})$ by 1 and the fact that $C_{k-1}$ is small relative to $N$
(details of this step are in \cref{sec:app_cr_last_step}).
As for the lower bound, we first show a lower bound for $\Pr(E_{k-1} = 1)$ using the following lemma:
\begin{lemma} \label{lem:highprob}
Let $p = \frac{1}{2}\left( \frac{\beta}{\beta+\gamma} + \frac{1}{2} \right) > \frac{1}{2}$.
There exists a constant $D$ that only depends on $\beta$ and $\gamma$ such that if $\frac{\beta (P- m-C_0)}{\beta (P - m - C_0) + P\gamma} > p$  and $I_0 \geq D$,
then $\Pr(E_{m}=1) \geq \frac{1}{2}$.
\end{lemma}
This result relies on an interesting stochastic dominance argument and can be found in the appendix.
Then, similarly to the upper bound, $\cI_{O_m}(N) = \Omega\left(\frac{m^3}{N^4}\right)$ follows from using $\Pr(E_{m}=1) \geq \frac{1}{2}$ and the fact that $C_k \geq \frac{k + I_0 + 2R_0}{2}$ when $E_k =1$ (Lemma~\ref{lem:E_k}).

\section{Numerical Results}
\label{sec:experiments}

We run experiments on real-world datasets for both the Bass and SIR models to demonstrate how the theoretical results from \cref{sec:limits_to_learning} manifest in practice. We describe three sets of empirical results:
\begin{itemize}
\item \cref{sec:exp:benchmark_datasets} mirrors the theory in this paper and makes two points: First, the relative error one sees in real-world datasets on quantities of interest as a function of the number of observations closely hews to that predicted by our results. Second, the time at which predictions of key quantities `turn accurate' is late in the diffusion and again matches our theory.

\item In \cref{sec:hetergeneous_mixing}, we conduct a set of semi-synthetic experiments on a variant of the SIR model that captures heterogeneously mixing subpopulations. Since the SIR model assumes that the population mixes uniformly, a practical use of the SIR model needs to be at the right level of granularity. We show that even when we divide the population into smaller subpopulations with different mixing rates, we observe the same phenomenon regarding the estimation of $N$ --- the accuracy sharply increases after $N^{2/3}$ observations.

\item In \cref{sec:exp:covid19}, we describe an approach for COVID-19 forecasting of US counties that leverages an \textit{informative bias} on $N$ to work around the limits of learning. We consider a realistic variant of the SIR model that accounts for non-stationarities and rich county-level covariates, and we employ a heuristic for estimating the parameters that was directly inspired by our main theoretical results.
Specifically, this estimation method leverages the plurality of US counties, as well as the heterogeneity in the timing of COVID-19 infections across these counties.
We show that the insight of our estimation results can guide the development of forecasting methods which significantly improved the forecasting power, compared to a naive estimation method.
% Our analysis demonstrates that one way to potentially produce accurate forecasts {\em early} in a diffusion would be the use of an estimator with an informative bias on $N$. \cref{sec:exp:covid19} describes how this insight was used in a broader effort to build one of the first broadly distributed county-level forecasts available for COVID-19. 
\end{itemize}

\subsection{A Discrete-Time Diffusion Model}

First, we describe the standard Euler-Maruyama discretization of our stochastic diffusion model; this discretization better aligns with aggregated (as opposed to event level) data. Real-world data is often stored as arrival counts $\Delta C_{i}[t]$ over a set of discrete time periods $t \in [T]$ and problem instances $i \in \mathcal{I}$. We model these counts as the following Poisson process, obtained by approximately discretizing the exponential arrival process \eqref{eq:deftk}. Precisely, we divide the time horizon into $T$ epochs of length $1$, where at each epoch $t \in [T]$ we observe random variables:
\begin{equation}
  \label{eq:poisson}
\begin{aligned}
  \Delta C_{i}[t] &\sim \mathrm{Poisson}(\lambda_{i, t}(a_{i}, \beta_{i}, N_{i})) \\
  \Delta R_{i}[t] &\sim \mathrm{Poisson}(\gamma I_i[t-1]) \\
\end{aligned}
\end{equation}
where $\lambda_{i, t}(a, \beta, N) = (a + \beta I[t-1]) \frac{S[t-1]}{N}$, and $\Delta C_{i}[t]$ and $\Delta R_{i}[t]$ are independent. Essentially, we evaluate the arrival rate of \eqref{eq:deftk} at the beginning of each epoch, and assume that it remains constant over the course of the epoch. This arrival process is then split into $\Delta C_{i}[t]$ and $\Delta R_{i}[t]$ according to the probabilities in \eqref{eq:defck}. The state space then evolves according to:
\begin{equation}
  \label{eq:discrete-dynamics}
\begin{aligned}
  S_{i}[t] &= S_{i}[t-1] - \Delta C_{i}[t] \\
  I_{i}[t] &=  I_{i}[t-1] + \Delta C_{i}[t] -  \Delta R_{i}[t]  \\
  R_{i}[t] &= R_{i}[t-1] + \Delta R_{i}[t]
\end{aligned}
\end{equation}

For the datasets we study, $\gamma$ is known a priori, (i.e. from clinical data for the ILINet flu datasets; for the Bass model $\gamma=0$). We then obtain maximum likelihood estimates $\hat{a}_{i}[t], \hat{\beta}_{i}[t], \hat{N}_{i}[t]$  for the remaining parameters by solving the problem:
\begin{equation}
  \label{eq:llh}
  \max_{\substack{a, \beta, N \in[0, N_{\max}]}} \quad \sum_{\tau=1}^t  \log p(\Delta C_{i} [\tau] ; \lambda_{i, \tau}(a, \beta, N))
\end{equation}
where $p(x; \lambda) = \frac{\lambda^{x} \exp(-\lambda)}{x!}$ denotes the Poisson PMF with rate parameter $\lambda$, and $N_{\max}$ is an upper bound on $N$ known a priori. This reflects that loose upper bounds on $N$ (e.g., the entire population of a geographic region, for epidemic forecasting) are typically known in real-world problems.

\subsection{MLE Performance on Benchmark Datasets} \label{sec:exp:benchmark_datasets}

In this section, we fit the Bass and SIR models to real-world datasets and compare the empirical results to the theoretical results from \cref{sec:limits_to_learning}.

\paragraph{Datasets.} We fit the Bass model to a dataset of Amazon product reviews
from \cite{ni2019justifying}, which we take as a proxy for product adoption\footnote{
This exact dataset is not necessarily the perfect use case of forecasting in the Bass model, as the data contains reviews rather than sales, and the time frame is quite long that an `early' forecast may not be necessary.
The dataset provides non-synthetic, real-world data on the growth and purchasing of many products, hence the experiments provides valuable insights on Bass model forecasting.}. Here each instance $i$ is a product, $t$ indexes weeks since the product's first review, and $I_{i}[t]$ represents cumulative number of reviews for product $i$. For the SIR model, we use the CDC's ILINet database of patient visits for flu-like illnesses. Here, each instance $i$ is a geographic region, $t$ indexes weeks, and $I_{i}[t]$ represents infected patients. See \cref{sec:datasets} for further details on these datasets.

\paragraph{Comparing actual error to error predicted from theory.}
Here we fit diffusion models to products from the Amazon data, as well as individual seasons from the ILINet data, while varying the number of observations used to fit the model. We compare the observed relative error in predicting the effective population size $N$ in each to the error predicted by Theorem~\ref{co:cr}. We find that Theorem~\ref{co:cr} provides a valuable lower bound despite potential model mis-specification, aggregated data, and the fact that we jointly estimate the $a, \beta$ and $N$ parameters.

%Despite differences between this real-world setting and the setting of Theorem~\ref{co:cr} -- namely, that the Poisson arrival process \eqref{eq:poisson} only approximates \eqref{eq:deftk}, and that we simultaneously estimate the parameters $a, \beta, N$ -- we show that the same Cramer-Rao bound predicts the rate at which the estimation error of the MLE decreases.

Specifically, let $T_{i}$ be the time index of the last observation we have for product $i$. We take $\hat{N}_{i}[T_{i}]$ to be the ground truth parameter for product $i$. Figure~\ref{fig:relerror-vs-cramer-rao} is a scatter plot of the mean (over instances $i$ and times $t$) observed relative error $\mathrm{RelError}(\hat{N}_{i}[t], \hat{N}_i[T_i])$
against the Cramer-Rao lower bound of \cref{co:cr}, $M {\hat{N}_{i}[T_{i}] ^{2}}/{ C_i[t]^{3}}$, where $M$ is a lower bound on the constant suppressed in the statement of \cref{co:cr}.  In addition to providing a lower bound, we find that the slope of the relationship is close to one in both datasets as the error grows small. It is worth re-emphasizing that this is the case despite the fact that the data here is not synthetic, so the Bass and SIR models are almost certainly not a perfect fit to the data.

%(observed relative error) as a function of $\log \frac{\hat{N}_{i}[T_{i}] ^{2}}{ C_{i}[t]^{3}}$ (predicted error from Theorem~\ref{co:cr}), which the Cramer-Rao bound predicts to be well-approximated by a line with slope 1 for efficient estimators. For both the Bass and SIR models, we indeed find that this relationship is well-approximated by the best-fit line of slope 1 (the dashed gray line in the figure).

%shows the distribution of $\mathrm{RelError}(\hat{N}_{i}[t], \hat{N}_{i}[T_{i}])$, plotted against the Cramer-Rao bound $M \frac{\hat{N}_{i}[T_{i}] ^{2}}{ C_{i}[t]^{3}}$ of Theorem~\ref{co:cr}, where $M$ is a constant for each model. In both cases, we see that the Cramer-Rao bound is indeed a lower bound, and that it roughly approximates the rate at which RelError decreases in practice.

\begin{figure}[h]
\begin{center}
\begin{subfigure}{0.48\textwidth}
\includegraphics[scale=0.4]{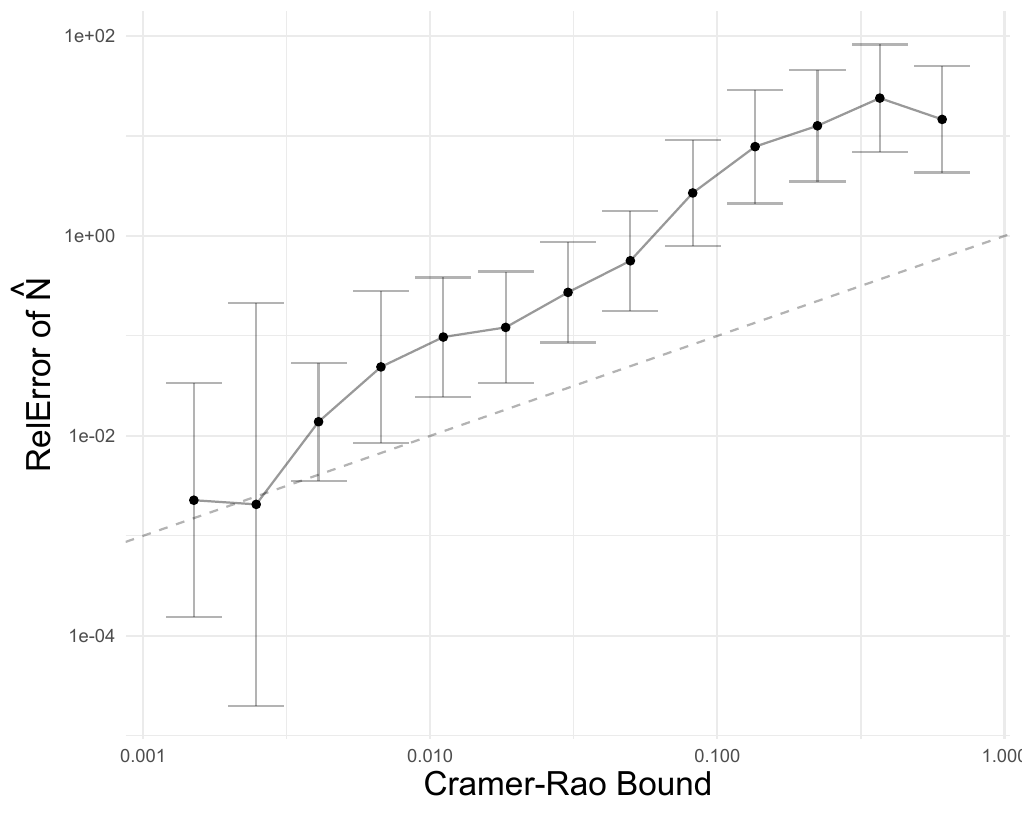}
\caption{Bass model fit on Amazon product reviews.}
\label{fig:amazon-relerror}
\end{subfigure} \hspace{0.02\textwidth}
\begin{subfigure}{0.48\textwidth}
\includegraphics[scale=0.4]{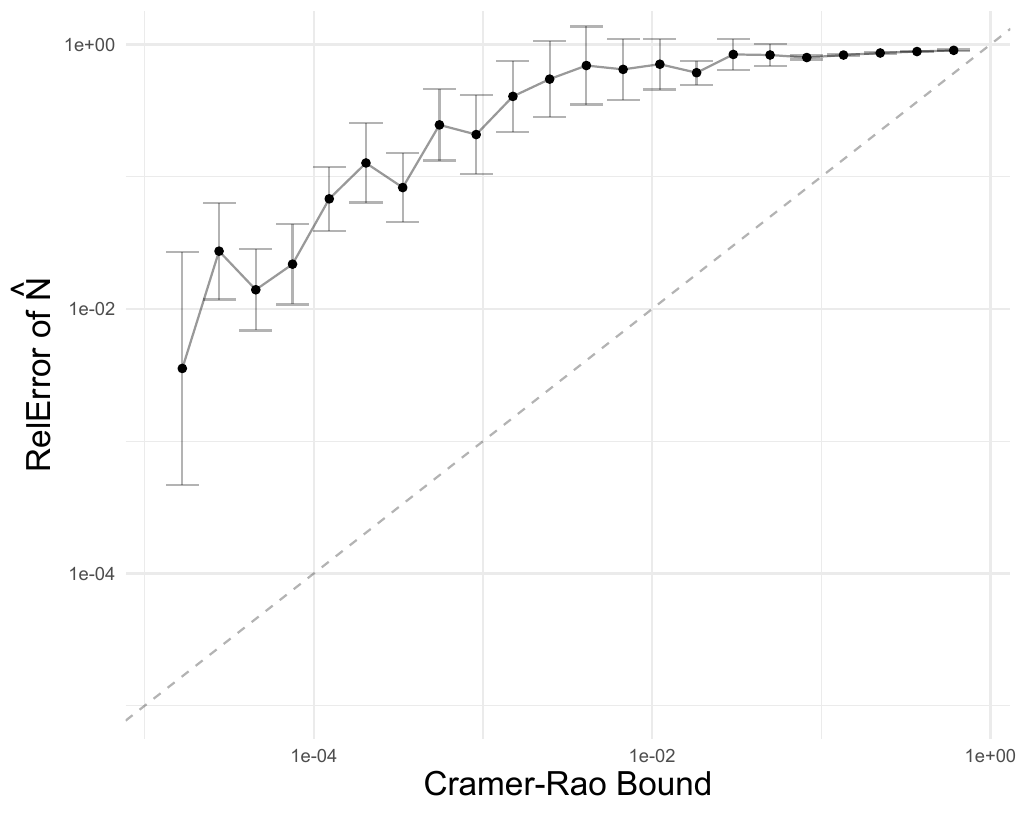}
\caption{SIR model fit on ILINet patient visits.}
\label{fig:flu-relerror}
\end{subfigure}
\vspace{5mm}
\end{center}
\caption{Each figure shows the mean of $\mathrm{RelError}(\hat{N}_{i}[t], \hat{N}_{i}[T_{i}])$ over instances $i \in \mathcal{I}$ and times $t \in [T]$ (error bars show 95\% CIs), vs. the Cramer-Rao bound $\frac{M \hat{N}_{i}[T_{i}] ^{2}}{ C_{i}[t]^{3}}$, where $M$ is a lower bound on the constant suppressed in the statement of \cref{co:cr}. We also show the $y=x$ line (dashed gray) for comparison. As predicted, the Cramer-Rao bound provides a lower bound on $\mathrm{RelError}(\hat{N}_{i}[t], \hat{N}_{i}[T_{i}])$, and the slope of this relationship is close to 1 as the error grows small.}
\label{fig:relerror-vs-cramer-rao}
\end{figure}

\paragraph{Time to accuracy of peak predictions.} As discussed earlier, predicting the peak of the infection process is a key task in the SIR model (as is predicting the peak in new adoptions in the Bass model). Here we show, through the ILINet data, that the time at which our prediction of the peak number\footnote{In \cref{sec:time_to_peak} the peak was defined as the time of the peak {\em rate} of infections rather than the peak number. Both peak definitions refer to a time when a constant fraction of the total population has been infected, and we use the peak number in these experiments as it is a time that is well-defined even with noisy, real-world data.} of infections in an epidemic `turns accurate' is close to the peak and matches what our theory suggests. Specifically, let $I^{*}_{i} = \max_{t \in [T_{i}]} I_{i}[t]$ be the maximal number of infections. Given estimates $\hat{\beta}, \hat{N}$ of the diffusion parameters, we define a point estimate for the peak number of infections
\[
\hat{I}^{*}_{i}(\hat{\beta}, \hat{N}) = \mathbb{E}\left[  \max_{\tau \in [T_{i}]}I_{i }[\tau] \,\big\vert\, \hat{\beta}, \hat{N} \right].
\]

The solid line in Figure~\ref{fig:peak-errors} depicts errors for the estimator $I^{*}_{i}(\hat{\beta}_{i}[t], \hat{N}_{i}[t])$, where
$\hat{\beta}_{i}[t], \hat{N}_{i}[t]$ are the MLE using data up to time $t$.
%These errors are unreasonably large until the peak actually occurs.
At 66\% of time to peak\footnote{For reference, the median peak time is 20 weeks.}, around 50\% of instances predict peak infections with $>50\%$ error. By the time the peak actually occurs, around 40\% of instances still suffer prediction error in this range. Errors then drop off sharply after this point.

For comparison, let $\Tilde{\beta}_{i}[t]$ be the solution to the MLE problem \eqref{eq:llh} {\it fixing $N = \hat{N}_i[T_i]$}; that is, the MLE for $\beta$ if we knew the ground-truth value of $N$ a priori.  The dashed line in Figure~\ref{fig:peak-errors} shows errors for the peak estimate $I^{*}(\Tilde{\beta}_{i}[t], \hat{N}_{i}[T_{i}])$.  Errors for this estimator drop off much more quickly, with almost $90\%$ of instances achieving $<50\%$ error by $66\%$ of time to peak. This bears out the predictions of Theorem~\ref{thm:estimation-beta-gamma} that once $N$ is known, the remaining parameters of the SIR process are easy to estimate.

\begin{figure}[h]
\begin{center}
\includegraphics[scale=0.5]{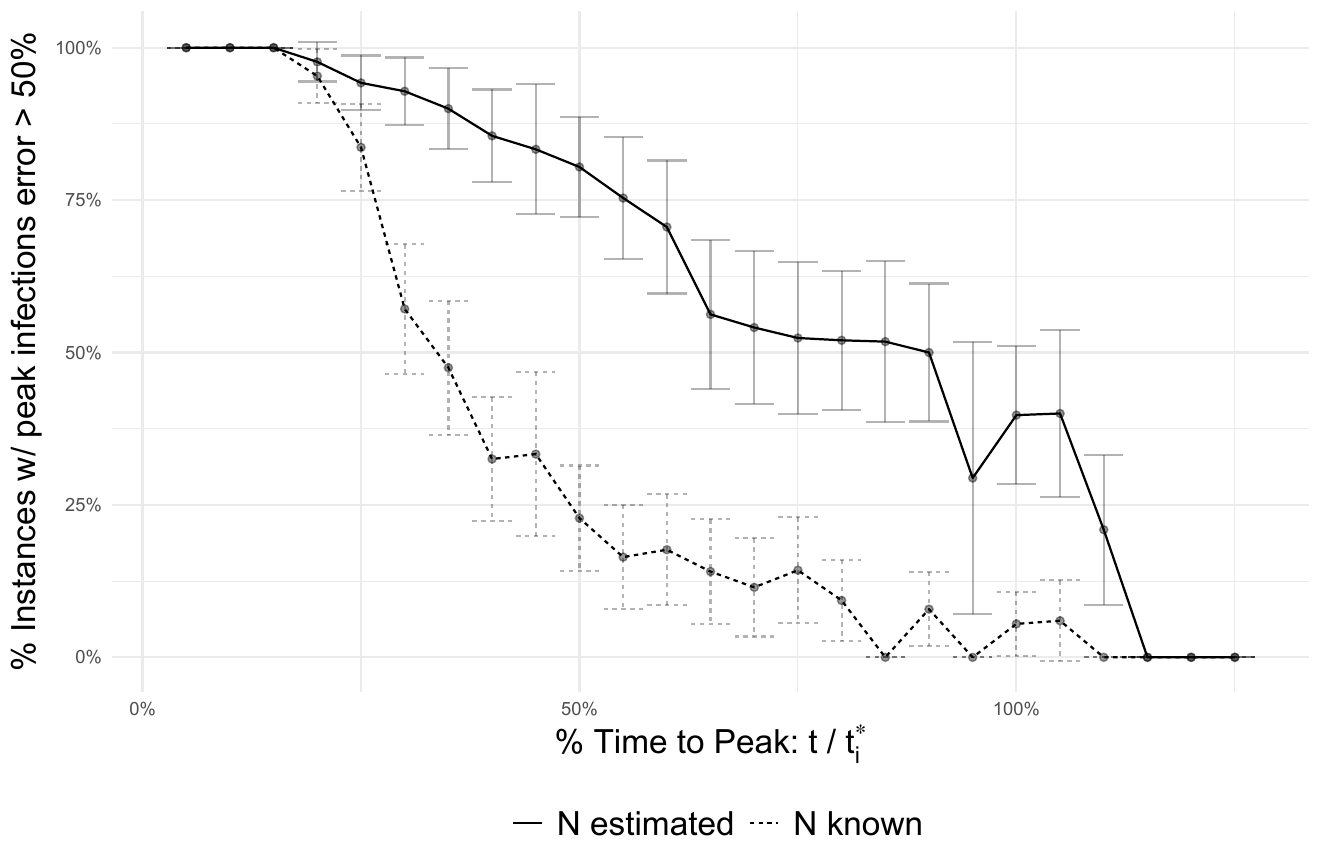}
\caption{\%  of instances with with relative prediction error $\left\lvert \hat{I}^{*}_i(\hat{\beta}, \hat{N}) - I^{*}_{i}\right\rvert / I^{*}_{i} > 0.5$, vs. \% of time to peak on the ILINet dataset. For `N estimated', we evaluate the MLE at time $t$ $\hat{I}^{*}(\hat{\beta}_{i}[t], \hat{N}_{i}[t])$. Errors for this estimator remain unreasonably large until around the peak occurs -- after which it drops dramatically. For `$N$ known', we evaluate the estimator $\hat{I}^{*}(\Tilde{\beta}_{i}[t], \hat{N}_{i}[T_{i}])$; that is, we assume $N$ known and estimate $\beta$ via MLE. Here, most instances estimate the peak accurately after 66\% of time to peak, reflecting the ease of estimating $\beta$ given $N$.}
\label{fig:peak-errors}
\end{center}
\end{figure}

\paragraph{Time to reach a lower bound of $\mathrm{RelError}(\hat{N}) = o(1)$.}
Finally, to understand the implications of our theory at a practical scale, we consider how long it takes to achieve a sample size large enough for Theorem~\ref{co:cr} to admit an error of $\mathrm{RelError}(\hat{N}) = o(1)$. In other words, how long (in real time) it takes to achieve roughly $N^{2/3}$ samples. 
We consider here the Amazon dataset under the Bass model. We first estimate $N$ for each product via MLE on all observations; use this estimated $N$ to determine the minimum required number of observations; then determine from the dataset how long it takes to reach that many observations. Figure~\ref{fig:time-to-N23} shows that this time is extremely long in practice: greater than 6 months for about 75\% of products; and greater than one year for 50\% of products. Notably, this is roughly the same time to sell $N/4$ units --- a constant fraction of all the units that will ever be sold.

\begin{figure}[htbp]
  \centering
  \includegraphics[width=0.7\linewidth]{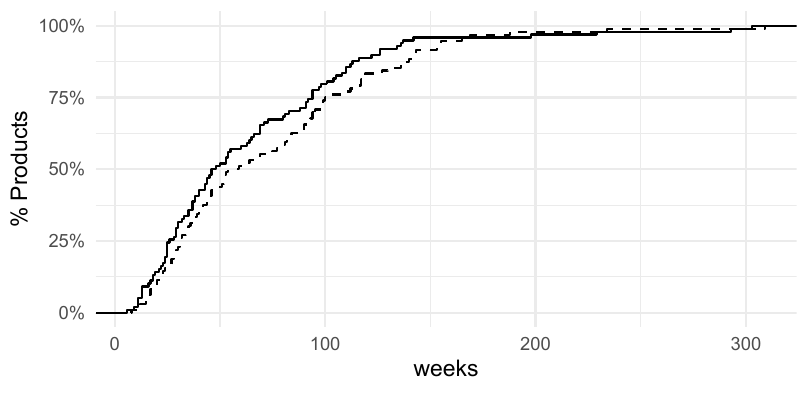}
  \caption{The solid line shows the cumulative distribution of the time needed to achieve a RelError lower bound of $o(1)$ (as given by Theorem~\ref{co:cr}), for products in the Amazon dataset. This is greater than 6 months for about 75\% of products; and greater than one year for 50\% of products. The dashed line shows the cumulative distribution of the time need to sell $N/4$ products; these distributions are comparable.}
  \label{fig:time-to-N23}
\end{figure}

\subsection{Heterogeneous Mixing Subpopulations} \label{sec:hetergeneous_mixing}
We now consider a variant of the SIR model in which the population is partitioned into groups based on their mixing rates. The goal is to determine whether our main results hold under a more complex SIR variant that better captures real-world population dynamics.

We follow \cite{delvalle2007mixing}, which defines subpopulations as age groups $[0, 5), [5, 10) \ldots [70, \infty)$. They provide mixing rates within and between age groups based on contact data collected in each of 131 countries. The paper also provides populations per age group per country, and an estimate of the `clinical fraction', i.e., the proportion of each age group which will present with clinically significant symptoms.

For each country, we use this data to construct a semi-synthetic instance with realistic subpopulations and mixing conditions. We assume for simplicity that the clinical fraction represents the `true' infected proportion of each subpopulation. Let $N_i$ be the total susceptible population in each subpopulation (i.e., population of the subpopulation, times clinical fraction) and let $N = \sum_{i=1}^{n} N_i$ be the total susceptible population across subpopulations. We will further simplify by assuming that the mixing rates $B_{ij}$, recovery rates $\gamma_i$ and population ratios $N_i / N$ are all known; but the decision maker must still estimate the overall scale $N$. We will show empirically that our results continue to hold in this more realistic setting -- even though there remains only one parameter to estimate.

The specific compartmental model we use is a stochastic generalization of \cite{delvalle2013mathematical}, following a jump process model analogous to that in our existing results. At each jump $k \in \mathbb{N}$, we observe either an infection or recovery in some subpopulation, as well as the time between events $k-1$ and $k$. Let $B \in \mathbb{R}^{n \times n}$ be the matrix of transmission rates between subpopulations,  and let $\gamma \in \mathbb{R}^{n}$ be the recovery rate within each subpopulation. Define the infection and recovery rates for subpopulation $i$  after the $k^{\rm th}$ event as follows: 
\begin{align*}
\lambda_{ik}^{I} &= \sum_{j \in [n]} B_{ij} \frac{S_{ik} I_{jk}}{ N_{j}} \\
\lambda_{ik}^{R} &= \gamma_{i} I_{ik}
\end{align*}
Then, let $T_{t}$ denote the duration between events $k-1$ and $k$:
\begin{align*}
  T_{k+1} &\sim {\rm Exp}\left( \sum_{i=1}^{n} (\lambda_{ik}^{I} + \lambda_{ik}^R) \right)
\end{align*}
State dynamics are as follows. There are $2 n$ mutually exclusive events possible each epoch: Infection in subpopulation $i$, and Recovery in subpopulation $i$, for each $i \in [n]$. On event Infection $i$, we have $S_{i, k+1} = S_{ik} - 1$, and $I_{i, k+1} = I_{ik} + 1$. All other state dimensions unchanged. On event Recovery $i$, we have $R_{i, k+1} = R_{i, k+1} + 1$ and $I_{i, k+1} = I_{ik} - 1$. All other state dimensions unchanged. Conditional on history, these have probabilities

\begin{align*}
P(\text{Infection } i) &= \lambda_{ik}^{I} / \sum_{i=1}^{N} (\lambda_{ik}^{I} + \lambda_{ik}^{R}) \\
P(\text{Recovery } i) &= \lambda_{ik}^{R} / \sum_{i=1}^{N} (\lambda_{ik}^{I} + \lambda_{ik}^{R})
\end{align*}

Figure~\ref{fig:average-mixing-results} shows the error in estimating $N$ as a function of $m$, the number of events observed, averaged over all 131 countries. Here we see that, precisely as the theory predicts, the relative error in estimating $N$ is significantly larger than 1 until $m > N^{2/3}$. After this point, the error drops precipitously, and precise estimation of $N$ becomes possible. Figure~\ref{fig:mixing-sir-by-country} shows results by country for 20 randomly selected countries, demonstrating that this holds not only in aggregate, but also for each individual problem instance.

\begin{figure}[h]
    \centering
    \includegraphics[width=0.6\linewidth]{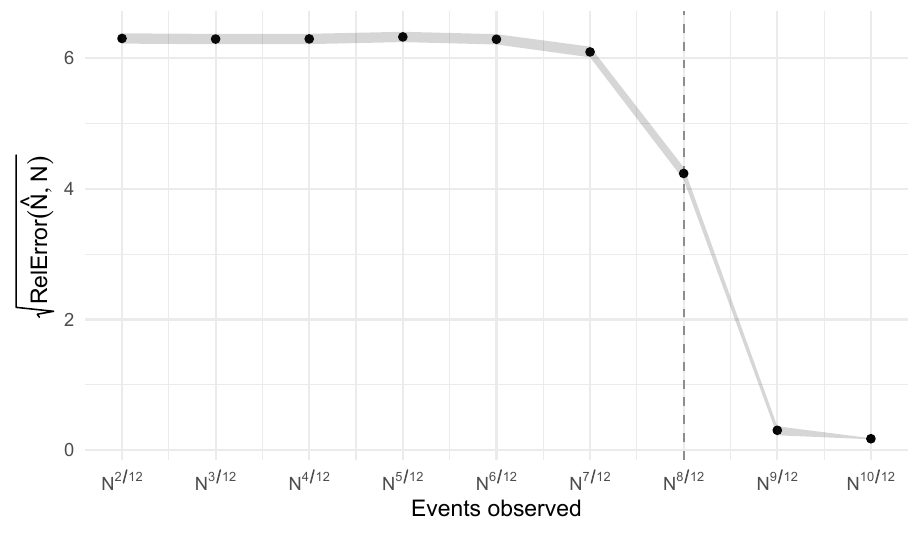}
    \caption{Error in estimating $N$ as a function of $m$, the number of events observed, averaged over all 131 countries. Error bars show standard errors over 50 random seeds.${\rm RelError}(\hat{N}, N)$ is significantly above 1 until $m > N^{2 / 3}$, after which the error drops sharply.} 
    \label{fig:average-mixing-results}
\end{figure}

\begin{figure}[h]
    \centering
    \includegraphics[width=0.6\linewidth]{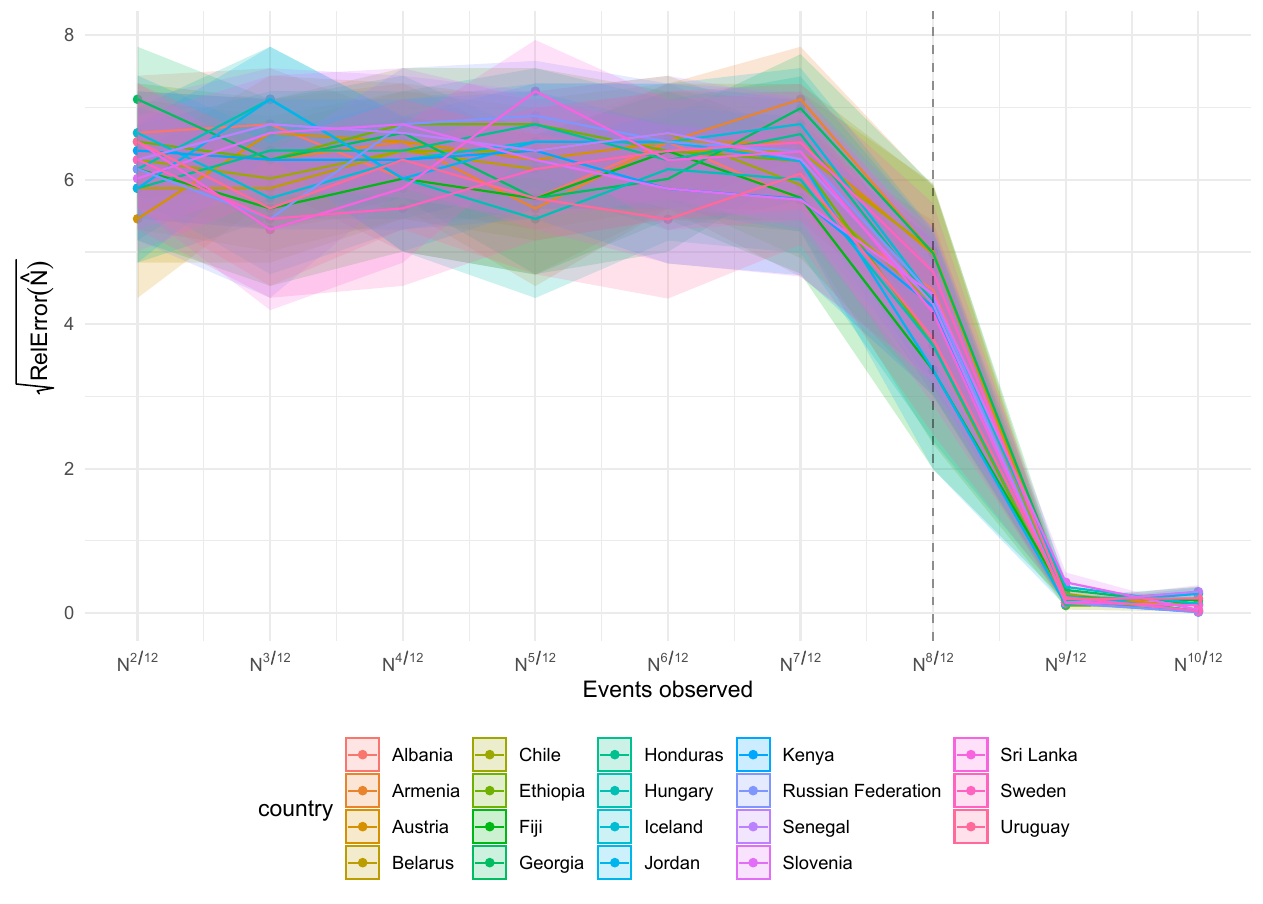}
    \caption{Error in estimating $N$ as a function of $m$, the number of events observed, for 20 randomly chosen countries. Error bars show standard errors over 50 random seeds. In all instances, ${\rm RelError}(\hat{N}, N)$ is significantly above 1 until $m > N^{2 / 3}$, after which the error drops sharply.} 
    \label{fig:mixing-sir-by-country}
\end{figure}

\subsection{Working Around the Limits to Learning in the COVID-19 Pandemic}
\label{sec:exp:covid19}

% The model was deployed in such an operational fashion in a large US state, where it was used to proactively place hospital resources (mobile surge capacity) in areas where we anticipated large peaks in infections.
% Our theory demonstrates that if we are to produce a valuable forecast early in a diffusion, we must rely on an estimator that places an informative bias on the effective population parameter, $N$. 

One approach to bypass the limits of learning is to rely on an estimator that places an informative prior on the effective population parameter, $N$. 
Here we briefly describe a heuristic that used this idea, which was used to produce one of the first broadly available county-level forecasts for COVID-19.

As above, we would like to forecast infections for a set of regions $i \in \mathcal{I}$. Recall that the effective population $N_i$ for region $i$ is the product of the actual population of the region (which is obviously known) and the fraction of infections that are actually observed (which is not). To arrive at a useful bias for $N_i$, we exploit hetereogeneity in the timing of infections in each region and use a `two-stage estimation' method. Specifically, infections start at different times in each region, and we typically have access to some set $Q[t] \subseteq \mathcal{I}$ of regions that have \textit{already experienced enough infections to reliably estimate} $N_{i}$ for $i \in Q[t]$. At a high level, our strategy will be to identify the set $Q[t]$, estimate $N_{i}$ for $i \in Q[t]$, then extrapolate these estimates (e.g., via matching on region-level covariates) to obtain $N_{i}$ for $i \notin Q[t]$. 
We describe this methodology briefly in the following section, as well as in detail in Appendix~\ref{sec:prior}.

\subsubsection{Two-Stage Estimation.}
Letting $P_i$ be the known population of region $i$, we parameterize the effective population as $N_i(\phi, \delta) = \exp(\phi ^\top Z_{i} + \delta_{i}) P_{i}$, where $Z_{i}$ are non-time-varying, region-specific covariates, $P_i$ is the population of region $i$, $\phi$ is a vector of fixed effects, and $\delta_{i} \sim \mathcal{N}(0, \sigma^{2}_{\delta})$ are region-specific random effects.

Next, we also incorporate demographic and mobility factors into the model, which influence the reproduction rate.
We define $X_i[t]$ as a set of time-varying covariates for region $i$, which represent both demographic features of the region as well as dynamic mobility features that represent the amount of movement of people in the county (leveraging anonymized location data). Then, we write $\beta_{i}[t]$ as a mixed effects model incorporating covariates $\beta_{i}[t] = \exp( X_i[t] ^\top \theta) + \epsilon_{i}$, where $\theta$ is a vector of fixed effects, and $\epsilon_{i} \sim \mathcal{N}(0, \sigma^{2}_{\epsilon})$ is a vector of random effects.
Lastly, we set $\gamma=1/4$ to be constant.

Given observations up to time $t$, we define the set $Q[t]$ to be the regions that have passed their peak rate of new infections.
Then, we estimate the model parameters $(\theta, \phi, \delta, \epsilon)$ in two stages:
\begin{enumerate}
  \item Estimate the peak parameters $\hat{\phi}, \hat{\delta}$ via MLE, for the regions $i \in Q[t]$.
  Set $\delta_i = 0$ for all $i \notin Q[t]$.
  \item Estimate the remaining parameters over all regions.
\end{enumerate}
Essentially, we use the regions in $Q[t]$ (the regions whose infection rate is passed its peak) in the first stage to learn the parameter $\phi$, which is a shared parameter for all regions that is used to determine $N_i$.
This estimated value for $\phi$ is used for determining $N_i$ for all $i \notin Q[t]$, which represent the regions that are in the earlier stages of the epidemic.

We compare the performance of the above approach to a naive `one-stage' approach (which we call \textit{MLE} in the next section), which simply estimates all parameters jointly.

\subsubsection{Experimental results}

We show the results of applying this methodology for forecasting in the COVID-19 pandemic. Our dataset consists of daily cumulative COVID-19 infections $C_{i}[t]$ at the level of sub-state regions $i \in \mathcal{I}$, from March to May 2020. The dataset also includes a rich set of covariates for each region, which we use to extrapolate the fits $N_{i} : i \in P[t]$ to other regions.

We compare the effectiveness of our heuristic (dubbed {\em Two-Stage}) to two extremes: {\em MLE} simply applies an approximate version of the MLE (the maximum likelihood problem here is substantially harder due to the recovery process being latent) to the data available and {\em Idealized} cheats by using a value of $N_i$ learned by looking into the future.
%perform an ablation study focused on the impact of assumed structure on the parameter $N$. This allows us to understand whether and why the Two-Stage approach improves estimation of $N$.
%
%\paragraph{The impact of learning $N$:} We compare four approaches to learning $N$:
%{\it MLE}: Learn $N$ for each region via an approximation to the MLE;
%{\it Two-Stage}: The two-stage approach specified above; and
%{\it Idealized}: Using a value of $N$ computed in-sample (i.e. by looking into the future) via \eqref{eq:llh}.
Figure~\ref{fig:compare-priors} shows weighted mean absolute percentage error (WMAPE) over regions, with weights proportional to infections on the last day in our dataset (May 21, 2020), for two metrics relevant to decision making: cumulative infections by May 21, 2020 and maximum daily infections, for regions that have peaked by May 21, 2020. Model vintages vary along the x-axis so that moving from left to right models are trained on an increasing amount of data.

\begin{figure}[htbp]
  \centering
  \includegraphics[scale = 0.75]{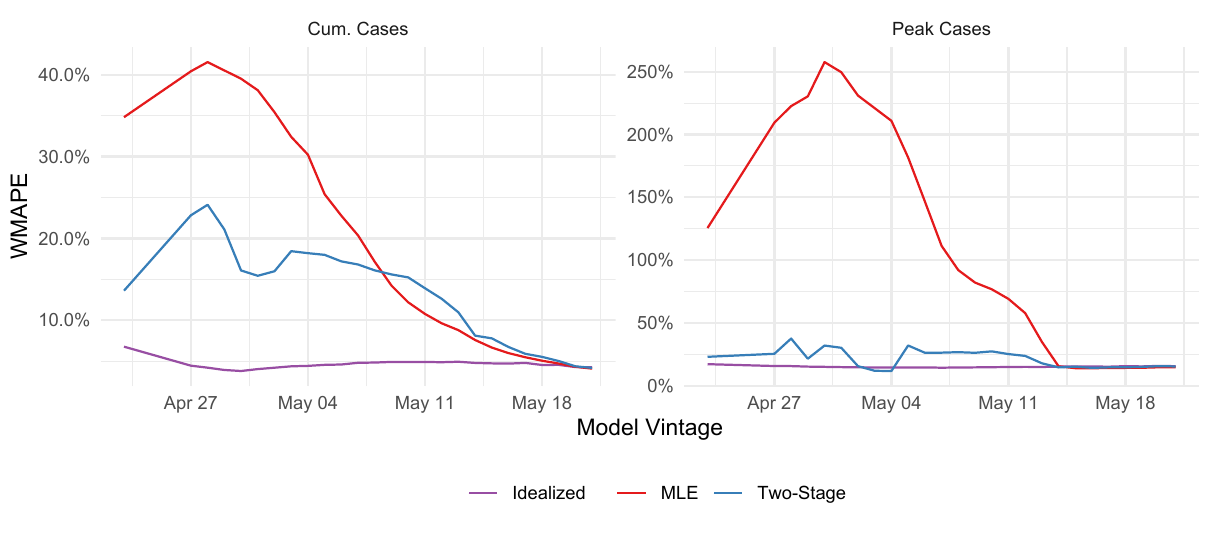}
  \caption{Prediction errors by model vintage, for regions that have peaked by May 21, 2020. Colors denote different approaches to learning $N_{i}$. }
  \label{fig:compare-priors}
\end{figure}

At one extreme, {\it Idealized} exhibits consistently low error even for early model vintages. This bears out the prediction of Theorem~\ref{thm:estimation-beta-gamma}: given $N$, $\beta$ is easy to learn even early in the infection with few samples. % It also illustrates that the SIR model itself provides a good fit to the overall trajectory of the epidemic (but that this requires careful calibration of $\alpha$).
{\it MLE} performs poorly until close to the target date of May 21 at which point sufficient data is available to learn $N$. This empirically illustrates the difficulty of learning $N$, as described in~\cref{co:cr}. Finally, we see that {\it Two-Stage} significantly outperforms {\it MLE} far away from the test date. Close to the test date the two approaches are comparable. For maximum daily infections, {\it MLE} drastically underperforms {\it Two-Stage} far from the test date. Our approach to learning from peaked regions significantly mitigates the difficulty of learning $N$. Further details on this study can be found in Appendix~\ref{sec:prior}.

%!TEX root=../SIR-model-MS.tex

\section{Conclusion} % (fold)
\label{sec:conclusion}

In this paper, we have shown fundamental limits to learning for the SIR and Bass models, two models that often serve as building blocks for epidemic and product adoption modeling.
By establishing sample complexity lower bounds, we have demonstrated the challenge of achieving early accurate forecasting due to the due to the time required to collect a sufficient number of observations. Moreover, our analysis extends to decision-making scenarios involving costly interventions to mitigate further infections, where we establish a lower bound on regret.

These findings highlight the difficulty of accurate forecasting based solely on infection trajectories and emphasize the need to incorporate additional data sources. 
We illustrate the potential benefits of seroprevalance testing, showing that even a sublinear-sized test, after a sufficient number of diffusion process samples, can significantly improve the estimation of $N$. Additionally, we introduce a heuristic approach employed in COVID-19 forecasting that biases the estimation of $N$ using a prior, leveraging regional infection timing heterogeneity.
Going forward, we believe that the development of new methods that can effectively overcome the established lower bounds represent a valuable avenue for future research.

\bibliographystyle{informs2014}
\bibliography{ref}

\newpage
\appendix
\setstretch{1.1}

Appendices~\ref{sec:app:cr_proof}, \ref{sec:app:estimator_proof_bass} and \ref{sec:app:estimator_proof} contain the proofs of Theorems~\ref{thm:fi}, \ref{thm:bass-estimation-p-a} and \ref{thm:estimation-beta-gamma} respectively.
Appendix~\ref{sec:app:prop_proofs} contains the proofs of Propositions~\ref{prop:constant_fraction}, \ref{prop:identity_fluid}, \ref{prop:time_to_peak_bass} and \ref{thm:deterministic-time}, each in their own subsections.
% Appendix~\ref{sec:app:peak_condition} contains a result justfying the use of \eqref{eq:peaks} as a sufficient condition in reliably estimating $\alpha$.
% Appendix~\ref{sec:app:covariates} contains a description of the covariates used in the practical SIR model.
\cref{sec:datasets} provides details on the datasets used in Section~\ref{sec:experiments}, and \cref{sec:prior} contains a detailed description of the COVID-19 forecasting model from \cref{sec:exp:covid19}.
% ---------------
% Main CR result
% ---------------
% \iffalse
\section{Proof of Theorem~\ref{thm:fi}} \label{sec:app:cr_proof}
%!TEX root=../SIR-model-EC.tex
We finish the sections of the proof that were not included in the main paper.
This includes the proof of Lemma~\ref{lem:E_k}, Lemma~\ref{lem:highprob}, calcuations for Lemma~\ref{lemma:fi}, and details regarding the final step of the proof.

We define $\lambda(N, k-1 , C_{k-1}) = \left(\frac{\beta(N-C_{k-1})}{ N} + \gamma\right)I_{k-1}$ and $\eta(N, C_{k-1}) = \frac{\beta (  N -C_{k-1}) }{\beta (  N -C_{k-1})+   N\gamma}$. Thus,  for $ k \leq \tau$, $\lambda(N, k-1, C_{k-1})$ is the mean of the $k$-th inter-arrival time and $\eta(N,  C_{k-1})$ is the probability that the arrival in the $k$-th instance is a new infection rather than a recovery.

% \begin{lemma}
% 	Define $r_k \triangleq \frac{I_0 + k + 2R_0}{2}$ for all $k \geq 0$.
% 	For all $k$, $E_k = \bI\{C_{k} > r_{k}\}$.
% 	Moreover, when $E_k = 1$, $I_k = 2C_k - k - I_0 - 2R_0 > 0$.
% 	% Moreover, $C_k$ determines $E_k$.
% \end{lemma}
\subsection{Proof of Lemma~\ref{lem:E_k}}
\begin{myproof} % [Proof of Lemma \ref{lem:E_k}]
	Suppose $k < \tau$ i.e $E_k = 1$. Then, $k$ is equal to total number of jumps that have occurred so far (the number of movements from S to I and from I to R).
	The number of individuals that have moved from S to I is $C_k - I_0 - R_0$, and the number of movements from I to R is $C_k - I_k - R_0$.
	Therefore, $k = 2C_k - I_0 - I_k - 2R_0$.
	Since $I_k > 0$, $C_k > r_k$.

	Suppose $k \geq  \tau$ i.e $E_k =0 $. Then, $k$ is greater than or equal to the total number of jumps, which is still equal to $2C_k - I_0 - I_k - 2R_0$.
	Hence $C_k \leq r_k$ in this case.

	%	Moreover, $C_k = r_k$ if and only if $I_k = 0$ (in which case $\tau = k$).

\end{myproof}

% \begin{lemma} \label{lem:highprob}
% 	There exists a constant $D$ that only depends on $\beta$ and $\gamma$ such that if $\frac{\beta (P- m-C_0)}{\beta (P - m - C_0) + P\gamma} > p$  and $I_0 \geq D$,
% 	then $\Pr(E_{m}=1) \geq \frac{1}{2}$.
% \end{lemma}

\subsection{Proof of Lemma~\ref{lem:highprob}}
\begin{myproof}% [Proof of Lemma~\ref{lem:highprob}]

Let $X_k \overset{iid}{\sim} \Bern(p)$ for $k = 1, 2, \dots$.
Let $\{A_k: k \geq 0\}$ be a stochastic process defined by:
\begin{align*} 
A_k &= \begin{cases}
          C_0 & \text{if }k = 0 \\
          C_0 + X_1 + \dots + X_k & \text{if $A_i > r_i \; \forall i < k$} \\
          A_{k-1} & \text{otherwise.} \\
       \end{cases}
\end{align*}

Let $\tau_A = \min\{k : A_k \leq r_k\}$ be the ``stopping time'' of this process.
% Because of the condition $\frac{\beta(P-m-C_0)}{\beta(P-m-C_0) + P\gamma} > p$, for $k \leq m$ and $k \leq \tau$, $C_k - C_{k-1} \sim \Bern(q)$ for $q > p$.
% We say that a process has ``stopped'' when the value is 0.
% $\tC_k$ is a modified version of $C_k$ where $\tC_k$ is 0 after the stopping time.
% Note that the two processes $A_k$ and $\tC_k$ have the same stopping condition --- the first time the $i$th term in the process is less than or equal to $r_i$.

\begin{claim}   \label{cl:stoch_dom}
$\Pr(\tau \leq m) \leq \Pr(\tau_A \leq m)$.
\end{claim}

The proof of this claim involves showing the process $\{A_k\}$ is stochastically less than $\{C_k\}$; the proof can be found in Section~\ref{sec:proofstochdom}.
We now upper bound $\Pr(\tau_A \leq m)$.
$\tau_A \leq m$ if and only if $A_k \leq r_k$ for some $k \leq m$.
Before this happens, $A_k = C_0 + X_1 + \dots + X_k$.
Therefore, if $\tau_A \leq m$, it must be that $C_0 + X_1 + \dots + X_k \leq \frac{k+I_0+2R_0}{2}$ for some $k \leq m$.
\begin{align*}
\Pr(\tau_A \leq m)
&\leq  \sum_{k = 1}^m \Pr\left(C_0 + X_1 + \dots + X_k \leq \frac{k+I_0 + 2R_0}{2}\right) \\
&=  \sum_{k = 1}^m \Pr\left(X_1 + \dots + X_k < pk\left( 1- \frac{2pk - k+I_0}{2pk}\right)\right).
\end{align*}
Since $\bE[X_1 + \dots + X_k] = pk$, using the Chernoff bound (multiplicative form: $\Pr(\sum_{i = 1}^k X_i \leq (1-\delta)\mu) \leq \exp(-\delta^2 \mu/2)$) gives
\begin{align}
\Pr(\tau_A \leq m)
% &\leq  \sum_{k = 0}^m \exp\left(-\left( \frac{2pk - k+I_0}{2pk}\right)^2 \frac{pk}{2} \right) \\
&\leq  \sum_{k = 1}^m \exp\left(- \frac{pk}{2} \left(\left(1-\frac{1}{2p}\right) + \frac{I_0}{2pk}\right)^2 \right) \nonumber \\
&=  \sum_{k = 1}^m \exp\left(- \frac{pk}{2} \left(1-\frac{1}{2p}\right)^2  - \frac{I_0}{2} \left(1-\frac{1}{2p}\right)   -   \frac{I_0^2}{8pk} \right) \nonumber \\
& \leq   \sum_{k = 1}^m \exp\left(- \frac{pk}{2} \left(1-\frac{1}{2p}\right)^2  - \frac{I_0}{2} \left(1-\frac{1}{2p}\right)   \right) \nonumber \\
% &\leq  \exp\left( - \frac{I_0}{4} \left(1-\frac{1}{2p}\right) - \frac{I_0}{8p}\right) \sum_{k = 0}^m \exp\left(- \frac{pk}{2} \left(1-\frac{1}{2p}\right)^2  \right) \\
&\leq  \exp\left(-\left(\frac{1}{2} - \frac{1}{4p} \right)I_0\right)\sum_{k = 1}^m \exp\left(- \frac{pk}{2} \left(1-\frac{1}{2p}\right)^2  \right)\nonumber  \\
&\leq   C_1\exp(-C_2 I_0),  \label{eq:a_m_1}
\end{align}
for constants $C_1 = \sum_{k = 1}^\infty \exp\left(- \frac{pk}{2} \left(1-\frac{1}{2p}\right)^2  \right), C_2 = \frac{1}{2} - \frac{1}{4p} > 0$.
($C_1$ is a constant since it is a geometric series with a ratio smaller than 1, since $p > 1/2$.)
Let $D$ be the solution to $C_1\exp(-C_2 D) = \frac{1}{2}$.
Then, if $I_0 \geq D$, $\Pr(E_{m}) = 1-\Pr(\tau \leq m) \geq 1-\Pr(\tau_A \leq m) \geq \frac{1}{2}$.

\end{myproof}

\subsubsection{Proof of Claim~\ref{cl:stoch_dom}.} \label{sec:proofstochdom}

\begin{definition}
For scalar random variables $X, Y$, we say that $X$ is \textit{stochastically less than} $Y$ (written $X \leq_{st} Y$) if for all $t \in \bR$,
\begin{align*}
\Pr(X > t) \leq \Pr(Y > t).
\end{align*}
For random vectors $X, Y \in \bR^n$ we say that $X \leq_{st} Y$ if for all increasing functions $\phi: \bR^n \rightarrow \bR$,
\begin{align*}
\phi(X_1, \dots, X_n) \leq_{st} \phi(Y_1, \dots, Y_n).
\end{align*}
\end{definition}

We make use of the following known result for establishing stochastic order for stochastic processes.
\begin{theorem}[Veinott 1965] \label{thm:veinott}
Suppose $X_1, \dots, X_n$, $Y_1, \dots, Y_n$ are random variables such that $X_1 \leq_{st} Y_1$ and
for any $x \leq y$,
\begin{align*}
(X_k | X_1 = x_1, \dots, X_{k-1}=x_{k-1})
\leq_{st}
(Y_k | Y_1 = y_1, \dots, Y_{k-1}=y_{k-1})
\end{align*}
for every $2 \leq k \leq n$.
Then, $(X_1, \dots, X_n) \leq_{st} (Y_1, \dots, Y_n)$.
\end{theorem}

\begin{myproof}[Proof of Claim~\ref{cl:stoch_dom}]
Because of the condition $\frac{\beta(N-m-C_0)}{\beta(N-m-C_0) + N\gamma} > p$, for $k \leq m$ and $k \leq \tau$, $C_k - C_{k-1} \sim \Bern(q)$ for $q > p$.
First, we show $(A_0, A_1, \dots, A_m) \leq_{st} (C_0, C_1, \dots, C_m)$ using Theorem~\ref{thm:veinott}.
$C_0 \leq_{st} A_0$ since $C_0 = A_0 = I_0$.
We condition on $A_{k-1} = x$ and $C_{k-1} = y$ for $x \leq y$, and we must show $A_k \leq_{st} C_k$. (We do not need to condition on all past variables since the both processes are Markov.)
If $x \leq r_{k-1}$, then $A_{k} = A_{k-1} = x \leq y = C_{k-1} \leq C_k$.
Otherwise, the process $A_k$ has not stopped, and neither has $C_k$ since $y \geq x$.
Then, $A_{k} \sim x + \Bern(p)$ and $C_{k} \sim y + \Bern(q)$ for some $q \geq p$.
Clearly, $A_{k} \leq_{st} C_k$ in this case.
We apply Theorem~\ref{thm:veinott}, which implies $A_m \leq_{st} C_m$.

Define the function $u:\bR^{m+1} \rightarrow \{0, 1\}$, $u(x_0, x_1, \dots, x_m) =\bI\{\cup_{k = 1}^m \{x_k \leq r_k\}\}$.
Then, $u(A_0, A_1, \dots, A_m) = 1$ if and only if $\tau_A \leq m$, and $u(C_0, C_1, \dots, C_m) = 1$ if and only if $\tau \leq m$.
$u$ is a decreasing function. Therefore, $u(A_0, A_1, \dots, A_m) \geq_{st} u(C_0, C_1, \dots, C_m)$.
Then, $\Pr(\tau \leq m) = \Pr(u(C_0, C_1, \dots, C_m) \geq 1) \leq \Pr(u(A_0, A_1, \dots, A_m) \geq 1) = \Pr(\tau_A \leq m)$ as desired.
\end{myproof}

\subsection{Calculations for Lemma~\ref{lemma:fi}} \label{sec:app:cr_calculations}

We define $\lambda(N, k-1 , C_{k-1}) = \left(\frac{\beta(N-C_{k-1})}{ N} + \gamma\right)I_{k-1}$ and $\eta(N, C_{k-1}) = \frac{\beta (  N -C_{k-1}) }{\beta (  N -C_{k-1})+   N\gamma}$. Thus,  for $ k \leq \tau$, $\lambda(N, k-1, C_{k-1})$ is the mean of the $k$-th inter-arrival time and $\eta(N,  C_{k-1})$ is the probability that the arrival in the $k$-th instance is a new infection rather than a recovery.

\textbf{Derivation of $\bE_{C_k}[g_{C_k|C_{k-1}}(C_k,C_{k-1}, N) | E_{k-1} =1]$.}
When $E_{k-1} = 1$, we have $C_k \sim C_{k-1} +  \Bern(\eta(N,  C_{k-1})) $.
Therefore, $ \bE_{C_k}[g_{C_k|C_{k-1}}(C_k,C_{k-1}, N) | E_{k-1} =1] = \cJ_{C_k \sim \Bern(\eta(N,  C_{k-1}))} (N)$.
We reparameterize to write the Fisher information as:
\begin{align*}
\bE_{C_k}[g_{C_k|C_{k-1}}(C_k,C_{k-1}, N) | E_{k-1} =1]
&= \cJ_{C_k \sim \Bern(\eta)}(\eta) \left( \partN \eta(N,  C_{k-1}) \right)^2 \\
&= \frac{1}{\eta(1-\eta)} \left( \partN \eta(N,  C_{k-1}) \right)^2.
\end{align*}

Use $\eta(N,  C_{k-1} ) = \frac{\beta ( N - C_{k-1} ) }{\beta ( N -C_{k-1})+  N\gamma}$ to derive
\begin{align*}
\partN \eta(N,  C_{k-1} )  &= \frac{\beta (\beta( N- C_{k-1} ) + \gamma  N) - \beta ( N- C_{k-1} )(\beta +\gamma)}{(\beta (  N- C_{k-1} ) + \gamma  N)^2} \\
&= \frac{\beta \gamma C_{k-1}}{(\beta( N- C_{k-1} ) + \gamma  N)^2}.
\end{align*}

Also, $\frac{1}{\eta (1-\eta)} = \frac{(\beta(N- C_{k-1} ) + N\gamma)^2}{(N- C_{k-1} )\beta N\gamma} $.

Substituting,

\begin{align*}
\bE_{C_k}[g_{C_k|C_{k-1}}(C_k,C_{k-1}, N) | E_{k-1} =1]
&= \frac{(\beta(N- C_{k-1} ) + N\gamma)^2}{(N- C_{k-1} )\beta N\gamma} \left(\frac{\beta \gamma C_{k-1}}{(\beta( N- C_{k-1} ) + \gamma  N)^2}\right)^2 \\
&= \frac{\beta \gamma C_{k-1}^2}{(N- C_{k-1} )N(\beta( N- C_{k-1} ) + \gamma  N)^2}  \\
\end{align*}

\noindent \textbf{Derivation of $\bE_{T_k}[g_{T_k|C_{k-1}}(T_k,C_{k-1}, N) | E_{k-1} =1]$.}
Similarly, conditioned on $ E_{k-1} =1, T_k \sim \Exp(\lambda(N, k-1, C_{k-1}) )$.
Therefore, $\bE_{T_k}[g_{T_k|C_{k-1}}(T_k, C_{k-1} , N)] = \cJ_{T_k \sim \Exp(\lambda(N, k-1,  C_{k-1}))}(N)$.
We reparameterize to write
\begin{align*}
\bE_{T_k}[g_{T_k|C_{k-1}}(T_k,  C_{k-1} , N)]
&= \cJ_{T_k \sim \Exp(\lambda)}(\lambda) \left( \partN\lambda(N, k-1, C_{k-1}) \right)^2 \\
&= \frac{1}{\lambda^2} \left( \partN \lambda(N, k-1, C_{k-1} ) \right)^2.
\end{align*}
Use $\lambda(N, k-1, C_{k-1}) = (\frac{ \beta(N- C_{k-1} )}{N} + \gamma)(2C_{k-1} - (k-1) - I_0 - 2R_0)$ to derive
\begin{align*}
\partN \lambda(N, k-1, C_{k-1}) &= \frac{ \beta C_{k-1} (2C_{k-1} - (k-1) - I_0 - 2R_0)}{ N^2} \\
\frac{1}{\lambda(N, k-1, C_{k-1})} &= \frac{N}{(\beta(N- C_{k-1} ) + \gamma N)(2C_{k-1} - (k-1)-I_0 - 2R_0)}.
\end{align*}
Substituting,
\begin{align*}
\bE_{T_k}[g_{T_k|C_{k-1}}(T_k, C_{k-1} , N)]
&= \left(\frac{\beta  C_{k-1} }{N(\beta(N- C_{k-1} ) + \gamma N)}\right)^2 \\
\end{align*}

\noindent \textbf{Derivation of  $\cI_{O_m}(N)$.}
Using the expressions derived above for $\bE_{C_k}[g_{C_k|C_{k-1}}(C_k,C_{k-1}, N) | E_{k-1} =1]$ and \\ $\bE_{T_k}[g_{T_k|C_{k-1}}(T_k, C_{k-1} , N)] $, we get
\begin{align*}
&	\bE_{C_k}[g_{C_k|C_{k-1}}(C_k,C_{k-1}, N) | E_{k-1} =1] + \bE_{T_k}[g_{T_k|C_{k-1}}(T_k, C_{k-1} , N)]  \\
&=  \frac{\beta \gamma C_{k-1}^2}{(N- C_{k-1} )N(\beta( N- C_{k-1} ) + \gamma  N)^2}  + \left(\frac{\beta  C_{k-1} }{N(\beta(N- C_{k-1} ) + \gamma N)}\right)^2 \\
% &=  \frac{\beta \gamma C_{k-1}^2}{(N- C_{k-1} )N(\beta( N- C_{k-1} ) + \gamma  N)^2}  + \left(\frac{\beta  C_{k-1} }{N(\beta(N- C_{k-1} ) + \gamma N)}\right)^2 \\
&= \frac{C_{k-1}^2}{(N - C_{k-1})N^2 (N - C_{k-1}  + \frac{\gamma}{\beta}N)}
\end{align*}
Thus,
\begin{align*}
	\cI_{O_m}(N)
	&= \sum_{k=1}^m \bE[g_{C_k|C_{k-1}}(C_k,C_{k-1}, N) + g_{T_k|C_{k-1}}(T_k,C_{k-1}, N)| E_{k-1}=1] \Pr(E_{k-1}=1) \\
	&=\sum_{k=1}^{m} \bE\left[\frac{C_{k-1}^2}{(N - C_{k-1})N^2 (N - C_{k-1}  + \frac{\gamma}{\beta}N)} \; \bigg| \; E_{k-1} = 1 \right] \Pr(E_{k-1}=1) .
\end{align*}

\subsection{Details of Final Step of \cref{thm:fi}} \label{sec:app_cr_last_step}
Define $p \triangleq \frac{1}{2}(\frac{\beta}{\beta+\gamma} + \frac{1}{2}) > \frac{1}{2}$ as in \cref{lem:highprob}.
Assume $N$ is large enough so that $m + C_0 \leq \frac{N}{2}$ and $\frac{\beta (N - m - C_0)}{\beta (N - m - C_0) + P \gamma} > p$ (this is possible since $\frac{\beta}{\beta + \gamma} > p$ and $m = o(N)$).

For the upper bound, we have that $C_k \leq k + I_0 + R_0$ by definition.
Since $I_0, R_0 \leq m$ by assumption, $C_k \leq 3m$.
Moreover, by assumption, $C_k \leq m + C_0 \leq \frac{N}{2}$.
Plugging these into \eqref{eq:exactFI} results in
\begin{align*}
\cI_{O_m}(N)
&\leq \sum_{k=0}^{m-1} \Pr(E_{k-1}=1)  \frac{(3m)^2}{N^2 (N-\frac{1}{2} N)((N-\frac{1}{2} N)+ \frac{\gamma}{\beta} N)}
\leq H_1 \frac{m^3}{N^4},
\end{align*}
for a constant $H_1$.

Then, similarly to the upper bound, $\cI_{O_m}(N) \geq H_2 \frac{m^3}{N^4}$ follows from using $\Pr(E_{m}=1) \geq \frac{1}{2}$ and the fact that $C_k \geq \frac{k + I_0 + 2R_0}{2} \geq \frac{k}{2}$ when $E_k =1$ (Lemma~\ref{lem:E_k}):
\begin{align*}
\cI_{O_m}(N)
&\geq \sum_{k=0}^{m-1} \frac{1}{2} \frac{\left(\frac{k}{2}\right)^2}{N^4}
\geq H_2 \frac{m^3}{N^4},
\end{align*}
Combining the upper and lower bounds finish the proof.

\subsection{Generalization of Theorem \ref{co:cr}: a Finite-sample Result}\label{sec:generation-finite-sample}
Note that the proof of \cref{thm:fi} provides exact formulas for the Fisher information, where quantities such as $N$ and $m$ are finite. This leads to the following result of which \cref{co:cr} is a special case, and it holds for any initial conditions $I_0$ and $R_0.$
 \begin{theorem}
Consider any observation $(T_0, I_{0}, R_0, T_1, I_{1}, R_{1}, \dotsc, T_{m}, I_{m}, R_{m})$ from either a Bass model or an SIR model (with any initial $I_{0}$ and $R_{0}$). Suppose $\hat{N}$ is an un-biased estimator for $N$. Let $I_{\max} = \max_{1\leq i\leq m} I_{i}$ with $I_{\max} \leq c(p, \gamma, \beta) N.$ Then
\begin{align*}
    \mathbb{E}\left[\frac{(\hat{N}-N)^2}{N^2}\right] \geq C(p, \gamma, \beta) \frac{N^2}{I_{\max}^3}
\end{align*}
where $C(p, \gamma, \beta)$ and $c(p, \gamma, \beta)$ are constants that are explicit functions of $p, \gamma, \beta$.
\end{theorem}

% \fi

% ---------------------
% Beta/gamma estimator
% ---------------------
\section{Proof of Theorem~\ref{thm:bass-estimation-p-a}}
\label{sec:app:estimator_proof_bass}
%!TEX root=../SIR-model-EC.tex

\begin{myproof}
We construct the estimators $\hat{a}$ for $a$ and $\hat{\beta}$ for $\beta$ as the following. To begin, let
\begin{align*}
\hat{A} &:= \frac{\sum_{i=1}^{m/2} \min(T_{i}, T_{m-i})}{m / 2}\\
\hat{B} &:= \frac{\sum_{i=1}^{m/4} \min(T_{i}, T_{m/2-i})}{m/4}.
\end{align*}
We will show momentarily that $\hat{A}$ approximates $\frac{1}{2a + m\beta}$ and $\hat{B}$ approximates $\frac{1}{2a + (m/2) \beta}$. We then construct $\hat{a}$ and $\hat{\beta}$:
\begin{align*}
\hat{\beta} &:= \left(\frac{1}{\hat{A}} - \frac{1}{\hat{B}}\right) \frac{2}{m}\\
\hat{a} &:= \left(\frac{2}{\hat{B}} - \frac{1}{\hat{A}} \right) \frac{1}{2}.
\end{align*}

To start the proof, let us analyze $\hat{A}$. Let $A_i = \min(T_i, T_{m-i}), 1\leq i \leq \lfloor m/2 \rfloor$. By the property of independent exponential random variables, we have $A_i \sim \mathrm{Exp}(l_i)$
where
$$
l_i := (2a + \beta m)-(a + i\beta)\frac{i}{N} - (a + (m-i)\beta)\frac{m-i}{N}.
$$
Note that when $N \gg m$, we shall have $l_{i} \approx 2a + \beta m$, which is independent from $i$. This inspires us to use $\hat{A} := \frac{\sum_{i} A_{i}}{m/2}$ as an estimator for $\frac{1}{2a + \beta m}.$

More specifically, Let $\mu = \frac{\sum_{i} E[A_i]}{m/2}  = \frac{1}{m/2}\sum_{i} \frac{1}{l_i}$. Note that
\begin{align*}
\frac{N-m}{N} (2a + \beta m) \leq l_i \leq 2a+\beta m.
\end{align*}
Then, this implies that $\mu$ is close to $\frac{1}{2a + \beta m}$:
\begin{align}\label{eq:mu-bound}
\frac{1}{2a+m\beta} \leq \mu \leq \frac{N}{N-m}\frac{1}{2a+m\beta}
\end{align}
On the other hand, we can invoke the multiplicative Bernstein inequality (\cite{janson2018tail}) to obtain, with probability $1-O(1/N^2)$, 
\begin{align}\label{eq:berinstein}
(1-\delta)\mu \leq \hat{A} \leq \mu (1+\delta)
\end{align}
where $\delta := O(\sqrt{\log(N)/m}).$ Combining \cref{eq:mu-bound} and \cref{eq:berinstein}, we then have
\begin{align*}
\frac{1}{2a+m\beta} (1-\delta) \leq \hat{A} \leq \frac{N}{N-m} (1 + \delta) \frac{1}{2a+m\beta}.
\end{align*}
This further implies desired bounds for using $\frac{1}{\hat{A}}$ to estimate $2a + m\beta$:
\begin{align*}
\left|\frac{1}{\hat{A}} - (2a+m\beta) \right| 
&\lesssim (\delta + \frac{m}{N}) (2a+m\beta)\\
&\lesssim \sqrt{\frac{\log(N)}{m}} (2a + m \beta)
\end{align*}
where the last inequality uses $m=O(N^{2/3}\log^{1/3}(N))$ (hence $m/N \lesssim \delta$). 

A similar analysis can be conducted for $\hat{B}$, which implies that
\begin{align*}
\left|\frac{1}{\hat{B}} - (2a+(m/2)\beta) \right| 
&\lesssim (\delta + \frac{m}{N}) (2a+(m/2)\beta)\\
&\lesssim \sqrt{\frac{\log(N)}{m}} (2a + (m/2) \beta)
\end{align*}

Combining the bounds of $\hat{A}$ and $\hat{B}$, we then obtain the bounds for $\hat{a}$ and $\hat{\beta}$, which completes the proof\footnote{A further refinement can be performed for analyzing $\hat{a}$ by considering a set of estimators $\hat{S}_{k} = \sum_{i=1}^{k/2} \min(T_{i}, T_{k-i}$ that generalize $\hat{A}$ and $\hat{B}.$ We omit the details for simplicity.}. 
%
%In order to obtain the optimal estimators, we consider a union of estimators where $S_{k} = \sum_{i=1}^{k/2} \min(T_{i}, T_{k-i}).$ We find the estimator $\hat{\alpha}, \hat{\beta}$ such that for every $k$,
%\begin{align*}
%\left|\frac{k}{2S} - \left(2\hat{a} + k\hat{\beta}\right)\right| \leq C_1\left(\frac{k}{2S}\left(\frac{m}{N} + \sqrt{\frac{\log(N)}{m}}\right)\right).
%\end{align*}
%This will guarantee, with probability $1-O(1/N)$, 
%\begin{align*}
%|\hat{\beta}-\beta| &= O\left(a \sqrt{\frac{\log(N)}{m^3}} +  \beta \sqrt{\frac{\log(N)}{m}} \right)\\
%|\hat{a} - a| &= O\left(\min_{1\leq k\leq m} \left(a \sqrt{\frac{\log N}{k}} + \sqrt{k} \beta \sqrt{\log N}\right)\right) = O\left(a \sqrt{\frac{\log N}{m}} + \sqrt{\beta a} \sqrt{\log N}\right),
%\end{align*}
%as desired.
\end{myproof}

\section{Proof of Theorem~\ref{thm:estimation-beta-gamma}} \label{sec:app:estimator_proof}
%!TEX root=../SIR-model-EC.tex

\subsection{Construction of Estimators}\label{sec:explicit-construction}
Our construction of estimators $\hb$ for $\beta$ and $\hg$ for $\gamma$ is the following. To begin, let
\begin{align*}
\hat{A} := \frac{C_{m} - C_{0}}{m} \\
\hat{B} := \frac{\sum_{k=1}^{\min (m,\tau)} I_{k-1}T_{k}}{m}.
\end{align*}
We will show momentarily that $\hat{A}$ can be viewed as an estimator for $\frac{\beta}{\beta+\gamma}$ and $\hat{B}$ an estimator for $\frac{1}{\beta+\gamma}.$ Then given $\hat{A}$ and $\hat{B}$, we construct
\begin{align*}
\hb := \frac{\hat{A}}{\hat{B}}\\
\hg := \frac{1}{\hat{B}} - \hb.
\end{align*}
This construction leads to the guarantees stated in \cref{thm:estimation-beta-gamma}. 

%Fix an instance in which the assumptions of the theorem statement hold. 
%% Assume $\beta > \gamma$ is fixed.
%Let $p \triangleq \frac{1}{2}(\frac{\beta}{\beta+\gamma} + \frac{1}{2}) > \frac{1}{2}$.
%Let $\hat{A} =\frac{C_m-C_0}{m}$ be an estimator for $\frac{\beta}{\beta + \gamma}$, $\hat{B} = \frac{\tS_{m}}{m}$ be an estimator for $\frac{1}{\beta + \gamma}$ for $\tS_{m} = \sum_{k=1}^{\min(m, \tau)} I_{k-1}T_{k}$.
%Let $\hb = \hat{A}/\hat{B}$ and $\hg = 1/\hat{B} - \hb$.

The proof is based on a series of lemmas stated below. The first lemma bounds the probability that the epidemic diminishes before $m$ samples, which follows from \eqref{eq:a_m_1} of the proof of Lemma~\ref{lem:highprob}.
\begin{lemma} \label{lem:highprob2}
If $\frac{\beta}{\beta+\gamma}\frac{N-m-C_0}{N} > p$, $\Pr(\tau < m) \leq B_1 e^{-B_2 I_0}$, where $B_1, B_2 > 0$ are constant that depend only on $\beta$ and $\gamma$.
\end{lemma}

The next two lemmas give a high probability confidence bound for estimators $\hat{A}$ and $\hat{B}$.
\begin{lemma}\label{lem:Cm-lemma}
For any $m, I_0$ where $\frac{\beta}{\beta+\gamma}\frac{N-m-C_0}{N} > \frac{1}{2}$,
for any $\delta > 0$,
\begin{align*}
\Pr\left(\frac{C_m-C_0}{m} \notin \left[\frac{\beta}{\beta+\gamma}(1-\delta)\frac{N-m-C_0}{N},  \frac{\beta}{\beta+\gamma}(1+\delta)\right], \tau \geq m\right) \leq 2\exp(-m\delta^2/(4+2\delta)).
\end{align*}
\end{lemma}

\begin{lemma}\label{lem:Sm-lemma}
Let $\tS_{m} = \sum_{k=1}^{\min(m, \tau)} I_{k-1}T_{k}$. Then
\begin{align*}
\Pr\left(\frac{\tS_{m}}{m} \notin [\frac{(1-\delta)}{\beta+\gamma}, \frac{(1+\delta)}{\beta+\gamma}\frac{N}{N-m-C_0} ], \tau\geq m\right) \leq  2e^{-m  \frac{N-m-C_0}{N} (\delta - \ln(1+\delta))} .
\end{align*}
\end{lemma}

The next proposition combines the two estimators from the above lemmas and into estimators $\hb$ and $\hg$.
\begin{proposition}\label{thm:tail-bound}
Assume $\beta > \gamma > 0$.
Let $I_0 \leq m < N$ such that $\frac{\beta}{\beta+\gamma}\frac{N-m-C_0}{N} > p$.
Let $z = \frac{N-m-C_0}{N}$.
Then, for any $0< \delta < 1$, with probability $1-4e^{-m (\delta - \ln(1+\delta))} - 4e^{-m\delta^2/(4+2\delta)}-2B_1 e^{-B_2 I_0}$,
\begin{align}
% \hat{\beta} &\in [\beta-3\delta\beta, \beta +3\delta\beta] \\
\hat{\beta} &\in \left[\beta\frac{(1-\delta)z^2}{1+\delta}, \beta\frac{1+\delta}{1-\delta}\right] \label{eq:ci_beta} \\
\hat{\gamma} &\in \left[\gamma\frac{z}{1+\delta}+\beta\frac{(1-\delta)z-(1+\delta)^2}{(1+\delta)(1-\delta)}, \gamma\frac{1}{1-\delta}+\beta\frac{1+\delta-(1-\delta)^2z^2}{(1-\delta)(1+\delta)}\right], \label{eq:ci_gamma}
\end{align}
where $B_1, B_2 > 0$ are constants that depend on $\beta$ and $\gamma$.
\end{proposition}

We first show Theorem~\ref{thm:estimation-beta-gamma} using these results.
We then prove \cref{lem:Cm-lemma}, \cref{lem:Sm-lemma}, and \cref{thm:tail-bound} in \cref{sec:beta_gamma:proofs}.

\subsection{Proof of Theorem~\ref{thm:estimation-beta-gamma}}
\begin{myproof}
Let $\delta = \sqrt{\frac{5\log m}{m}}$.
First, we claim that the probability in \cref{thm:tail-bound} is greater than $1-\frac{8}{m} - 2B_1 e^{-B_2 I_0}$.
Note that $\ln(1+\delta) \leq \delta - \frac{\delta^2}{2} + \delta^3$, implying $\delta - \ln(1+\delta) \geq \delta^2(\frac{1}{2} - \delta)$.
Since $\delta \leq \frac{1}{4}$,
\begin{align*}
4e^{-m(\delta - \ln(1+\delta)}
&\leq 4e^{-m \frac{\delta^2}{4}} \leq \frac{4}{m}.
\end{align*}
Using $\delta \leq \frac{1}{4}$ again,
\begin{align*}
4e^{-m\delta^2/(4+2\delta)}
&\leq 4e^{-m \frac{\delta^2}{5}} = \frac{4}{m}.
\end{align*}
Hence, the bound in \ref{thm:tail-bound} holds with probability greater than $1 - \frac{8}{m} - 2B_1 e^{-B_2 I_0}$.

Since we assume $m(m + C_0) \leq N$ and $z = 1- \frac{m + C_0}{N}$,
\begin{align} \label{eq:zinequality}
1-z \leq \frac{1}{m}.
\end{align}

% Let $R$ be the event that the confidence bounds \eqref{eq:ci_beta}-\eqref{eq:ci_gamma} hold.
From here on, assume the confidence bounds \eqref{eq:ci_beta}-\eqref{eq:ci_gamma} hold.
Note that $\frac{1+\delta}{1-\delta} \leq 1+3\delta$ and $\frac{1-\delta}{1+\delta} \geq 1-3\delta$ for $\delta < \fourth$.
Then,
\begin{align*}
(\hb- \beta)^2 
&\leq \beta^2 \left( 1+3\delta - (1-3\delta)z^2 \right)^2 \\
&\leq \beta^2 \left( (1-z) + 3\delta(1+z) \right)^2 \\
&\leq \beta^2 \left( \frac{1}{m} + 6 \sqrt{\frac{5 \log m}{m}} \right)^2 \\
&\leq \beta^2 M_3 \frac{\log m}{m}
% &\leq \beta \left( \frac{1+\delta}{1-\delta} - \frac{(1-\delta)z^2}{1+\delta} \right)^2
\end{align*}
for an absolute constant $M_3 > 0$. The second last step uses \eqref{eq:zinequality} and $1+z \leq 2$.
Therefore, $\relerror(\hb, \beta) \leq M_1 \frac{\log m}{m}$.

Similarly,
\begin{align}
(\hg- \gamma)^2 
% &\leq \left( \gamma \frac{1}{1-\delta}+\beta\frac{1+\delta-(1-\delta)^2z^2}{(1-\delta)(1+\delta)} - \gamma \frac{z}{1+\delta} - \beta\frac{(1-\delta)z-(1+\delta)^2}{(1+\delta)(1-\delta)} \right)^2 \\
&\leq \left( \gamma \left(\frac{1}{1-\delta}  -  \frac{z}{1+\delta}\right) +\beta\left(\frac{1+\delta-(1-\delta)^2z^2}{(1-\delta)(1+\delta)} - \frac{(1-\delta)z-(1+\delta)^2}{(1+\delta)(1-\delta)}\right) \right)^2. \label{eq:gamma_var}
\end{align}

Using the fact that $(1-\delta)(1+\delta) \geq \frac{1}{2}$,
\begin{align*}
\frac{1}{1-\delta}  -  \frac{z}{1+\delta}
% &= \frac{1 + \delta - z(1-\delta)}{1-\delta^2} \\
\leq 2((1 - z) + \delta(1+z) )
\leq 2\left(\frac{1}{m} + 2 \sqrt{\frac{5 \log m}{m}}\right).
\end{align*}

\begin{align*}
\frac{1+\delta-(1-\delta)^2z^2}{(1-\delta)(1+\delta)} - \frac{(1-\delta)z-(1+\delta)^2}{(1+\delta)(1-\delta)}
&= \frac{(1+\delta) - (1-\delta)z + (1 + \delta)^2 - (1-\delta)^2 z^2}{1-\delta^2} \\
&\leq 2(1-z) + 4\delta (1+z) + \frac{1+\delta}{1-\delta} - \frac{1-\delta}{1+\delta} z^2 \\
&\leq 2(1-z) + 8\delta + (1+3\delta) - (1-3\delta)z^2 \\
&\leq 2(1-z) + 8\delta + (1-z^2)+ 6\delta(1+z^2) \\
&\leq (1-z)(3 +z) + \delta(8 + 6(1+z^2)) \\
&\leq \frac{4}{m} + 20\sqrt{\frac{5 \log m}{m}}.
\end{align*}

Substituting back into \eqref{eq:gamma_var} results in
\begin{align*}
(\hg- \gamma)^2
&\leq \left(\gamma \left(\frac{2}{m} + 4 \sqrt{\frac{5 \log m}{m}}\right) + \beta\left(\frac{4}{m} + 20\sqrt{\frac{5 \log m}{m}}\right)\right)^2  \\
&\leq M_2 \beta^2\frac{\log m}{m},
\end{align*}
for an absolute constant $M_2$, since $\beta > \gamma$.
This implies the desired result.
\end{myproof}

% -------------------------
% Proofs
% -------------------------
\subsection{Proofs of Intermediate Results} \label{sec:beta_gamma:proofs}

\subsubsection{Proof of \cref{lem:Cm-lemma}.}
\begin{myproof}
Fix $m$, let $z := \frac{N-m-C_0}{N}, p = \frac{\beta}{\beta+\gamma}z$. Then $p  > \frac{1}{2}.$ Define three stochastic processes $\{A_k: k \geq 0\}$, $\{B_k: k \geq 0\}$, $\{\tC_k: k \geq 0\}$:
\begin{align*}
A_k &= \begin{cases}
          C_0 & \text{if }k = 0 \\
          A_{k-1} + \Bern(p) & \text{otherwise}.\\
       \end{cases} \\
B_k &= \begin{cases}
          C_0 & \text{if }k = 0 \\
          B_{k-1} + \Bern(p/z) & \text{otherwise}.\\
       \end{cases} \\
\tC_k &= \begin{cases}
          C_0 & \text{if }k = 0 \\
          \tilde{C}_{k-1} + \Bern\left\{\frac{\beta(N-\tilde{C}_{k-1})}{\beta(N-\tilde{C}_{k-1}) + N\gamma}\right\} &\text{otherwise}.\\
       \end{cases}
\end{align*}
Note that $\tC_{k}$ is a modified version of $C_k$ where $\tC_k$ still evolves after the stopping time.

\begin{claim} \label{cl:stoch_dom2}
$A_m$ is stochastically less than $\tC_m$ ($A_m \leq_{st} \tC_m$); $\tC_m$ is stochastically less than $B_m$ ($\tC_m \leq_{st} B_m$); that is, for any $\ell \in \bR$,
\begin{align*}
\Pr(B_{m} \leq \ell) \leq \Pr(\tC_m \leq \ell) \leq \Pr(A_m \leq \ell).
\end{align*}
\end{claim}
This claim follows from \cref{thm:veinott}, using a similar argument to \cref{cl:stoch_dom}.

Let $A_{k} = C_0 + X_1 + X_2 + \dotsc X_{k}$ where $X_{i}\sim \Bern(p)$ are independent.
We provide the left tail bound for $C_{m}$. Note that when $\tau \geq m$, $C_{m} \overset{d}{=} \tilde{C}_{m}.$ Hence,
\begin{align}
\Pr(C_m \leq mp(1-\delta)+ C_0, \tau \geq m)
&= \Pr(\tC_m \leq mp(1-\delta) + C_0, \tau \geq m) \nonumber \\
&\leq \Pr(\tC_m \leq mp(1-\delta) + C_0)\nonumber  \\
&\leq \Pr(A_m \leq mp(1-\delta) + C_0)\label{eq:lower-tail}.
\end{align}
Using the Chernoff bound gives,
\begin{align*}
 \Pr(A_m \leq mp(1-\delta) + C_0) &= \Pr(C_0 + X_1 + \dots + X_m \leq pm(1-\delta) + C_0)\\
&= \Pr\left(X_1 + \dots + X_m \leq mp\left(1- \delta \right)\right) \\
&\leq \exp\left(- \frac{mp}{2}\delta^2\right).
\end{align*}
Therefore,
\begin{align*}
\Pr\left(\frac{C_m-C_0}{m}  \leq \frac{p}{z}(1-\delta)z, \tau \geq m\right)
&= \Pr\left(C_m \leq mp(1-\delta)+C_0, \tau \geq m\right)\\
&\leq \exp\left(- \frac{mp}{2}\delta^2\right) \leq \exp(-m\delta^2/4).
\end{align*}

Let $B_{k} = C_0 + Y_1 + \dotsc + Y_{k}$ where $Y_{i} \sim \Bern(p/z)$ are independent. Similarly, for the upper tail bound, we have
\begin{align*}
\Pr\left(\frac{C_m-C_0}{m} \geq \frac{p}{z}(1+\delta), \tau \geq m\right)
&=\Pr(C_m \geq mp/z(1+\delta)+ C_0, \tau \geq m)\\
&\leq \Pr(B_m \geq  mp/z(1+\delta) + C_0) \\
&\leq \Pr(C_0 + Y_1 + \dots + Y_m \geq mp/z(1+\delta) + C_0) \\
&\leq \exp(- \frac{mp/z}{2+\delta}\delta^2) \leq \exp(-m\delta^2/(4+2\delta))
\end{align*}
due to the multiplicative Chernoff bound $\Pr(Z \geq E[Z](1+\delta)) \leq e^{-\frac{Z}{2+\delta}\delta^2}$ where $Z$ is the sum of i.i.d Bernoulli random variables.

Combine upper and lower tail bounds and note that $p/z = \frac{\beta}{\beta+\gamma}$. Then, we can conclude, for any $\delta > 0$,
\begin{align*}
\Pr\left(\frac{C_m-C_0}{m} \notin [\frac{\beta}{\beta+\gamma}(1-\delta)z,  \frac{\beta}{\beta+\gamma}(1+\delta)], \tau \geq m\right) \leq 2\exp(-m\delta^2/(4+2\delta)).
\end{align*}

\end{myproof}

% \cref{lem:Cm-lemma} provides an estimator for $\frac{\beta}{\beta+\gamma}$. We next prove a lemma for estimating $\frac{1}{\beta+\gamma}$.
\subsubsection{Proof of Lemma~\ref{lem:Sm-lemma}.}
\begin{myproof}
Conditioned on $(I_0, C_0, I_1, C_1, \dotsc, I_{m-1}, C_{m-1})$ with $\tau \geq m$, we have $$
I_{k-1}T_{k} \sim \Exp\left(\beta\frac{N-C_{k-1}}{N}+\gamma\right)$$
 are independent exponential random variables.

 Theorem 5.1 in \cite{janson2018tail} gives us a tail bound for the sum of independent exponential random variables: let $X=\sum_{i=1}^{n} X_i$ with $X_i \sim \Exp(a_i)$ independent, then for $\delta > 0$,
 \begin{align}
 \Pr(X\geq (1+\delta)\mu) &\leq \frac{1}{1+\delta} e^{-a_* \mu (\delta - \ln(1+\delta))}  \leq e^{-a_* \mu (\delta - \ln(1+\delta))} \label{eq:upper-tail}\\
 \Pr(X\leq (1-\delta)\mu) &\leq e^{-a_*\mu (\delta -\ln(1+\delta))} \label{eq:lower-tail}
 \end{align}
 where $\mu = E[X], a_{*} = \min_{1\leq i\leq n} a_i.$

 Let $\tS_{m|\vec{C},\vec{I}}$ be $\tS_{m}$ conditioned on $(I_0, C_0, I_1, C_1, \dotsc, I_{m-1}, C_{m-1})$ with $\tau \geq m$. Let $\mu = E[\tS_{m|\vec{C},\vec{I}}] = \sum_{k=1}^{m} \frac{1}{\beta(N-C_{k-1})/N+\gamma}$, $a_* = \min_{1\leq k\leq m}\beta(N-C_{k-1})/N+\gamma.$ It is easy to verify the following facts
 \begin{align*}
 \mu a_* &\geq \sum_{k=1}^{m} \frac{a_*}{(\beta+\gamma)} \geq m \frac{N-m-C_0}{N}\\
 \frac{1}{\beta+\gamma}&\leq \frac{\mu}{m} \leq \frac{1}{\beta+\gamma}\frac{N}{N-m-C_0}.
 \end{align*}

Combining these with \cref{eq:upper-tail,eq:lower-tail}, we have
\begin{align*}
\Pr\left(\frac{\tS_{m|\vec{C},\vec{I}}}{m} \notin  \left[\frac{(1-\delta)}{\beta+\gamma}, \frac{(1+\delta)}{\beta+\gamma}\frac{N}{N-m-C_0} \right] \right)
&\leq \Pr\left(\frac{\tS_{m|\vec{C},\vec{I}}}{m} \notin  \left[\frac{\mu(1-\delta)}{m}, \frac{\mu(1+\delta)}
{m}\right] \right) \\
&\leq 2e^{-m  \frac{N-m-C_0}{N} (\delta - \ln(1+\delta))}.
\end{align*}

Therefore,
\begin{align*}
\Pr\left(\frac{\tS_m}{m} \notin I, \tau \geq m\right)
&= \int_{\vec{C}, \vec{I}|\tau\geq m}\Pr\left(\frac{\tS_m}{m} \notin I~|~ \vec{C}, \vec{I}, \tau \geq m\right) f(\vec{C}, \vec{I} | \tau \geq m) \Pr(\tau \geq m)\\
&\leq 2e^{-m  \frac{N-m-C_0}{N} (\delta - \ln(1+\delta))} \Pr(\tau \geq m)\\
&\leq 2e^{-m  \frac{N-m-C_0}{N} (\delta - \ln(1+\delta))}.
\end{align*}
\end{myproof}

\subsubsection{Proof of \cref{thm:tail-bound}.}
\begin{myproof}
Let $\hat{\beta} = \frac{C_{m} - C_0}{\tS_m}, z = \frac{N-C_0-m}{N}$. Suppose $x \in \frac{\beta}{\beta+\gamma}[(1-\delta)z, 1+\delta], y \in \frac{1}{\beta+\gamma}[1-\delta, (1+\delta)1/z]$. Then, %for $\delta \in (0, \frac{1}{3})$,
\begin{align} \label{eq:interval_beta}
\frac{x}{y} \in \left[\beta\frac{(1-\delta)z^2}{1+\delta}, \beta\frac{1+\delta}{1-\delta}\right] % \subseteq [\beta(1-3\delta)z^2, \beta(1+3\delta)],
\end{align}

% Note that $\frac{1-\delta}{1+\delta} \geq 1-3\delta$, and $\frac{1+\delta}{1-\delta} \geq 1+3\delta$ for $\delta \in (0, \frac{1}{3})$.
% \begin{align}
% \Pr\left(\hat{\beta} \notin \left[\beta\frac{(1-\delta)z^2}{1+\delta}, \beta\frac{1+\delta}{1-\delta}\right], \tau > m\right) \leq 2e^{-m  \frac{N-m-I_0}{N} (\delta - \ln(1+\delta))} + 2e^{-m\delta^2/(4+2\delta)}.
% \end{align}

Similarly, let $\hat{\gamma} = \frac{m}{\tS_m} - \hat{\beta}.$ Suppose $a \in (\beta+\gamma)[\frac{z}{1+\delta}, \frac{1}{1-\delta}], b \in \beta[\frac{(1-\delta)z^2}{1+\delta}, \frac{1+\delta}{1-\delta}]$. Then
\begin{align} \label{eq:interval_gamma}
a - b \in \left[\gamma\frac{z}{1+\delta}+\beta\frac{(1-\delta)z-(1+\delta)^2}{(1+\delta)(1-\delta)}, \gamma\frac{1}{1-\delta}+\beta\frac{1+\delta-(1-\delta)^2z^2}{(1-\delta)(1+\delta)}\right].
\end{align}

Then, for any sets $U_1, U_2$,
\begin{align*}
\Pr(\hb \in U_1, \hg \in U_2)
&\geq  1 - \Pr(\hb \notin U_1) - \Pr(\hg \notin U_2) \\
&\geq 1-\Pr(\hb \notin U_1, \tau > m) - \Pr(\hb \notin U_2, \tau > m)  - 2 \Pr(\tau < m) \\
&\geq 1-4e^{-m (\delta - \ln(1+\delta))} - 4e^{-m\delta^2/(4+2\delta)}-2B_1 e^{-B_2 I_0},
\end{align*}
where the last step uses \cref{lem:highprob2}, \cref{lem:Cm-lemma} and \cref{lem:Sm-lemma}, using the intervals  \eqref{eq:interval_beta} and \eqref{eq:interval_gamma} for $U_1$ and $U_2$ respectively.

\end{myproof}

\section{Proofs of Propositions} \label{sec:app:prop_proofs}
% The proof of each result is contained in its own subsection in the order they appear in the paper.

% ------------------
% Constant fraction
% ------------------
\subsection{Proof of Proposition~\ref{prop:constant_fraction}}
%!TEX root=../SIR-model-EC.tex

%% removed >
\iffalse
We begin by describing the {\em deterministic} SIR model. Let $s(t), i(t)$ and $r(t)$ be the size of the susceptible, infected and recovered populations respectively, as observed at time $t$.

The SIR model is defined by the following system of ODEs, specified by the tuple of parameters $(\alpha, \beta, \gamma)$:
\begin{align}
\label{eq:sirfluid}
\frac{ds}{dt} &=  - \beta \frac{s}{\alpha N} i   &
\frac{di}{dt} &= \beta \frac{s}{\alpha N} i - \gamma i  &
\frac{dr}{dt} &= \gamma i
\end{align}
The rate of recovery is specified by $\gamma > 0$; $1/\gamma$ is frequently referred to as the {\em infectious period}. $\beta > 0$ quantifies the rate of transmission; $\beta/\gamma \triangleq R_0$ is also referred to as the {\em basic reproduction number}. $N$ is the total population (assumed known). The quantity $\alpha \in (0,1]$ corresponds to the fraction of true infections we actually observe. The role of $\alpha$ is made clear by the following Proposition:

\begin{proposition} \label{prop:constant_frac}
Let $\{(s'(t),i'(t),r'(t)): t\geq 0\}$ be a solution to \eqref{eq:sirfluid} for parameters $\alpha = 1, \beta = \beta', \gamma = \gamma'$ and $i(0) = i'(0)$. Then, for any $\eta > 0$, $\{(\eta s'(t), \eta i'(t), \eta r'(t)): t\geq 0\}$ is a solution to \eqref{eq:sirfluid} for parameters $\alpha = \eta, \beta = \beta', \gamma = \gamma'$ and $i(0) = \eta i'(0)$.
\end{proposition}
\fi
%% removed <

\begin{myproof}%[Proof of Proposition~\ref{prop:constant_fraction}]
As in \cite{miller2017mathematical, miller2012note},  the solution $\{(s'(t),i'(t),r'(t)): t\geq 0\}$ can be written as:

\begin{align*}
s'(t) &=s'(0) e^{-\xi'(t)} \\
i'(t) &=N' - s'(t) - r'(t) \\
r'(t) &=r(0) + \frac{\gamma' N' }{\beta'}\xi'(t) \\
\xi'(t) &=\frac{\beta'}{N'} \int_{0}^{t} i'\left(t^{*}\right) d t^{*}
\end{align*}

Making the appropriate substitutions yields the following equivalent system:

\begin{align}
  i'(t) &= N' - s'(0) \exp \left(- \frac{\beta'}{N'} \xi(t) \right) - r(0) -  \frac{\gamma' N'}{\beta'}\xi'(t) \label{eq:i} \\
  \xi'(t) &=\frac{\beta'}{N'} \int_{0}^{t} i'\left(t^{*}\right) d t^{*}. \label{eq:xi}
\end{align}

Therefore, it remains to show that for $\eta > 0$, $\{(s(t), i(t), r(t)): t\geq 0\} \triangleq \{(\eta s'(t), \eta i'(t), \eta r'(t)): t\geq 0\}$ is a solution for \eqref{eq:i} and \eqref{eq:xi} where $N'$ is replaced with $\eta N'$. Starting with \eqref{eq:i},

\begin{align*}
  i'(t) &= N' - s'(0) \exp \left(-\xi'(t) \right) - r'(0) -  \frac{\gamma' N'}{\beta'} \xi'(t) \\
  \eta i'(t) &= \eta \left(N' - s'(0) \exp \left(-  \xi'(t) \right) - r'(0) -  \frac{\gamma' N'}{\beta'}\xi'(t) \right) \\
  &= \eta N' - \alpha s'(0) \exp \left(-\xi(t) \right) - \eta r(0) -  \frac{\gamma' \eta N' }{\beta'} \xi(t)  
\end{align*}

where $\xi(t) = \xi'(t) =\frac{\beta'}{N' \eta} \int_{0}^{t} \eta i'\left(t^{*}\right) d t^{*}$. Noting that $\xi'(t) = \xi(t)$ and substituting $i(t) = \eta i'(t)$ yields the equations below, clearly showing that $\{(s(t), i(t), r(t)): t\geq 0\}$ satisfy \eqref{eq:i} and \eqref{eq:xi}:

\begin{align*}
    i(t) &= \eta N' - s(0) \exp \left(-\xi(t) \right) - r(0) -  \frac{\gamma' \eta N'}{\beta'} \xi(t) \\
    \xi(t) &=\frac{\beta'}{N' \eta} \int_{0}^{t} i\left(t^{*}\right) d t^{*}.
\end{align*}

\end{myproof}

% ---------------
% Identifiability
% ---------------
\subsection{Proof of Proposition~\ref{prop:identity_fluid}}
%!TEX root=../SIR-model-EC.tex

%%removed >
\iffalse
\textcolor{teal}{$<$ Andreea: See paper for deterministic model. $>$}

\begin{proposition}
Fix $N$, and let $i(t)$ be observed over some open set in $\mathbb{R}_+$. Then the parameters $(\alpha, \beta, \gamma)$ are identifiable.
\end{proposition}
\fi
%%removed >

\subsubsection{SIR Model}
\begin{myproof}%[Proof of Proposition~\ref{prop:identity_fluid}]
Consider initial conditions $(s(0),i(0),0)$, as in \cite{miller2017mathematical,miller2012note}, the analytical solution is given by
\begin{align*}
    s(t) & =  s(0) e^{-\xi(t)}, \\
    i(t) & = N - s(t) - r(t), \\
    r(t) & =  \frac{\gamma \; N}{\beta} \xi(t), \\
    \xi(t) & = \frac{\beta}{N} \int_0^t i(t') dt'.
\end{align*}
Consider two SIR models with parameters $(N, \beta, \gamma)$ and $(N', \beta', \gamma')$, and initial conditions $(s_0, i_0, 0)$ and $(s'_0, i'_0, 0)$ respectively. We claim that infection trajectories $i(t)$ and $i'(t)$ being identical on an open set $[0,T)$ implies the parameters and initial conditions are identical as well.

Assume $i(t) = i'(t)$ for all $t \in [0,T)$; then, given the exact solution above it follows that
\begin{align*}
    N -  s_0 e^{-\frac{\beta}{N} x}  - \gamma x = N' -  s'_0 e^{-\frac{\beta'}{N'} x} -  \gamma' x, & & \text{ for all } x \in \Big[ 0, \int_0^T i(t)dt \Big]
\end{align*}
As functions of $x$, both the RHS and LHS in the equality above are holomorphic, and hence, using the identity theorem, we then have 
for all $x \in \mathbb{R}$, there is $N -  s_0 e^{-\frac{\beta}{N} x}  - \gamma x = N' -  s'_0 e^{-\frac{\beta'}{N'} x} -  \gamma' x$. 

Then the following implies $\gamma = \gamma'$: 
\begin{align*}
-\gamma = \lim_{x\rightarrow +\infty} \frac{ N -  s_0 e^{-\frac{\beta}{N} x}  - \gamma x}{x} = \lim_{x\rightarrow +\infty}=\frac{  N' -  s'_0 e^{-\frac{\beta'}{N'} x} -  \gamma' x}{x} = -\gamma'.
\end{align*}

Hence for all $x \in \mathbb{R}$, $N -  s_0 e^{-\frac{\beta}{N} x}  =  N' -  s'_0 e^{-\frac{\beta'}{ N'} x}.$ Again, by taking $x$ to infinity, we can conclude $N= N'$ by the following
\begin{align*}
N = \lim_{x\rightarrow +\infty}  \left( N -  s_0 e^{-\frac{\beta}{N} x}\right) = \lim_{x\rightarrow +\infty} \left(N' -  s_0' e^{-\frac{\beta'}{ N'} x}\right) = N'.
\end{align*}

Furthermore, by taking $x=0$, we can also get $s_0 = s_0'$ and then $\beta = \beta'$ follows. This completes the proof. 
\end{myproof}

\subsubsection{Bass Model}
\begin{myproof} 
Consider the initial condition $i(0) = 0.$ By the analytic solution given by \cite{bass1969new}, we have
\begin{align*}
i(t) = N \frac{1-e^{-(p+\beta)t}}{\frac{\beta}{p}e^{-(p+\beta)t}+1}.
\end{align*}

Consider two bass models with parameters $(N, \beta, p)$ and $(N', \beta', p')$ and initial conditions $i(0) = 0, i'(0) = 0$ respectively. We claim that trajectories $i(t)$ and $i'(t)$ being identical on an open set $[0,T)$ implies the parameters are identical as well. 

Assume $i(t) = i'(t)$ for all $t \in [0,T)$; then, given the exact solution above it follows that
\begin{align} \label{eq:identity-T}
   N \frac{1-e^{-(p+\beta)t}}{\frac{\beta}{p}e^{-(p+\beta)t}+1} = N' \frac{1-e^{-(p'+\beta')t}}{\frac{\beta'}{p'}e^{-(p'+\beta')t}+1}, & & \text{ for all } t \in [0, T)
\end{align}
As functions of $t$, both the RHS and LHS in the equality above are holomorphic, and hence, using the identity theorem, we then have \cref{eq:identity-T} holds for all $t \in \mathbb{R}.$ 

By taking $t$ to infinity,  we can easily obtain $N = N'.$ Furthermore, taking the derivative for $t$ on both sides of \cref{eq:identity-T}, one can obtain
\begin{align}\label{eq:derivative-identity-T}
\frac{(p+\beta)^2}{p} \frac{e^{-(p+\beta)t}}{(\beta/p\cdot e^{-(p+\beta)t} + 1)^2} = \frac{(p'+\beta')^2}{p'} \frac{e^{-(p'+\beta')t}}{(\beta'/p'\cdot e^{-(p'+\beta')t} + 1)^2}.
\end{align}
By taking $t = 0$ on both sides of \cref{eq:derivative-identity-T}, one can verify that $p = p'.$ Furthermore, let $g(t) = \frac{(p+\beta)^2}{p} \frac{e^{-(p+\beta)t}}{(\beta/p\cdot e^{-(p+\beta)t} + 1)^2}$ and $g'(t) = \frac{(p'+\beta')^2}{p'} \frac{e^{-(p'+\beta')t}}{(\beta'/p'\cdot e^{-(p'+\beta')t} + 1)^2}.$ 

Note that
\begin{align*}
-(p+\beta) = \lim_{t\rightarrow +\infty} \frac{\ln(g(t))}{t} = \lim_{t\rightarrow +\infty} \frac{\ln(g'(t))}{t} = - (p' + \beta').
\end{align*}

We then can conclude $\beta = \beta'$. This completes the proof. 

\end{myproof}

% ---------------
% Time to peak
% ---------------
\subsection{Proof of \cref{prop:time_to_peak_bass}}
%!TEX root=../SIR-model-EC.tex

\begin{myproof}
Note that $\bE[T_{i}] = \frac{N}{pN(N-i) + \beta i (N-i)}.$ Then 
$$
\bE[\basscr] = \bE\left[\sum_{i=1}^{N^{2/3}-1} T_{i}\right] = \sum_{i=1}^{N^{2/3}-1} \frac{N}{pN(N-i) + \beta i (N-i)}.
$$
Let $f(x) = \frac{N}{pN(N-x) + \beta x (N-x)}$, we use $f(x)$ as a proxy to bound $\bE[\basscr].$ Easy to verify that $f(x)$ is decreasing when $x \in (0, \mathring{r}]$ where $\mathring{r} = (1-p/\beta)N/2.$  Note that $p/\beta < c$ for some constant $c$ since $p/\beta = \Theta(N^{-\alpha})$ for $\alpha>0$. Hence when $N \rightarrow \infty$, we have $\mathring{r} \gg N^{2/3}$ and
\begin{align*}
\sum_{i=1}^{N^{2/3}-1} \frac{N}{pN(N-i) + \beta i (N-i)} 
&\geq \int_{x=1}^{N^{2/3}} f(x)dx\\
&= \frac{\ln (\beta x + Np) - \ln (N-x)}{p+\beta} \Big\rvert_{x=1}^{N^{2/3}}\\
&= \frac{\ln (\beta N^{2/3} + Np) - \ln (\beta + Np) + \ln (N-1) - \ln (N - N^{2/3})}{p + \beta}\\
&\geq \frac{\ln (\beta N^{2/3} + Np) - \ln (\beta + Np)}{p + \beta}.
\end{align*}

Similarly, for $\basspeak$, we have
\begin{align*}
E[\basspeak] &= \sum_{i=1}^{\mathring{r}-1} \frac{N}{pN(N-i) + \beta i (N-i)} \\
&\leq f(1) + \int_{x=1}^{\mathring{r}} f(x)dx\\
&\leq f(1) + \frac{\ln (\beta N + Np) - \ln (pN + \beta) + \ln \frac{1}{1-c}}{p + \beta}\\
&\leq \frac{\ln (\beta N + Np) - \ln (pN + \beta) + c'}{p + \beta}
\end{align*}
for some absolute constant $c'$.

Let $\frac{\beta}{p} = C\cdot N^{\alpha}$ for some constant $C.$ We then have
\begin{align*}
\frac{E[\basscr]}{E[\basspeak]} 
&\geq  \frac{\ln (\beta N^{2/3} + Np) - \ln (\beta + Np)}{\ln (\beta N + p N) - \ln (pN+\beta) + c'}\\
&\geq \frac{\ln\left(\frac{CN^{2/3+\alpha} + N}{CN^{\alpha} + N} \right)}{\ln\left(\frac{CN^{1+\alpha} + N}{CN^{\alpha} + N} \right) + c'} =: k_{N}.
\end{align*}

Then, it is easy to verify that when $ \frac{1}{3} <\alpha \leq 1$, $\lim_{N\rightarrow \infty} k_{N} = \frac{\alpha-1/3}{\alpha}.$ When $\alpha 
> 1$, $\lim_{N\rightarrow \infty} k_{N} = \frac{2}{3}.$ 

Note that we also have $\bE[\basscr] \leq f(1) + \int_{x=1}^{N^{2/3}} f(x)dx$ and $\bE[\basspeak] \geq \int_{x=1}^{\mathring{r}} f(x)dx$. Similarly, one can verify that
\begin{align*}
\limsup_{N\rightarrow \infty} \frac{\bE[\basscr]}{\bE[\basspeak]} \leq  \begin{cases} 0 & \alpha \leq \frac{1}{3} \\ \frac{\alpha - \frac{1}{3}}{\alpha} & \frac{1}{3} <\alpha \leq 1 \\ \frac{2}{3} & \alpha > 1\end{cases}.
\end{align*}
This completes the proof. 
\end{myproof}

\subsection{Proof of Proposition~\ref{thm:deterministic-time}} \label{sec:app:time_to_peak_sir}
%!TEX root=../SIR-model-EC.tex

Let $\sirpeak = \inf\{t : \beta(s)t / N < \gamma\}$ be the time when the number of infections is at its peak.
It is easy to show that $\sirrtpeak \leq \sirpeak$.
We show the analog of Proposition~\ref{thm:deterministic-time} with the peak defined instead as $\sirpeak$ --- i.e. we show $\liminf_{N \rightarrow \infty} \frac{\sircr}{\sirpeak} \geq \frac{2}{3}$.
Then, the desired result follows since $\sirrtpeak \leq \sirpeak$.

First, we prove $\sirrtpeak \leq \sirpeak$.
We can write $\frac{d^2s}{dt^2}$ as
\begin{align}
\frac{d^2s}{dt^2}
&= \frac{-\beta}{\pop}\left( \frac{ds}{dt} i + \frac{di}{dt}s \right) \nonumber \\
&= \frac{-\beta}{\pop}\left( \frac{-\beta s}{\pop} i^2+ \left(\frac{\beta s}{\pop} - \gamma \right)i s \right) \nonumber \\
% &= -\beta is\left( -\beta i+ \beta s - \gamma \pop  \right)  \\
&= \frac{ \beta^2 is}{N^2} \left( i- s + \frac{\gamma}{\beta} \pop  \right). \label{eq:peak4}
\end{align}
From \eqref{eq:peak4}, we see that $\frac{d^2s}{dt^2} > 0$ if and only if
\begin{align*} 
s <  \frac{\gamma}{\beta} \pop + i.
\end{align*}

By definition, $\sirpeak$ occurs at a time when 
\begin{align*}
s < \frac{\gamma}{\beta} \pop.
\end{align*}

Since $s$ is decreasing and $i$ is non-negative, clearly $\sirrtpeak$ occurs before $\sirpeak$.

Next, we prove $\liminf_{N \rightarrow \infty} \frac{\sircr}{\sirpeak} \geq \frac{2}{3}$.
The crux of the problem is summarised in two smaller results, bounding $\sircr$ and $\sirpeak$ respectively. 
% To ease our exposition, we will let $P_n = \alpha_n N_n$ and drop the index $n$ in these helper results.
Let $\rho_1 = 1 - \frac{1}{\log \log N}$ and $\rho_2 = \frac{\gamma}{\beta}$.

\begin{proposition}
 \label{prop:T1}
There exists a constant $\nu_1$ that only depends on $\gamma, \beta$ such that
\[ \sircr \geq \frac{1}{\beta - \gamma} \left( \frac{2}{3} \log\frac{\nu_1 N}{c(0)^{3/2}}  + \log\frac{\nu_1^{2/3}}{ c(0)} \Big( 1-\frac{c(0)}{N^{2/3}} \Big) \right). \]
\end{proposition}

\begin{proposition}
 \label{prop:T2}
 There exists a constant $\nu_2$ that only depends on $\gamma, \beta$ and a constant $C = O(1)$, such that
\[ \sirpeak \leq \frac{1}{\beta \rho_1 - \gamma} \log \frac{\nu_2 N}{i(0)} + \frac{C}{1- \rho_1}. \]
\end{proposition}

The argument follows directly by taking the limit of the bounds we provide in Propositions~\ref{prop:T1}-\ref{prop:T2}. Specifically, using that the constants $\nu_{1}, \nu_{2}$ do not depend on $N$, we arrive at

\begin{align*}
\limsup_{N \rightarrow \infty} \frac{\sirpeak}{\sircr}
&\leq \limsup_{N \rightarrow \infty}  \frac{\frac{1}{\beta \rho_1 - \gamma} \log \frac{\nu_{2} N}{i(0)} + \frac{C}{(1-\rho_1)}}{\frac{1}{\beta - \gamma} \left( \frac{2}{3} \log\frac{\nu_{1} N}{c(0)^{3/2}}  + \log\frac{\nu_{1}^{2/3}}{ c(0)} \Big( 1-\frac{c(0)}{N^{2/3}} \Big) \right)} \\
&= \limsup_{N \rightarrow \infty}  \frac{\beta - \gamma}{\beta \rho_1 - \gamma}\cdot \frac{\log N + \log \nu_{2} - \log i(0)}{\frac{2}{3}\log N + \frac{4}{3}\log \nu_{1} - 2\log c(0) + \log \Big( 1-\frac{c(0)}{N^{2/3}} \Big) } \\
&+ \limsup_{N \rightarrow \infty} \frac{ (\beta - \gamma) C \log \log N}{\frac{2}{3}\log N + \frac{4}{3}\log \nu_{1} - 2\log c(0) + \log \Big( 1-\frac{c(0)}{N^{2/3}} \Big)}
\end{align*}

$\rho_1 \rightarrow 1$ as $N \rightarrow \infty$, so $\frac{\beta - \gamma}{\beta \rho_1 - \gamma} \rightarrow 1$.
Since $c(0) = O(\log(N))$ by assumption (and $i(0) \leq c(0)$), and $C = O(1)$ by Proposition~\ref{prop:T2}, the limits of the two summands above are $3/2$ and $0$ respectively, which concludes the proof.

\subsubsection{Proof of Proposition~\ref{prop:T1}.}
\begin{myproof}[Proof of Proposition~\ref{prop:T1}]
Define $\ti(t)$ such that $\ti(0) = i(0)$ and $\frac{d\ti}{dt} = (\beta - \gamma) \ti$, implying
\begin{align*}
\ti(t) = i(0) \exp\{ (\beta - \gamma) t\}.
\end{align*}
Since $\frac{d\ti}{dt} \geq \frac{di}{dt}$ for all $t$, $\ti(t) \geq i(t)$ for all $t$. Then, for all $t$,
\begin{align*}
\frac{ds}{dt} = - \beta \frac{s}{N} i \geq -\beta i
              &\geq - \beta \ti.
\end{align*}
Hence we can write
\begin{align*}
s(t)
&\geq s(0) +  \int_{0}^{t} - \beta \ti(t') dt' \\
&=s(0)- \beta  i(0)\int_{0}^{t}  \exp\{ (\beta  - \gamma) t'\} dt' \\
&= s(0)- \frac{\beta  i(0)}{\beta  - \gamma}( \exp\{ (\beta- \gamma) t \} -1)  \\
\end{align*}

Since $s(0) - s(\sircr) = N^{2/3} - c(0)$, setting $t = \sircr$ and solving for $\sircr$ in the inequality above results in
\begin{align*}
\sircr &\geq \frac{1}{\beta - \gamma} \log\left( \frac{\beta - \gamma}{\beta i(0)} (N^{2/3}-c(0)) \right) \\
&\geq \frac{1}{\beta - \gamma} \log\left( \frac{\beta - \gamma}{\beta c(0)} (N^{2/3}-c(0)) \right) \\
&= \frac{1}{\beta - \gamma} \left( \log\frac{\beta - \gamma}{\beta c(0)} (N^{2/3})  + \log\frac{\beta - \gamma}{\beta c(0)} \Big( 1-\frac{c(0)}{N^{2/3}} \Big) \right) \\
&= \frac{1}{\beta - \gamma} \left( \frac{2}{3} \log\frac{\nu_1 N}{c(0)^{3/2}}  + \log\frac{\nu_1^{2/3}}{ c(0)} \Big( 1-\frac{c(0)}{N^{2/3}} \Big) \right)
\end{align*}
for $\nu_1 = \left(\frac{\beta - \gamma}{\beta}\right)^{3/2}$ as desired.
\end{myproof}

\subsubsection{Proof of Proposition~\ref{prop:T2}.}
For $\rho \in [0, \frac{\gamma}{\beta}]$, let $t_\rho$ be the time $t$ when $\frac{s(t)}{N} = \rho$.
$\rho$ will represent the fraction of the total population that is susceptible.
Since $\rho \leq \frac{\gamma}{\beta}$, $i$ is increasing for the time period of interest.

Let $\beta > \gamma$, $N$ be fixed. Let $\rho_1 = 1 - \frac{1}{\log \log N}$ and $\rho_2 = \frac{\gamma}{\beta}$.
We assume $N$ is large enough that $\rho_1 > \rho_2$, hence $\tone < \ttwo$.
$\sirpeak = \ttwo$.

\begin{lemma} \label{lemma1}
For any $\rho \in [0, \frac{\gamma}{\beta}]$,
$i(t_\rho) \geq N(1-\rho)\frac{\beta \rho - \gamma}{\beta \rho}  - \frac{c(0)}{2}$.
\end{lemma}

\begin{myproof}[Proof of Lemma~\ref{lemma1}]
Fix $\rho$.
At time $t_\rho$, the total number of people infected is $c(t_\rho) = i(t_\rho) + r(t_\rho) = N(1-\rho)$, by definition.
At any time $t \leq t_\rho$, the rate of increase in $i$ is $\frac{\beta \frac{s(t)}{N} - \gamma}{\beta \frac{s(t)}{N}} \geq \frac{\beta \rho - \gamma}{\beta \rho}$ of the rate of increase in $c$.
Therefore, $i(t_\rho) - i(0) \geq \big(\frac{\beta \rho - \gamma}{\beta \rho} \big) \big( c(t_\rho) - c(0) \big)$ and $i(t_\rho) \geq \big(\frac{\beta \rho - \gamma}{\beta \rho} \big) N(1-\rho) - \frac{\beta \rho - \gamma}{\beta \rho} c(0) +  i(0)$. Using the fact that $i(0) \geq \frac{c(0)}{2}$ and rearranging terms gives the desired result.
%Therefore, since $I(0) = C(0)$, $I(t_\rho) \geq C(t_\rho) \frac{\beta \rho - \gamma}{\beta \rho} = N(1-\rho) \frac{\beta \rho - \gamma}{\beta \rho}$.
\end{myproof}

\begin{lemma} \label{lemma2}
% \begin{align*}
For $t \in [\tone, \ttwo]$, where $\rho_2 > \rho_1$ for $\rho_1, \rho_2 \in [ 0, \frac{\gamma}{\beta}]$, $\ttwo-\tone \leq \frac{N(\rho_1-\rho_2)}{\beta \rho_2 i(\tone)}$.
% \end{align*}
\end{lemma}
\begin{myproof}[Proof of Lemma~\ref{lemma2}]
The difference in $s$ between $\tone$ and $\ttwo$ is $s(\tone) - s(\ttwo) = N(\rho_1 - \rho_2)$. As a consequence of the mean value theorem, $\frac{s(\ttwo) - s(\tone)}{\ttwo - \tone} \leq \max_{t \in [\tone, \ttwo]} \{ \frac{ds}{dt} \}$. Using these two expressions,
% removed --> The difference in $S$ between $\tone$ and $\ttwo$ is $N(\rho_1-\rho_2)$. Then, for all $t \in [\tone, \ttwo]$,
\begin{align*}
\frac{N(\rho_1 - \rho_2)}{\ttwo - \tone} &\geq \min\left\{ -\frac{ds}{dt} \right\} = \min\left\{\beta \frac{s(t)}{N} i(t) : t \in [\tone, \ttwo]\right\} \geq \beta \rho_2 i(\tone)
% removed > \bigg|\frac{dS}{dt}\bigg| \geq \min\left\{\beta \frac{S(t)}{N} I(t) : t \in [\tone, \ttwo]\right\} \geq \beta \rho_2 I(\tone).
\end{align*}
The desired expression follows from rearranging terms.
\end{myproof}

\begin{lemma} \label{lemma3}
For any $\rho \leq \min\{\frac{\gamma}{\beta}, 1/2\}$,
$t_\rho \leq \frac{1}{\beta \rho - \gamma}\log \frac{\nu_2}{i(0)} N$, for $\nu_2 = \frac{2(\beta - \gamma)}{\beta}$.
\end{lemma}
The proof of this lemma follows the exact same procedure as the proof of Proposition~\ref{prop:T1}.

\begin{myproof}[Proof of Lemma~\ref{lemma3}]
We proceed in the same way as the proof of Proposition~\ref{prop:T1} except in this case we will lower bound $s(0)-s(t)$. We achieve this by letting $\ti$ be defined to grow slower than $i$, so it is used as a lower bound.
Define $\ti(t)$ such that $\ti(0) = i(0)$ and $\frac{d\ti}{dt} = (\beta \rho - \gamma) \ti$, implying
\begin{align*}
\ti(t) = i(0) \exp\{ (\beta \rho - \gamma) t\}.
\end{align*}
Since $\frac{d\ti}{dt} \leq \frac{di}{dt}$ when , $\ti(t) \leq i(t)$ for all $t < \ttwo$. In addition, when $t < \ttwo$, $\frac{s}{N} \geq \frac{\gamma}{\beta} \geq \rho$. Then, for $t < \ttwo$,
\begin{align*}
\frac{ds}{dt} = - \beta \frac{s}{N} i
              &\leq - \beta \rho \ti.
\end{align*}
Hence we can write
\begin{align*}
s(t)
&\leq s(0) +  \int_{0}^{t} - \beta \rho \ti(t') dt' \\
&=s(0)- \beta \rho i(0)\int_{0}^{t}  \exp\{ (\beta \rho - \gamma) t'\} dt' \\
&= s(0)- \frac{\beta \rho i(0)}{\beta \rho - \gamma}( \exp\{ (\beta \rho - \gamma) t \} -1)  \\
\end{align*}

%Notice we used the fact that for every $t \leq t_\rho$, $s/N$ equals $\rho'$ corresponding to $t$ (i.e. $t_{\rho'}=t$) and hence, $s/N$ is bounded by $\rho$.
Since $s(t_\rho) = \rho N$,
\begin{align*}
\rho N &\leq  s(0) - \frac{\beta \rho i(0)}{\beta \rho - \gamma}( \exp\{ (\beta \rho - \gamma) t_\rho \} -1).
\end{align*}
Solving for $t_\rho$ results in
\begin{align*}
t_\rho
 \leq \frac{\log\left(\frac{\beta \rho - \gamma}{\beta \rho i(0)} (s(0) - \rho N) + 1 \right)}{\beta \rho - \gamma} \leq \frac{1}{\beta \rho - \gamma} \log \Big( \frac{\nu_2}{i(0)} N \Big )
\end{align*}
% Let $\rho_0 = S(0)/N$. Since $I(0) = N - S(0)$ is a constant, $\rho_0 \rightarrow 1$ as $N \rightarrow \infty$.
where $\nu_2 = \frac{2(\beta-\gamma)}{\beta}$, using the fact that $\rho \leq 1/2$. %$\rho > 1/2$.
\end{myproof}

\begin{myproof}[Proof of Proposition~\ref{prop:T2}]
Using the results from Lemmas~\ref{lemma1}-\ref{lemma3},
\begin{align*}
\ttwo
&= \tone + (\ttwo - \tone) \\
&\leq \frac{1}{\beta \rho_1 - \gamma} \log \Big ( \frac{\nu_2}{i(0)} N \Big )%(\frac\nu_2 N)
+ \frac{N(\rho_1-\rho_2)}{\beta \rho_2 i(\tone)} \\
&\leq \frac{1}{\beta \rho_1 - \gamma} \log \Big ( \frac{\nu_2}{i(0)} N \Big ) + \frac{N(\rho_1-\rho_2)}{\frac{\rho_2}{\rho_1} N (1-\rho_1) (\beta \rho_1 - \gamma) - \frac{\beta \rho_2}{2} c(0)} \\
&= \frac{1}{\beta \rho_1 - \gamma} \log \Big ( \frac{\nu_2}{i(0)} N \Big ) + \frac{C}{1-\rho_1},
%&\leq \frac{1}{\beta \rho_1 - \gamma} \log(\nu N) + \frac{(\rho_1-\rho_2) \beta \rho_1}{\beta \rho_2 (1-\rho_1) (\beta\rho_1-\gamma)} \\
%&\leq \frac{1}{\beta \rho_1 - \gamma} \log(\nu N) + \frac{C}{(1-\rho_1)},
\end{align*}
%where $C = \frac{(1-\rho_2) \beta}{\beta \rho_2 (\beta/2 - \gamma)}$ only depends on $\beta, \gamma$, using the fact that $\rho_1 > 1/2$.
where $C = \frac{\rho_1 - \rho_2}{\frac{\rho_2}{\rho_1}(\beta \rho_1 - \gamma)%\rho_2)
- \frac{\beta \rho_2}{2} \frac{c(0)}{N(1-\rho_1)}}$. Note that, as required in the statement, $C = O(1)$. Indeed,
\begin{align*}
    C =
    \frac{(\rho_1 - \rho_2)}{\frac{\rho_2}{\rho_1}(\beta \rho_1 - \gamma)%\rho_2)
    - \frac{\beta \rho_2}{2}\frac{c(0)}{N(1-\rho_1)}} =
    \frac{1 - \rho_2 - \frac{1}{\log\log N}}{\beta \rho_2
    - \frac{\gamma \rho_2}{1 - \frac{1}{\log \log N}}
    %+\frac{\rho_2^2}{\log\log N - 1}
    - \frac{\beta \rho_2}{2}\frac{c(0) \log\log N}{N}},
\end{align*}
and so, as $N$ grows large, $C$ tends to $(1-\rho_2)/\rho_2( \beta  - \gamma)$ (recall that $c(0) = O(\log\log N)$).

\end{myproof}

\subsection{Proof of Proposition~\ref{thm:decision}} \label{sec:app:decision}
%!TEX root=../SIR-model-MS.tex
For a given $N$, let $\beta=1, \gamma = 1/2, f_0 = N, f_1 = c$.\footnote{We select $f_{1}$ in a way such that $f_1 c_{\infty} = N$, taking into account that $c_{\infty}$ is linear in $N$ in the deterministic SIR model.} We construct the following two instances: $\mathcal{M}_1 = (f_0, f_1, \beta, \gamma, N_1), \mathcal{M}_2 = (f_0, f_1, \beta, \gamma, N_2)$ where $N_{1} = N + N^{2/3}, N_{2} = N - N^{2/3}$. 

Intuitively, we need at least $m = N^{2/3}$ samples to distinguish between $\mathcal{M}_1$ and $\mathcal{M}_2$, which is precisely why a lower bound on the regret will be incurred. To be precise, let the policy $\pi_0$ be the policy that chooses to implement the drastic intervention at the beginning ($m=0$) since any intervention after $m=0$ will be worse. Then, we have
\begin{align*}
    \mathrm{cost}^{\pi_0}(\mathcal{M}_1) := N, \quad 
    \mathrm{cost}^{\pi_0}(\mathcal{M}_2) := N.
\end{align*}
Let the policy $\pi_1$ be the policy that does not implement the intervention at all. By our construction, we have
\begin{align*}
    \mathrm{cost}^{\pi_1}(\mathcal{M}_1) := N + N^{2/3}, \quad 
    \mathrm{cost}^{\pi_1}(\mathcal{M}_2) := N - N^{2/3}.
\end{align*}
The optimal cost for these two problem instances is given by
\begin{align*}
    \mathrm{cost}^{*}(\mathcal{M}_1) &:= N\\
    \mathrm{cost}^{*}(\mathcal{M}_2) &:= N - N^{2/3}.
\end{align*}

On the other hand, for any policy $\pi$, consider the probability of choosing to implement the drastic intervention given $m=N^{2/3}$ observations. Let
\begin{align*}
    p_1 := \text{Prob}(\pi(O_{m}) = \text{using drastic intervention}), \quad O_{m} \sim \mathcal{M}_1\\
    p_2 := \text{Prob}(\pi(O_{m}) = \text{using drastic intervention}), \quad O_{m} \sim \mathcal{M}_2.
\end{align*}
It is easy to verify that $|p_{1}-p_{2}| \leq D_{\mathrm{TV}}(O_{m}^{1}, O_{m}^2)$, where $O_{m}^{1} := O_{m}$ with $O_{m} \sim \mathcal{M}_{1}$ and $O_{m}^{2} := O_{m}$ with $O_{m} \sim \mathcal{M}_{2}$ and $D_{\mathrm{TV}}$ is the total variation distance. By Pinsker's inequality, $$
D_{\mathrm{TV}}(O_{m}^{1}, O_{m}^2)^2 \leq \frac{1}{2} D_{\mathrm{KL}}(O_{m}^1, O_{m}^2).
$$
Further, by the first-order approximation of KL divergence using Fisher information, we have $$
D_{\mathrm{KL}}(O_{m}^1, O_{m}^2) \lesssim (N_{1}-N_2)^2 J_{O_{m}}(N) = (N^{2/3})^2 \frac{m^3}{N^4} = \left(\frac{1}{N^{1/3}}\right)^2.
$$
This then implies 
\begin{align*}
    |p_1-p_2| = O(1/N^{1/3}).
\end{align*}
On the other hand, it is clear that in order to make the regret of $\mathcal{M}_1$ and $\mathcal{M}_2$ both less than $o(N^{2/3})$, $p_1$ must be close to $1$ and $p_{2}$ must be close to $0$. However, this contradicts the fact that $|p_1-p_2| = O(1/N^{1/3})$. Therefore, for any policy $\pi$, we have: 
$$
    \sup_{\mathcal{M} \in \{\mathcal{M}_1, \mathcal{M}_2\}}{\rm{regret}}^{\pi}(\mathcal{M}) = \Omega(N^{2/3}).
    $$
This completes the proof.

\subsection{Proof of Proposition~\ref{thm:sero}} \label{sec:app:sero}
%!TEX root=../SIR-model-MS.tex

Note that, conditioned on $N$, $O_{m}$ and $\tilde{O}_{m}$ are independent. Thus,
\begin{align*}
    J_{O_{m}\cup \tilde{O}_{m}}(N) = J_{O_{m}}(N) + J_{\tilde{O}_{m}}(N).
\end{align*}
Note that $J_{O_{m}}(N) = \Theta\left(\frac{m^{3}}{N^4}\right)$ has been calculated in the main theorem. It is sufficient to consider $J_{\tilde{O}_{m}}(N)$, which is
\begin{align*}
    J_{\tilde{O}_{m}}(N) = K J_{\mathrm{Ber}(\kappa_{m})}(N)
\end{align*}
since $X_{k}$ are independent from each other. Note that for any function $\eta(N)$, we have
\begin{align*}
    J_{\mathrm{Ber}(\eta)}(\eta(N)) 
    &= \frac{1}{\eta(1-\eta)}\left(\frac{d\eta(N)}{dN}\right)^2.
\end{align*}
Using $\eta(N) = E[C_{m}]/N$, we have
\begin{align*}
    J_{\mathrm{Ber}(\eta)}(\eta(N)) 
    &= \frac{N^2}{E[C_{m}](N-E[C_{m}])}\frac{E[C_{m}]^2}{N^4}\\
    &= \frac{E[C_{m}]}{N^2(N-E[C_{m}])}. 
\end{align*}
Note that $E[C_{m}] = \Theta(m)$ and $m = o(N).$ Therefore, 
\begin{align*}
    J_{\tilde{O}_{m}}(N) = K\cdot J_{\mathrm{Ber}(\kappa_{m})}(N)  =  \Theta\left(\frac{K m}{N^3}\right)
\end{align*}
which completes the proof. 

% --------------------------
% Peak sufficient condition
% --------------------------
% \section{Sufficient Condition for \texorpdfstring{$P[t]$}{P[t]}} \label{sec:app:peak_condition}
% \input{appendix_sections/peak_condition.tex}

% \section{Covariate Details for Practical SIR Model}
%!TEX root=../SIR-model-MS.tex
  \section{Datasets}
\label{sec:datasets}

Here we provide details on the datasets used in Section~\ref{sec:experiments}.

\subsection{Amazon product reviews}

For the Bass model, we use the Amazon product dataset of \cite{ni2019justifying}, which contains product reviews for Amazon products over more than twenty years. We take these reviews as a proxy for sales. Products in Amazon's electronics category typically have review trajectories well-approximated by the Bass model, marked by slow initial adoption and a long tail of sales towards the end of the product lifecycle -- see Figure~\ref{fig:amazon-trajectories} for examples of such trajectories. For our experiments, we randomly selected 100 products with over four years of reviews, and over 100 reviews by the fourth year. Review counts are taken at a weekly granularity. Here we use $N_{\max} = 1e5$ -- an order of magnitude larger than any of the true product sales numbers in the dataset.

\begin{figure}[htbp]
  \centering
\includegraphics[width=0.5\linewidth]{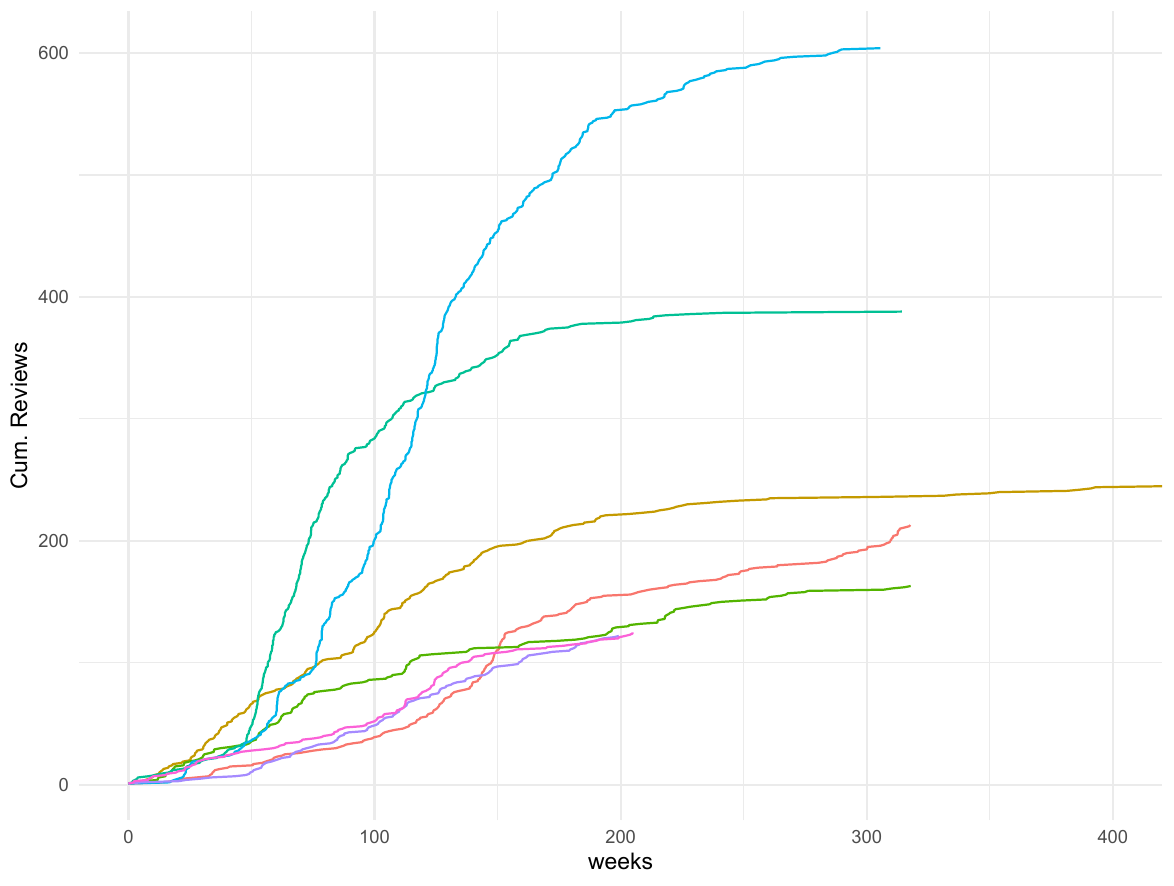}
  \caption{Cumulative weekly product reviews for randomly selected products from our subset of the Amazon dataset.}
  \label{fig:amazon-trajectories}
\end{figure}

\subsection{CDC ILINet influenza database}

For the SIR model, we use the CDC's ILINet database of patient visits for flu-like illnesses in the United States, broken down by Department of Health and Human Services region. Each instance in the dataset consists of weekly patient visits in a given region, over the course of one year. Each year starts in September, at the low point of the flu season. We use data from 2010 through 2019 for each of 10 regions, for 100 instances total. As the dataset only includes cumulative infections $C_{i}[t]$, rather than observations of infection and recoveries $I_{i}[t], R_{i}[t]$, we simulate these based on the dynamics \eqref{eq:discrete-dynamics}.

Here, we take $\gamma = 0.24$ as in \cite{chowell2008seasonal}, and $a$ is assumed to be 0. We take $N_{\max}$ to be the total patient population (including for non-flu illnesses) in the dataset.

\subsection{COVID-19 Datasets}
\label{sec:app:covariates}

For observed COVID-19 cases, we use publicly available case data from the ongoing COVID-19 epidemic provided by \cite{dong2020interactive}. We aggregate data into sub-state regions, corresponding broadly to public health service areas. The median state has seven regions. Here we take $\gamma = 1/4$.

The dataset contains static demographic covariates and time-varying mobility features that affect the disease transmission rate. The dynamic covariates proxy mobility by estimating the daily fraction of people staying at home relative to a region-specific benchmark of activity in early March before social distancing measures were put in place. We also include a regional binary indicator of the days when the fraction of people staying home exceeds the benchmark by 0.2 or more.

These data are provided by Safegraph, a data company that aggregates anonymized location data from numerous applications in order to provide insights about physical places. To enhance privacy, SafeGraph excludes census block group information if fewer than five devices visited an establishment in a month from a given census block group. Documentation can be found at \cite{Safegraph2020}.

The static covariates capture standard demographic features of a region that influence variation in infection rates. These features fall into several categories:

\begin{itemize}
	\item Fraction of individuals that live in close proximity or provide personal care to relatives in other generations. These covariates are reported by age group by state from survey responses conducted by \cite{UMichHRS}.
	\item Family size from U.S. Census data, aggregated and cleaned by \cite{Claritas}.
	\item Fraction of the population living in group quarters, including colleges, group homes, military quarters, and nursing homes (U.S. Census via \cite{Claritas}).
	\item Population-weighted urban status (US Census via \cite{Claritas})
	\item Prevalence of comorbidities, such as cardiovascular disease and hypertension (\cite{CDCAtlas})
	\item Measures of social vulnerability and poverty (U.S. Census via \cite{Claritas}; \cite{CDC_SVI})
	\item Age, race and occupation distributions (U.S. Census via \cite{Claritas})
\end{itemize}

\section{Detailed description of the COVID-19 model}
\label{sec:prior}

\subsection{Approximating the arrival process with latent state}
Recall the stochastic SIR process, ${ (S(t), I(t), R(t)) : t \geq 0 }$, a multi-variate counting process determined by parameters $(N, \beta, \gamma )$. We now allow $\beta$ to be time-varying, yielding a counting process with jumps $C_k - C_{k-1} \sim \mathrm{Bern}\left\{ \beta S_{k-1} /( \beta_k S_{k-1} + \gamma N I(t))\right\}$.

We obtain discrete-time diffusion processes, $\{(S_i[t], I_i[t], R_i[t]): t \in \mathbb{N}\}$ for instances $i \in \mathcal{I}$ by considering the Euler-approximation to the stochastic diffusion process \eqref{eq:deftk} (e.g. \cite{jacod2005approximate}). Specifically, let $\Delta I[t] = I[t] - I[t - 1]$, and define $\Delta S[t]$ and $\Delta R[t]$ analogously. A discrete-time approximation to the SIR process is then given by:

\begin{equation}
  \label{eq:discrete-diffusion}
\begin{aligned}
  \Delta S_{i}[t+1]  &=  - \beta_{i}[t] ({S_{i}[t]}/{N_{i}}) I_{i}[t]  + \nu^S_{i,t}\\
  \Delta I_{i}[t+1] &= \beta_{i}[t] ({S_{i}[t]}/{ N_{i}}) I_{i}[t] - \gamma I_{i}[t]  + \nu^I_{i,t}\\
  \Delta R_{i}[t+1] &= \gamma I_{i}[t] + \nu^R_{i,t}
\end{aligned}
\end{equation}

where $\{ \nu^S_{i,t} \}, \{ \nu^I_{i,t} \}, \{ \nu^R_{i,t} \}$ are appropriately defined martingale difference sequences.

In the real world, the SIR model is a latent process -- we never directly observe any of the state variables $S_{i}[t], I_{i}[t], R_{i}[t]$. Instead, we observe $C_{i}[t] = I_{i}[t] + R_{i}[t] = N_i - S_i[t]$.  The MLE problem for parameters $(N, \beta)$ is simply $\max_{(\beta, N)} \sum_{i,t} \log \mathbb{P}\left( C_{i}[t] |  \beta, N\right)$.

This is a difficult non-linear filtering problem (and an interesting direction for research). We therefore consider an approximation: Denote by $\{(s_i[t], i_i[t], r_i[t]): t \in \mathbb{N}\}$ the deterministic process obtained by ignoring the martingale difference terms in the definition of the discrete time SIR process. We consider the approximation $C_i[t] =  N_i - S_i[t] \sim (N_i - s_i[t]) \omega_{i}[t]$, where $\omega_{i}[t]$ is log-normally distributed with mean $1$ and variance $\exp(\sigma^2) - 1$.

Under this approximation, we have the log likelihood function

\begin{equation}
  \label{eq:cum-llh}
  \log p(C_{i}[t] | N, \beta) = \left( \log C_i[t] - \log \left(N_i - s_i[t]\right) \right)^2
\end{equation}

\subsection{Two-Stage Estimation of the SIR model}
% \paragraph{Overall Learning Algorithm (`Two-Stage')} So motivated,  we now define our overall learning algorithm `Two-Stage' as follows.

We parameterize our estimates of $N$ as $\hat{N}_i(\phi, \delta) = \exp(\phi ^\top Z_{i} + \delta_{i}) P_{i}$, where $Z_{i}$ are non-time-varying, region-specific covariates, $P_i$ is the population of region $i$, $\phi$ is a vector of fixed effects, and $\delta_{i} \sim \mathcal{N}(0, \sigma^{2}_{\delta})$ are region-specific random effects.

Demographic and mobility factors also influence the reproduction rate of the disease. To model these effects, we estimate $\beta_{i}[t]$ as a mixed effects model incorporating covariates  $\beta_{i}[t] = \exp( X_i[t] ^\top \theta) + \epsilon_{i}$, where $\theta$ is a vector of fixed effects, and $\epsilon_{i} \sim \mathcal{N}(0, \sigma^{2}_{\epsilon})$ is a vector of random effects.

Given observations up to time $T$, we then estimate the model parameters $(\theta, \phi, \delta, \epsilon)$ in two stages:

\begin{enumerate}
  \item Estimate the peak parameters $\hat{\phi}, \hat{\delta}$ via MLE, for the regions $i \in Q[t]$:
    $$\hat{\phi}, \hat{\delta} = \arg \max_{\phi, \delta} \left\{\max_{\theta, \epsilon} \left\{ \sum_{i \in Q[t]} \sum_{t \in [T]}\log p \left(C_{i}[t] \,\big\vert\, \beta_{i}(\theta, \epsilon), \hat{N}_{i}(\phi, \delta) \right)  + \log p(\epsilon , \delta) \right\} \right\}$$
    where $p$ is the likelihood defined in \eqref{eq:cum-llh}. We let $\hat{\delta}_{i} = 0$ for  $i \notin Q[t]$.
  \item Estimate the remaining parameters over all regions $i \in \mathcal{I}$:
    \begin{equation}
      \label{eq:llh2}
      \hat{\theta}, \hat{\epsilon} = \arg \max_{\theta, \epsilon} \left\{ \sum_{i \in \mathcal{I}} \sum_{t\in [T]} \log p \left(C_{i}[t] \,\big\vert\, \beta_{i}(\theta, \epsilon), \hat{N}_{i}(\hat{\phi}, \hat{\delta}), \right) + \log p(\epsilon, \delta) \right\}
    \end{equation}
\end{enumerate}

We note that \eqref{eq:llh2} is differentiable with respect to the parameters ($\theta, \epsilon, \phi, \delta$), and we solve it (or a weighted version) using Adam \citep{kingma2014adam}.\footnote{Adam was run for 20k iterations, with learning rate tuned over a coarse grid. A weighted version of the loss function in \eqref{eq:llh2} with weights for $(i,t)$th observation set to $C_i[t]$ worked well.}

To identify the set $Q[t]$ of regions for which the variance of $\hat{N}$ may be small, we simply look for regions that have passed their peak rate of new infections. Concretely, we define $Q[t]$ as:
\begin{equation}
\label{eq:peaks}
Q[t] = \{ i \in \mathcal{I}:  C_{i}[t] - C_{i}[t - 1] \leq \gamma_1 \max_{\tau \leq t} \left(C_{i}[\tau] - C_{i}[\tau - 1]\right) \},
\end{equation}
where $\gamma_1 \in (0, 1)$ is a hyperparameter.

\subsection{Performance relative to other models}

To contextualize the quality of the Two-Stage model, we compare our analyzed models to the widely used IHME model \cite{ihme}. We note that there exist comparable models that may serve as stronger baselines; we include these results merely to demonstrate that the Two-Stage model yields high-quality predictions, comparable to widely-cited models in the literature.

Figure~\ref{fig:ihme} compares state-level\footnote{Due to IHME only providing state-level predictions. Additionally IHME only offers deaths predictions for these vintages; we show WMAPE on deaths for IHME and WMAPE on infections for MLE and Two-Stage.} WMAPE for MLE, Two-Stage and IHME models, for vintages stretching back 28 days. The IHME model up to this date is, in effect, an SI model with carefully tuned parameters. We report published IHME forecasts; 10 vintages of that model were reported between April 21 and May 21.
{\it Two Stage} dominates IHME across all model vintages.

\begin{figure}[H]
  \centering
  \includegraphics[scale=0.45]{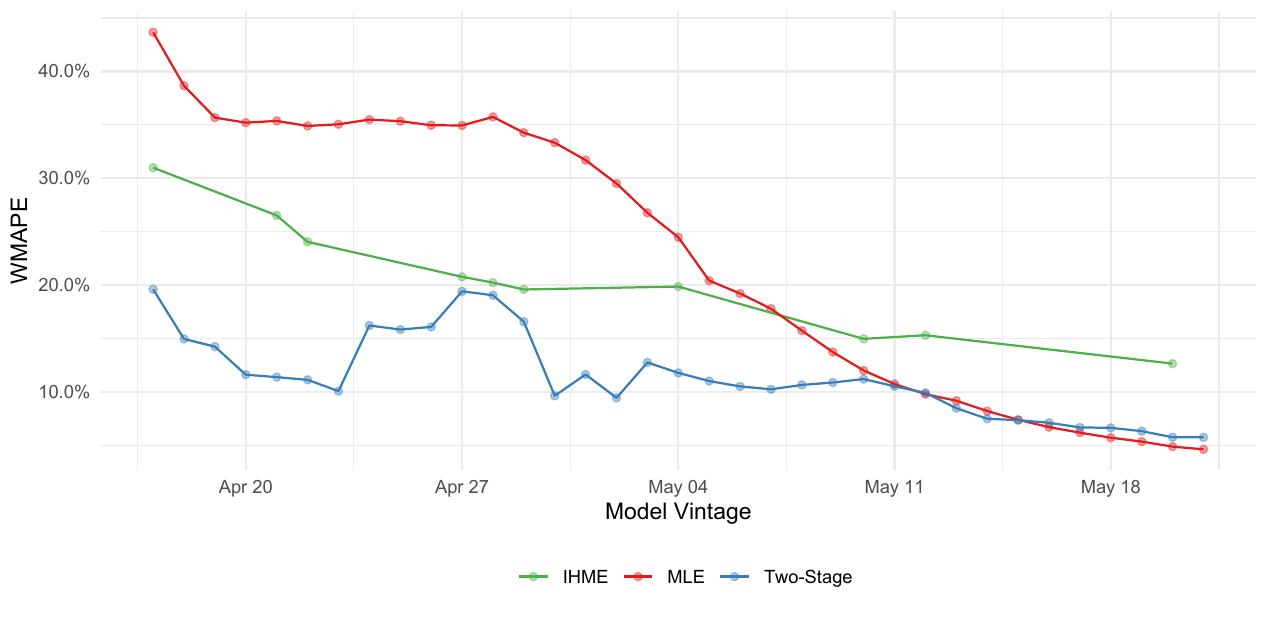}
  \caption{WMAPE for predicting state-level cumulative cases on May 21, 2020, comparing MLE and the Two-Stage approach against IHME.}
  \label{fig:ihme}
\end{figure}

\end{document}